\documentclass[10pt]{article} 

\usepackage[utf8]{inputenc} 

\usepackage{amsthm, amsmath, amssymb, mathrsfs, mathtools}
\usepackage[dvipsnames]{xcolor}
\usepackage[colorlinks]{hyperref}
\usepackage{tensor}
\usepackage[affil-it]{authblk}

\usepackage{caption}
\usepackage{subcaption}

\usepackage{enumitem}

\usepackage{bm}

\usepackage{tensor}

\usepackage{empheq}

\usepackage{faktor}

\usepackage{float}

\usepackage{cleveref}

\newtheoremstyle{break}
  {\topsep}{\topsep}%
  {\itshape}{}%
  {\bfseries}{}%
  {\newline}{}%
\theoremstyle{break}

\newtheorem*{theorem*}{Theorem}
\newtheorem{theorem}{Theorem}[section]
\newtheorem{proposition}[theorem]{Proposition}
\newtheorem{corollary}[theorem]{Corollary}
\newtheorem{remark}[theorem]{Remark}
\newtheorem{lemma}[theorem]{Lemma}
\newtheorem{definition}[theorem]{Definition}

\numberwithin{equation}{section}

\DeclarePairedDelimiter\autobracket{(}{)}

\newcommand{\br}[1]{\autobracket*{#1}}

\usepackage[normalem]{ulem}

\usepackage[top=53pt, bottom=53pt, left=80pt, right=80pt]{geometry} 
\usepackage{graphicx} 

\usepackage{fancyhdr} 
\hypersetup {
pdfpagemode = {UseNone},
linkcolor = NavyBlue,
citecolor = Green,
}

\title{Mass inflation from rough initial data 
 for the spherically symmetric Einstein-Maxwell-scalar field system with $\Lambda$}

\author{
Flavio Rossetti\thanks{\textit{E-mail address}: \texttt{flavio.rossetti@gssi.it}}
}

\affil{Gran Sasso Science Institute, Viale Francesco Crispi 7, L'Aquila (AQ), 67100, Italy}

\newcommand{\Ric}{\text{Ric}}
\newcommand{\R}{\text{R}}

\date{} 

\begin{document}
\maketitle

\begin{abstract}
Recent rigorous results on black hole interiors, combined with numerical experiments on black hole exteriors, clearly  suggest that the strong cosmic censorship conjecture fails in its most fundamental, i.e.\ weak,  formulation: violations are expected for a class $\mathcal{B}_{e, M, \Lambda}$ of spherically symmetric charged black holes in the presence of a positive cosmological constant $\Lambda$ near extremality. 

These results require sufficiently regular solutions. 
Conversely, when non-smooth, finite-energy initial data are prescribed for linear waves propagating on a fixed black hole background $\mathcal{B}_{e, M, \Lambda}$ belonging to the aforementioned family, it was shown  \cite{DafYak18} that the local energy of these linear waves blows up at the Cauchy horizon, hence hinting that non-smooth initial data may suppress the possible violations of the $H^1$ formulation of strong cosmic censorship.

In line with this intuition,  we prove that rough initial data  can also trigger an instability at the Cauchy horizon  in the non-linear setting, via mass inflation.
In particular, we analyse a characteristic initial value problem for the spherically symmetric Einstein-Maxwell-real scalar field system describing the interior of a black hole. Our results show that, when prescribing:
\begin{itemize}
    \item  initial data asymptotically approaching those of a sub-extremal Reissner-Nordstr{\"o}m-de Sitter solution $\mathcal{B}_{e, M, \Lambda}$, via an exponential Price law bound, and
    \item initial data belonging to $W^{1, 2}(\underline{C}_{v_0}; \mathbb{R})\setminus W^{1, q}(\underline{C}_{v_0}; \mathbb{R})$, for every $q > 2$, along the initial ingoing compact segment $\underline{C}_{v_0}$, 
\end{itemize}
then the Hawking mass (that, in this low-regularity setting, is an $H^1$ geometric quantity) diverges at the Cauchy horizon of the (non-smooth) black hole solution we construct, for every parameter choice $(e, M, \Lambda)$ of the reference black hole. 

In this larger family of configurations, we prove that the  smooth data suggesting violations of strong cosmic censorship are non-generic in a ``positive co-dimension'' sense, conditionally to the validity of the expected Price law bounds. 

Moreover, we illustrate the transition between smooth and rough initial data. If more regularity is assumed, then possible violations of strong cosmic censorship reemerge. Namely,  given $p > 2$, we further consider,  along the initial ingoing segment, data belonging to $W^{1, p} (\underline{C}_{v_0}; \mathbb{R}) \setminus W^{1, q}(\underline{C}_{v_0}; \mathbb{R})$, for every $q > p$. We find that these non-smooth initial data are still sufficiently regular to yield $H^1$ extensions beyond the Cauchy horizon, for choices of $(e, M, \Lambda)$ corresponding to  near-extremal black holes. In the large-$p$ limit, the set of $(e, M, \Lambda)$ yielding $H^1$ extensions  approaches, in a suitable sense, the set of black hole parameters giving rise to $H^1$ extensions in the smooth setting.
\end{abstract}

\tableofcontents

\section{Introduction} \label{section:intro}

The strong cosmic censorship conjecture, first proposed by Penrose \cite{PenroseBattelle}, is a  statement 
 of global uniqueness of solutions to the initial value problem for the Einstein equations. When global uniqueness fails, spacetime solutions can be extended in a non-unique way beyond the maximal globally hyperbolic development of the prescribed initial data, hence leading to a breakdown of Laplacian determinism in general relativity (we refer to \cite{ChristodoulouIVP, Dafermos_2003} for related analyses of the conjecture). 
 
Remarkably, recent results (see already section \ref{section:related_works} for an overview of past works) in the mathematics and physics literature  suggest a failure of  strong cosmic censorship  for charged black holes in the presence of a  positive cosmological constant, at least when considering smooth initial data.

However, the strong cosmic censorship conjecture can be formulated, by definition, in any space of functions in which the initial value problem for the Einstein equations is well-posed. By working in a non-smooth setting, the linear analysis in \cite{DafYak18} (see also the numerical results in \cite{Reall+Rough}) showed that low-regularity initial data lead to solutions to the linear wave equation, 
 on fixed backgrounds given by charged cosmological black holes, having unbounded $H^1$ norm. 
 By interpreting the linear wave equation as a proxy for the Einstein equations, this linear blow-up hints that global uniqueness may hold for metric solutions evolved from rough initial data, in stark contrast to the scenario emerging in the smooth setting.
 
 In particular, \cite{DafYak18} proved that, if non-smooth  initial data of Sobolev regularity  are prescribed along a suitable initial hypersurface, then solutions to the linear wave equation 
\begin{equation} \label{eqn_0}
\square_g \phi = 0
\end{equation}
 with respect to the metric\footnote{Here $\sigma_{S^2}$ is the standard metric on $S^2$.}
\begin{equation} \label{RNdS_metric}
g = -\br{1-\frac{2M}{r} + \frac{e^2}{r^2} - \frac{\Lambda}{3}r^2}dt^2+\br{1-\frac{2M}{r} + \frac{e^2}{r^2} - \frac{\Lambda}{3}r^2}^{-1}dr^2 + r^2 \sigma_{S^2},
\end{equation}
 blow up in $H^1$ (i.e.\ energy) norm at the Cauchy horizon of the sub-extremal Reissner-Nordstr{\"o}m-de Sitter spacetimes described by $g$, for every choice of the spacetime mass $M > 0$, spacetime charge $e \in \mathbb{R} \setminus \{0\}$ and cosmological constant $\Lambda > 0$.
 Note that:
 \begin{itemize}
     \item  the Reissner-Nordstr{\"o}m-de Sitter family is the unique family of spherically symmetric, asymptotically de Sitter (in a suitable sense) solutions to the Einstein-Maxwell equations with $\Lambda > 0$, and
     \item $H^1$ regularity is the minimum requirement needed for a metric to solve the Einstein equations in a weak sense.\footnote{See already section \ref{section:rough_LWP} for the definition of  weak solutions that we employ.}
 \end{itemize}
Seeing \eqref{eqn_0} as a \textbf{linear} surrogate of the quasi-linear system of the Einstein equations, it is natural to contemplate whether an $H^1$ instability result should similarly hold in any non-smooth, non-linear setting describing perturbations of Reissner-Nordstr{\"o}m-de Sitter black hole interiors. Far from being a purely technical remark, results in this direction would provide further insights to support  resolutions of the strong cosmic censorship conjecture in a non-smooth setting.

 In general, non-linear instability results of this kind cannot be inferred from the linear case alone. Indeed, it is a challenging task to determine a priori how non-linear terms could affect linear instability results. A paramount example is given by the long-standing problem of quantifying and describing the non-linear consequences of the blueshift effect, which is one of the main mechanisms driving the linear instabilities at Cauchy horizon of the black holes that we consider in this article (see already section \ref{section:related_works}).

In the present work, we initiate the analysis of the  \textbf{non-linear}   instability and stability of Cauchy horizons of cosmological black holes under the assumption of rough initial data.
More precisely, we investigate the spherically symmetric Einstein-Maxwell-real scalar field system
\begin{equation} \label{main_system}
\begin{cases}
\Ric(g)-\frac{1}{2}\R(g) \,g + \Lambda g = 2 T, &\textit{\normalfont (Einstein's equations)}\\[0.3em]
dF=0, \quad d \star F = 0, &\textit{\normalfont (Maxwell's equations)}\\[0.3em]
\square_g \phi = 0, &\textit{\normalfont (Wave equation)}
\end{cases}
\end{equation}
with $\Lambda > 0$, electromagnetic field tensor $F$ and energy-momentum tensor
\[
\tensor{T}{_\mu_\nu} = \tensor{\partial}{_\mu} \phi \tensor{\partial}{_\nu} \phi - \frac12 \tensor{g}{_\mu_\nu} \tensor{\partial}{_\alpha} \phi \tensor{\partial}{^\alpha} \phi + \tensor{F}{_\mu_\alpha} \tensor{F}{_\nu^\alpha} - \frac14 \tensor{F}{_\alpha_\beta} \tensor{F}{^\alpha^\beta} \tensor{g}{_\mu_\nu}. 
\]
We focus on the spherically symmetric setting as a proxy for the more general case of the vacuum Einstein equations without symmetry assumptions. A detailed description of the Einstein-Maxwell-scalar field model in the context of strong cosmic censorship can be found, e.g.\ in \cite{Dafermos_2003}.

In the present work, non-trapped initial data are prescribed along two transversal characteristic hypersurfaces, such that the outgoing initial hypersurface is identified with the event horizon of a black hole asymptotically approaching a reference Reissner-Nordstr{\"o}m-de Sitter solution having parameters $e, M, \Lambda$. This condition is described by an exponential Price law upper bound, i.e.\  the tangential derivative of the scalar field along the event horizon $C_0 = \mathcal{H}^+$ satisfies
\begin{equation} \label{price_law_intro}
    |\partial_v \phi|_{|C_0} \le C e^{-s v},
\end{equation}
for some $C, s > 0$ and where $v$ is an Eddington-Finkelstein type of coordinate. The exponential character of the scalar field decay is a consequence of the presence of a positive cosmological constant.

The null coordinate chart $(u, v)$, where $u$ is a Kruskal coordinate, will be properly introduced in section \ref{subsection:defs}. In the following, we will represent the initial characteristic hypersurfaces as $C_0 = \mathcal{H}^+= \{0\} \times [0, +\infty)$ and $\underline{C}_0 = [0, U] \times \{0\}$, for some $U > 0$. 

 We recall that the results in \cite{CGNS4} and \cite{Rossetti} imply that smooth initial data prescribed for  system \eqref{main_system} can lead to black hole  solutions presenting mass inflation\footnote{That is, the geometric quantity $\varpi$ that generalizes the black hole mass diverges.} along the Cauchy horizon $\mathcal{CH}^+$ determined by the initial hypersurfaces, or, alternatively, to solutions admitting metric extensions of $H^1$ regularity beyond $\mathcal{CH}^+$. The latter scenario occurs when the reference black hole is near-extremal, and is interpreted as a violation of the $H^1$ formulation of strong cosmic censorship for system \eqref{main_system}, conditionally to the validity of the assumed Price law. See also fig.\ \ref{fig:smooth_results}. We emphasize that  $H^1$  regularity is the minimum requirement for the metric to solve the Einstein equations in a weak sense, hence $H^1$ extensions beyond Cauchy horizons constitute violations of Laplacian determinism at a fundamental level of regularity.
\begin{figure}[h]
    \centering
    \includegraphics[width=0.5\linewidth]{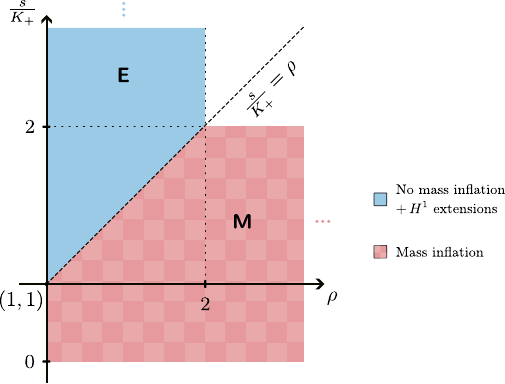}
    \caption{Known results for system \eqref{main_system} when \textbf{smooth} initial data are prescribed (see \cite{CGNS4} and \cite{Rossetti}). Here $s$ encodes the decay dictated by Price's law upper bound (see \eqref{price_law_intro}) and $\rho = K_{-}/K_+$ is the ratio of the surface gravities of the reference black hole. The mass inflation results require  exponential Price law upper and lower bounds, whereas the stability result requires an upper bound only. No result is available in the region $\{\rho \ge 2\} \cap \{ s \ge 2K_+ \}$ and in $\{s = \rho K_+\}$. See already remark \ref{remark:comparison}, where this point is discussed.}
    \label{fig:smooth_results}
\end{figure}

In the current paper we show that, for \textbf{every choice} of smooth initial data  for \eqref{main_system} among those leading to $H^1$ extensions, if we perturb such data in a larger class of functions for which less regularity is assumed, then the perturbed data yield mass inflation along the Cauchy horizon (i.e.\ an $H^1$ geometric quantity diverges).  In other terms, we show that \textbf{the smooth data associated to the failure of global uniqueness for system \eqref{main_system} should be regarded as \textit{exceptional} in a larger moduli space of well-posed initial data}. In the following, we present a more precise description of this  result, together with a summary of the main non-linear instability and stability results that we obtain.

Given $p \ge 2$, let us  consider the following space of initial data\footnote{We perturb the smooth initial data for $\partial_u \phi$ and $\partial_v \phi$ in order to construct rough initial data. Since the remaining PDE quantities are assumed to be smooth and unperturbed, here we focus only on the moduli space of initial data for these two functions.} for $\partial_u \phi$ and $\partial_v \phi$:
    \[
    \mathcal{R}^p_{\text{i.d.}} \coloneqq L^p_u([0, U]) \times C^0_v([0, +\infty))
    \]
    and its subset 
    \[
    \mathcal{R}^{p, \text{FE}}_{\text{i.d.}} \coloneqq \Big ( L^p_u([0, U])\setminus \bigcup_{q > p} L^q_u([0, U]) \Big ) \times  C^0_v([0, +\infty))    \subset \mathcal{R}^p_{\text{i.d.}}.
    \]
    When $p=2$, we denote these two sets by the names \textit{rough data} and \textit{finite-energy data}, respectively. For the sake of simplicity, we keep an analogous notation  for $p > 2$ as well.
    Moreover, consider the space of smooth initial data 
    \[
    \mathcal{S}_{\text{i.d.}} \coloneqq C^0_u([0, U]) \times C^0_v([0, +\infty)).
    \]
    Note that $\mathcal{S}_{\text{i.d.}} \subset \mathcal{R}^p_{\text{i.d.}}$ but $\mathcal{S}_{\text{i.d.}} \cap \mathcal{R}^{p, \text{FE}}_{\text{i.d.}} = \emptyset$ for any $p \ge 2$.

We prove the following non-linear \textbf{instability and stability results} (see also section \ref{section:main_results} for a more precise statement of our theorem and see fig.\ \ref{fig:moduli_space} for a depiction of the following statements):
\begin{enumerate}
    \item \textbf{Mass inflation for finite-energy initial data:} Let $p=2$ and assume that
    \begin{itemize}
        \item the smooth initial data in $\mathcal{S}_{\text{i.d.}}$ satisfy, along $C_0 = \{0\} \times [0, +\infty)$,  the pointwise exponential Price law upper bound \eqref{price_law_intro}, and
        \item the finite-energy initial data in $\mathcal{R}_{\text{i.d.}}^{2, \text{FE}}$ satisfy, along $C_0 = \{0\} \times [0, +\infty)$,  a pointwise exponential Price law upper and lower bound.\footnote{The Price law lower bound assumption is expected to hold along the event horizon in the near-extremal scenarios yielding $H^1$ extensions. In fact, the near-extremal modes determining the decay of real scalar fields satisfying the wave equation in the exterior of Reissner-Nordstr{\"o}m-de Sitter do not present any oscillatory behaviour. See also the numerical work \cite{QNMandSCCC} and the recent rigorous proof  in \cite{HintzQNM}. }
    \end{itemize}
    Moreover, for every $Q=((\partial_u \phi)_0, (\partial_v \phi)_0) \in \mathcal{S}_{\text{i.d}}$ yielding $H^1$ extensions for \eqref{main_system} we define a line
    \[
    L_Q \coloneqq \{ Q + \omega \bm{\mathcal{Q}} \colon \omega \in \mathbb{R}, \, \bm{\mathcal{Q}}=\bm{\mathcal{Q}}(Q) \in \mathcal{R}_{\text{i.d.}}^{2, FE}\}, 
    \]
    where $L_Q \setminus \{Q\}$ is contained in $\mathcal{R}_{\text{i.d.}}^{2, FE}$. Then, every choice of initial data in $L_Q$ (except for $Q$ itself) yields mass inflation in the non-smooth black hole solution we construct. The choice of $\bm{\mathcal{Q}}=\bm{\mathcal{Q}}(Q)$ in the definition of $L_Q$ is highly non-unique, i.e.\ we are able to construct infinitely many lines passing through $Q$. Given two different points $Q, R \in \mathcal{S}_{\text{i.d.}}$, the lines $L_Q$ and $L_R$ do not intersect. 
    \item \textbf{Mass inflation and $\bm{H^1}$ extensions between finite-energy and smooth initial data:} Let $p > 2$ and assume the exponential Price law bounds of the previous point.\footnote{Strictly speaking, we require an exponential Price law lower bound only when we show mass inflation. This assumption can be dropped when constructing $H^1$ extensions.} Moreover, for every $Q=((\partial_u \phi)_0, (\partial_v \phi)_0) \in \mathcal{S}_{\text{i.d}}$ yielding $H^1$ extensions, we consider a  line 
    \[
    L_Q \coloneqq \{ Q + \omega \bm{\mathcal{Q}}\colon \omega \in \mathbb{R}, \, \bm{\mathcal{Q}}=\bm{\mathcal{Q}}(Q) \in \mathcal{R}_{ \text{i.d.}}^{p, FE}\}, 
    \]
    where $L_Q \setminus \{ Q \}$ is contained in $\mathcal{R}_{ \text{i.d.}}^{p, FE}$. Then, two mutually  exclusive scenarios occur:\footnote{Except for a  choice of black hole parameters  corresponding to the boundary, in the $(s, \rho)$ parameter space, of the set of parameters yielding $H^1$ extensions. For this specific choice we cannot prove mass inflation nor the existence of $H^1$ extensions.} either every choice of initial data in $L_Q$ (except for $Q$ itself) yields mass inflation, or every choice of initial data in $L_Q$ gives $H^1$ extensions. In particular, the latter scenario occurs if enough regularity is assumed. Namely, for each $Q \in \mathcal{S}_{\text{i.d.}}$ yielding $H^1$ extensions in the smooth case, there exists $p > 2$ sufficiently large such that $L_Q \setminus \{Q\} \subset \mathcal{R}_{\text{i.d.}}^{p, \text{FE}}$ and such that every choice of initial data in $L_Q$ yields $H^1$ extensions. In the large-$p$ limit, we can construct lines of data yielding $H^1$ extensions for all near-extremal initial data giving $H^1$ extensions in the smooth setting. Moreover, any two distinct lines $L_Q$ and $L_R$ corresponding to different points $Q, R \in \mathcal{S}_{\text{i.d.}}$ do not intersect.
\end{enumerate}
\begin{figure}[h]
    \centering
    \includegraphics[width=\linewidth]{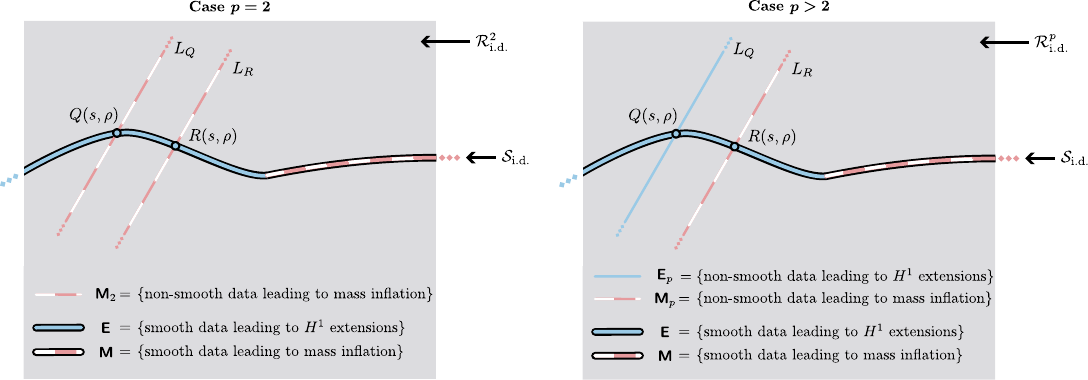}
    \caption{A schematic representation of the moduli space of initial data, illustrating the main results of this paper. When $p=2$, for each choice $Q$ of smooth initial data (prescribed in the black hole interior) leading to $H^1$ extensions for system \eqref{main_system}, we can find a line $L_Q$ of rough initial data that, except for the point $Q$, lead to mass inflation. When $p>2$, either we have lines of non-smooth data yielding mass inflation, or lines of non-smooth data yielding $H^1$ extensions. Lines corresponding to different initial data do not intersect. Here $\mathcal{R}^p_{\text{i.d.}}$ and $\mathcal{S}_{\text{i.d.}}$ are the spaces of non-smooth and smooth initial data, respectively, considered in the present paper. Each choice of smooth initial data can be associated with certain values of $s$ and $\rho= K_{-}/K_+$, see also fig.\ \ref{fig:smooth_results} for a representation of this parameter space. Different initial data may potentially be associated with the same values of $s$ and $\rho$.   We are assuming initial data that satisfy a Price law upper and lower bound along the event horizon, as expected from the near-extremal scenario for real scalar fields.}
    \label{fig:moduli_space}
\end{figure}
As a consequence of the above, we conclude that the smooth initial data for \eqref{main_system} yielding $H^1$ extensions beyond the Cauchy
horizon are non-generic (in a ``positive co-dimension'' sense) in a large set of finite-energy initial data. In the latter, well-posedness of the initial value problem holds (see already section \ref{section:rough_LWP}). We refer the reader to \cite{ChristodoulouIVP} for an analogous
discussion of \textit{exceptional} initial data in the context of the weak cosmic censorship conjecture.
 
  Furthermore, the rough initial data we consider are regular enough so that we are able to prove continuous stability results (even without assuming an exponential Price law lower bound):
\begin{enumerate}
    \setcounter{enumi}{2}
    \item \textbf{$\bm{C^0}$ extensions:} Given $p \ge 2$, every choice of initial data in $\mathcal{R}^p_{\text{i.d.}}$ yields a  black hole interior solution, endowed with a Cauchy horizon, solving the Einstein-Maxwell-scalar field system \eqref{main_system} in an integrated sense. The metric solution and the scalar field $\phi$ are continuously extendible up to the Cauchy horizon.
\end{enumerate}

The above results concern perturbations of smooth initial data leading to $H^1$ extensions (i.e.\ points belonging to $\mathsf{\textbf{E}}$ in the parameter space of fig.\ \ref{fig:smooth_results}). This is the most relevant case of the present paper, since such extensions suggest a negative resolution to the $H^1$ formulation of the strong cosmic censorship conjecture for system \eqref{main_system}, in the smooth setting. 

On the other hand, we also obtain additional mass inflation results when perturbing smooth initial data in the complement of the set $\mathsf{\textbf{E}}$ (including the region in which neither mass inflation nor $H^1$ extensions are proved). These results are presented in detail in section \ref{section:main_results}. We also stress that smooth initial data in  $\mathsf{\textbf{E}}$ are regarded as near-extremal initial data: the exponential Price law lower bound that we assume for $\partial_v \phi$ is expected to hold in this case (even though such lower bound is  not required to prove the existence of $H^1$ extensions, see \cite{Rossetti}). Conversely, smooth data for $\partial_v \phi$ in the complement of $\mathsf{\textbf{E}}$ are expected to potentially exhibit an oscillatory behaviour.\footnote{In the language of \cite{QNMandSCCC}, this is due to the presence of photon sphere modes dictating an oscillating behaviour for solutions to the linear wave equation on a fixed Reissner-Nordstr{\"o}m-de Sitter exterior.}

\subsection{The Einstein-Maxwell-scalar field system} \label{subsection:defs}

Motivated by our symmetry assumptions, we require that the scalar field $\phi$ is spherically symmetric and that every spacetime solution $(\mathcal{M}, g)$ to \eqref{main_system} is such that its Lorentzian metric can be expressed as
\begin{equation} \label{sph_symm_metric}
g = -\Omega^2(u, v)du dv + r^2(u, v) \sigma_{S^2},
\end{equation}
for some real-valued function $\Omega$, where $\sigma_{S^2}$ is the standard metric on $S^2$ and $(u, v)$ are null coordinates\footnote{Here $\partial_u$ is null, ingoing and future-pointing and $\partial_v$ is null, outgoing and future-pointing. The remaining gauge freedom in the $(u, v)$ chart will be fixed in the following sections.} in the $(1 + 1)-$dimensional manifold $\mathcal{M}/SO(3)$.  The function $r$ measures the area of the orbits of the $SO(3)$ action, i.e.\ if $\Pi \colon \mathcal{M} \to \mathcal{M}/SO(3)$ is the canonical projection, then:
\[
r(u, v) = \sqrt{\frac{\text{Area}(\Pi^{-1}(u, v))}{4\pi}}.
\]
For ease of reading, in our notation we will not distinguish between quantities defined on $\mathcal{M}$ and their push-forward through the canonical projection. 

Since the scalar field $\phi$ is real-valued, the Maxwell equations decouple and yield (see e.g.\ \cite{CGNS1} for  the next lines):
\[
F = -\frac{e \Omega^2(u, v)}{2 r^2(u, v)} du \wedge dv,
\]
where $e \in \mathbb{R} \setminus \{0\}$ is interpreted as the black hole charge.

In this coordinate system, the Einstein-Maxwell-real-scalar field system reads
\begin{align} 
\partial_u \partial_v r &= -\frac{\Omega^2}{4r} -\frac{\partial_u r \partial_v r}{r} + \frac{\Omega^2 e^2}{4 r^3} + \frac{\Omega^2 \Lambda r}{4}, \label{wave_r} \\
r\partial_u \partial_v \phi &= -\partial_u r \partial_v \phi - \partial_v r \partial_u \phi, \label{wave_phi} \\
\partial_v \partial_u \log \Omega^2 &= -2\partial_u \phi \partial_v\phi - \frac{\Omega^2 e^2} { r^4} + \frac{\Omega^2}{2 r^2} + \frac{2 \partial_u r \partial_v r}{r^2}, \label{wave_Omega} \\
\partial_u \br{\frac{\partial_u r}{\Omega^2}} &= - r \frac{(\partial_u \phi)^2}{\Omega^2}, \label{raych_u}\\
\partial_v \br{\frac{\partial_v r}{\Omega^2}} &= -r \frac{(\partial_v \phi)^2}{\Omega^2}, \label{raych_v}
\end{align}
namely it is composed by wave equations for $r, \phi$ and $\log \Omega^2$, respectively, and two Raychaudhuri equations. 
The unique set of (smooth and asymptotically de Sitter) solutions to the above system when $\phi \equiv 0$ is given by the Reissner-Nordstr{\"o}m-de Sitter family of metrics, characterized by a black hole mass parameter $M > 0$, black hole charge $e \in \mathbb{R} \setminus \{0\}$ and cosmological constant $\Lambda > 0$. We will take a member of this family as a \textit{reference} solution to compare our dynamical solution with.

We also define 
\begin{align*}
    \nu &\coloneqq \partial_u r, \\
    \lambda &\coloneqq \partial_v r
\end{align*}
and, furthermore:
\begin{align}
    \varpi &\coloneqq \frac{e^2}{2r} + \frac{r}{2} - \frac{\Lambda}{6}r^3 + \frac{2r}{\Omega^2} \nu \lambda, \label{def_varpi} \\
    \kappa &\coloneqq -\frac{\Omega^2}{4 \nu}, \label{def_kappa}\\
    \mu &\coloneqq \frac{2 \varpi}{r} - \frac{e^2}{r^2} + \frac{\Lambda}{3}r^2, \label{def_mu} \\
    K &\coloneqq \frac{\varpi}{r^2} - \frac{e^2}{r^3} - \frac{\Lambda}{3}r, \label{def_K}
\end{align}
whenever these quantities are well-defined, 
where $\varpi$ is the (renormalized) Hawking mass and $K$ generalizes the quantities
\begin{equation} \label{def_Kplus}
K_+ \coloneqq \frac{M}{r_+^2} - \frac{e^2}{r_+^3} - \frac{\Lambda}{3}r_+ > 0, 
\end{equation}
and 
\begin{equation} \label{def_Kminus}
-K_{-} \coloneqq  \frac{M}{r_{-}^2} - \frac{e^2}{r_{-}^3} - \frac{\Lambda}{3}r_{-} < 0, 
\end{equation}
that are the surface gravities of the event horizon and of the Cauchy horizon, respectively, of the reference 
 Reissner-Nordstr{\"o}m-de Sitter black hole. The values $r_{-}$ and $r_+$ are radii of the two Killing horizons.

Local well-posedness results for \eqref{wave_r}--\eqref{raych_v}, in a low-regularity setting, will be discussed in section \ref{section:rough_LWP}. The main equations stemming from the integral formulation of the system are listed in appendix \ref{appendix:useful_expressions}.

\subsection{Main results} \label{section:main_results}
\begin{theorem} \label{thm:main}
    Let $e \in \mathbb{R} \setminus \{0\}$, $M > 0$ and $\Lambda > 0$ be, respectively, the charge, mass and cosmological constant associated to a reference sub-extremal Reissner-Nordstr{\"o}m-de Sitter black hole. Let us consider the future maximal globally hyperbolic development $(\mathcal{M}, g, F, \phi)$ of \eqref{main_system} with respect to  spherically symmetric initial data prescribed on two transversal null hypersurfaces $ \underline{C}_{v_0} \cup C_0$. Let $(u, v)$ be the system of null coordinates given in \eqref{sph_symm_metric}.

    Assume that the initial data to the characteristic initial value problem (IVP) satisfy the assumptions in section \ref{section:assumptions}, that require, in particular, that:
    \begin{itemize}
        \item the initial data asymptotically approach those of the reference sub-extremal Reissner-Nordstr{\"o}m-de Sitter black hole along the initial outgoing hypersurface,
        \item  the following exponential Price law upper bound holds along the event horizon $\mathcal{H}^+ \equiv C_0 = [v_0, +\infty)$:
         \begin{equation} \label{thm_Price_law}
    |\partial_v \phi|_{|\mathcal{H}^+}(0, v) <  e^{-sv},
    \end{equation}
    where  $v$ is an Eddington-Finkelstein type of coordinate,
    \item given $p$ fixed, $2 \le p < +\infty$, along the initial (compact) ingoing hypersurface $\underline{C}_{v_0} = [0, U]$ we assume that
    \[
    (\partial_u \phi)_{|\underline{C}_{v_0}}(u) = (\partial_u \phi)_0(u) + \omega f_1(u), \quad \forall\, u \in [0, U],
    \]
    where $u$ is a Kruskal type of null coordinate,  $\omega \in \mathbb{R} \setminus \{0\}$,  $(\partial_u \phi)_0 \in C^0(\underline{C}_{v_0})$ and
    \begin{equation} \label{hR_rule}
    f_1 \in  L^{p}_u(\underline{C}_{v_0}) \setminus  \bigcup_{q > p} L^{q}_u(\underline{C}_{v_0}).
    \end{equation}
    \item along the initial outgoing hypersurface $C_0 = [v_0, +\infty)$ we assume that
    \[
    \partial_v \phi_{|C_0}(v) = (\partial_v \phi)_0(v) +\omega f_2(v), \quad \forall\, v \ge v_0,
    \]
    where $(\partial_v \phi)_0, f_2 \in C^0({C}_0)$ and $\omega \in \mathbb{R} \setminus \{ 0 \}$.
\end{itemize}
    Then, for $U > 0$ sufficiently small compared to the initial data, there exists a unique solution (in a suitable integrated sense\footnote{See already definition \ref{def_W1psolution}.}) to the characteristic IVP, belonging to $W_u^{1, p} \setminus W_u^{1, q}$ (for every $q$ such that $q > p$), defined in $[0, U] \times [v_0, +\infty)$ (which is the future domain of dependence  of $ \underline{C}_{v_0} \cup C_0$ expressed in our coordinates) and described by  metric \eqref{sph_symm_metric}. Moreover:
    \[
    u \mapsto \lim_{v \to +\infty} r(u, v) \ge \mathfrak{r} > 0
    \]
    is a continuous function satisfying
    \[
    \lim_{u \to 0} \lim_{v \to +\infty} r(u, v) = r_{-},
    \]
    where $r_{-}$ is the radius of the Cauchy horizon of the reference black hole.

    Furthermore, we have the following stability results:
  \begin{enumerate}
      \item $(\mathcal{M}, g, F, \phi)$ can be extended up to the Cauchy horizon with continuous regularity of the metric and of the scalar field.
      \item If $p > 2$, consider the quantity
      \[
      \rho \coloneqq \frac{K_{-}}{K_+} > 1,
      \]
      where $K_{-}$ and $K_+$ are the absolute values of the surface gravities of the Cauchy horizon and of the event horizon, respectively, of the reference sub-extremal Reissner-Nordstr{\"o}m-de Sitter black hole.
      If:
      \begin{equation} \label{sufficient_H1_extensions}
      s > K_{-} \quad \text{ and } \quad \rho < 2 \br{1 - \frac{1}{p}},
      \end{equation}
     then the Hawking mass $\varpi$ remains bounded\footnote{Notably, there exists a subset of black hole parameters for which $H^1$ extensions are allowed but the Ricci scalar diverges. See already remark \ref{no_C2_extensions}.} up to the Cauchy horizon $\mathcal{CH}^+$ and there exists a coordinate system in which $g$ can be extended continuously up to $\mathcal{CH}^+$, with Christoffel symbols in $L^2_{\text{loc}}$ and $\phi \in H^1_{\text{loc}}$. In this context, $\varpi$ is a $C^0_u C^1_v \cap W^{1,p}_u W^{1, p}_v$ geometric quantity.
    \end{enumerate}

    Additionally assume that, for some $\gamma > \frac12$:
    \begin{equation} \label{chosen_initial_data}
        f_1(u) = \begin{cases}
        u^{-\frac{1}{p}} |\log(u)|^{-\frac{2 \gamma}{p}}  , &\text{if }u \in (0, U],\\
        0, &\text{if } u = 0,
        \end{cases}
    \end{equation}
    so that $f_1 \in 
 L^p_u(\underline{C}_{v_0}) \setminus \bigcup_{q > p} L^q_u(\underline{C}_{v_0})$,
    and assume that:    
    \begin{align} 
         f_2(v) &\ge C \, e^{-l(s)v}, \label{thm_Price_law_lower}\\
         \limsup_{v \to +\infty} \frac{|\partial_v \phi|(0, v)}{f_2(v)} &= 0,
      \end{align}
      with $s < l(s) < 3s$ and $C > 0$.    
    Then, under these further assumptions, we have the following instability results:
    \begin{enumerate}   
      \setcounter{enumi}{2}
      \item If $p = 2$, mass inflation occurs at the Cauchy horizon  for every value  of $s$. In other terms, the $C^0_u C^1_v  \cap H^1_u H^1_v$ geometric quantity\footnote{See already remark \ref{remark:signs}.}
      \[
      \varpi = \frac{r}{2}\br{1 - g(dr^{\sharp}, dr^{\sharp})} + \frac{e^2}{2r} - \frac{\Lambda}{6}r^3
      \]
      diverges at the Cauchy horizon:
      \[
      \lim_{v \to +\infty} \varpi(u, v) = +\infty,
      \]
      for every $u \in (0, U]$. 
      \item  If $p > 2$ 
      then mass inflation occurs if
      \[
      \text{either } \quad l(s) < K_{-}  \quad \text{ or } \quad \rho > 2 \br{1 - \frac{1}{p}}.
      \]       
  \end{enumerate}
\end{theorem}
We refer to fig.\ \ref{fig:main_results} for a representation of these statements. To the best of our knowledge, these are the first rigorous mass inflation results for the Einstein-Maxwell-scalar field system in the case in which non-smooth initial data are prescribed.

\begin{figure}[h]
    \centering
    \includegraphics[width=0.85\linewidth]{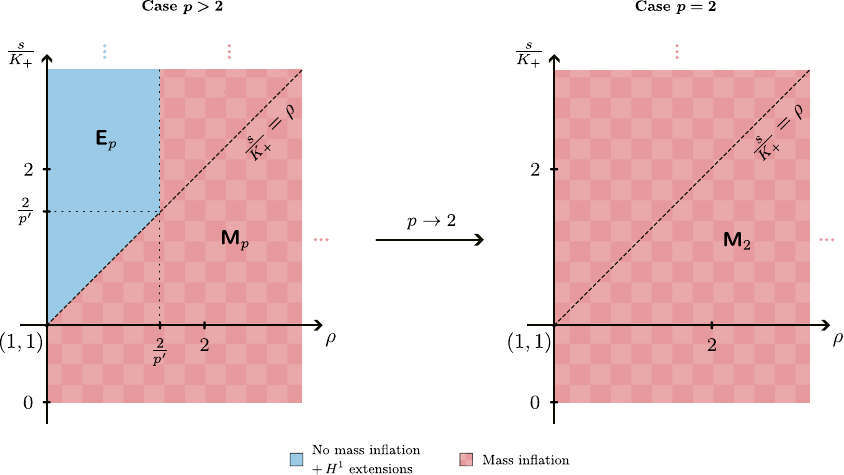}
    \caption{
    The main results of this paper in terms of the parameters of the reference Reissner-Nordstr{\"o}m-de Sitter black hole and of the decay rate given by the assumed Price law, when \textbf{non-smooth} initial data are considered, i.e.\ for a choice $\bm{\mathcal{Q}} = \bm{\mathcal{Q}}(s, \rho) \in \mathcal{R}_{\text{i.d.}}^{p, \text{FE}}$ of initial data (see also fig.\ \ref{fig:moduli_space}). 
    Here, we describe the role of $s$ and $\rho$ in yielding mass inflation or $H^1$ extensions for the non-smooth black hole we construct. 
    The conjugate exponent $p'$ is defined by $pp'=p+p'$.     
    In this figure, we are assuming the Price law upper and lower bounds from \cite{CGNS4}, i.e.\ \eqref{thm_Price_law} and \eqref{thm_Price_law_lower} with $l(s) = s + \varepsilon$, for some $\varepsilon > 0$ small. The stability results in the case $p > 2$ do not require any lower bound. The mass inflation results for $ p > 2$ hold (in a smaller region) even  when assuming a weaker Price law lower bound, see also fig.\ \ref{fig:inflation_lowerbound}. }
    \label{fig:main_results}
\end{figure}

\begin{remark} \label{remark:id_example}
    Initial data of the form of \eqref{chosen_initial_data} were previously taken into account in \cite{Luk}, to construct a class of solutions to the vacuum Einstein equations having weak null singularities.

    We stress that different choices of $\gamma$ in \eqref{chosen_initial_data} give rise to different lines (passing through the same point belonging to $\mathcal{S}_{\text{i.d.}}$) in fig.\ \ref{fig:moduli_space}.
\end{remark}

\subsection{Strategy of the proof}
We proceed as follows:
\begin{enumerate}
    \item First, we study the local well-posedness of \eqref{wave_r}--\eqref{raych_v} in our low-regularity setting (section \ref{section:rough_LWP}).
    \item In section \ref{section:construction_interior}, we construct a black hole interior solution (locally near timelike infinity $i^+$) by adapting the bootstrap methods of \cite{Rossetti}.\footnote{During the proof, we show that fig.\ 4 in \cite{Rossetti} describes our setting as well, with the minor difference that the apparent horizon $\mathcal{A}$ is now a $C^0$, rather than $C^1$, curve.} These were developed in the smooth case to deal with the exponential bounds stemming from the presence of a positive cosmological constant.\footnote{For instance, the methods used in the early-blueshift region $\{u_{\text{EF}} \sim v_{\text{EF}}\}$ in the asymptotically flat case, where $u_{\text{EF}}$ and $v_{\text{EF}}$ are both Eddington-Finkelstein coordinates,  cannot be employed in the present cosmological scenario, since they would yield exponentially large error terms.} Differently from the  smooth case, however, only an integrated formulation of the system is available. In particular, $\partial_u \phi$ is not continuous, hence its $L^{\infty}$ norm cannot be bootstrapped. Moreover, instead of controlling the main PDE quantities in terms of the function
    \[
        e^{-c(s)v}, \quad \text{ with }c(s) \coloneqq \min\{s, 2K_+-\eta\},
    \]
    as in the smooth case (where $\eta > 0$ small and $c(s)$ measures the competition between redshift and the exponential Price law at the event horizon), now we actually propagate bounds in terms of
    \[
        e^{-c_p(s)v}, \quad \text{ with }c_p(s) \coloneqq \min \left \{s, \frac{2K_+}{p'}-C(\eta) \right \},
    \]
    where $p'$ is the conjugate exponent of $p$, i.e.\ $pp' = p + p'$. It is indeed explicit in the proof that the lower regularity we impose, the slower is the decay.
    \item To deal with rough initial data, we bootstrap estimates only on the continuous PDE quantities and, from these, we obtain integral bounds on $\partial_u \phi$ throughout the interior region. The bounds on $\partial_u \phi$ degenerate near the Cauchy horizon $\mathcal{CH}^+$, due to the fact that the scalar field potentially fails to be $H^1$ in a neighbourhood of $\mathcal{CH}^+$, under our low-regularity assumptions.
    To show the integral bounds, we employ the density of smooth functions in $W^{1, p}$.\footnote{More precisely, by a Picard-iteration type of argument, it is possible to construct a sequence of smooth solutions to our PDE system (in its integrated formulation) that converges in $W^{1, p}$ norm. See section \ref{section:rough_LWP}.}
    To construct the black hole interior, we do not rely on the monotonicity properties of the system, hence the proof is well-suited for generalization to the case of a charged scalar field.
    \item To construct $H^1$ extensions in the case $p > 2$ (i.e.\ when initial data are more regular than just finite-energy), we use the techniques of \cite{CGNS4, Rossetti} that, after minor modifications, apply to the current low-regularity setting (see section \ref{section:no_mass_inflation}).
    \item Finally, in section \ref{section:mass_inflation} we exploit the mass inflation proofs first derived in  \cite{Dafermos_2005_uniqueness} (for the $\Lambda = 0$ case) and \cite{CGNS3} (for the $\Lambda > 0$ case), by adapting them to the integral formulation of \eqref{main_system}. In particular, the mild blow-up encoded in the initial data allows to propagate estimates in a non-trapped region and to prove mass inflation for a larger set of black hole parameters, compared to the smooth case. Moreover, these instability results depend on the profile (hence on the regularity) of the prescribed initial data. Differently from the previous sections, this part highly relies on the monotonicity properties of the Einstein-Maxwell-real scalar field system.
\end{enumerate}
Further details on the critical technical steps of the proofs are summarized in the following section.

\subsection{Technical overview} \label{section:technical_diff}

The black hole solution that we construct relies on a subdivision of the interior region in terms of curves of constant area-radius, as usual in the case of cosmological black holes. During our analysis, we address the following points:
\begin{enumerate}
    \item \textbf{Strength of the blow-up of the initial data}: we prescribe initial data for $\partial_u \phi$ that are integrable but not continuous. Our finite-energy data present an exponential blow-up in the Eddington-Finkelstein coordinate $u_{EF}$ as $u \to 0$. Applying, e.g.\ perturbative methods to solutions of the linear wave equation generally proves ineffective for the aims of this paper:  here we pursue a more precise non-linear stability analysis.
    \item \textbf{Low regularity}: not only pointwise bounds are not available for $\partial_u \phi$ and $\partial_u \Omega^2$ (so that the bootstrap procedure has to rely on other continuous PDE quantities), but the unbounded initial data have to be controlled in the entire black hole interior region. We handle this issue by exploiting the decay due to the redshift effect near the event horizon. Notably, this effect is sufficiently strong to control such data (in $L^p_u$) up to the Cauchy horizon. This redshift contribution is encoded in the term $e^{-av}$ (where $0 < a < 2K_+$ can be chosen close to $2K_+$) that multiplies $\partial_u \phi(u, v_0)$ during our proof and allows to prove integral estimates (see e.g.\ \eqref{redshift:pointwise_duphi1} and following). The term $e^{-av}$, which could be ignored in the smooth case without any consequence on the results, now plays a prominent role. In fact, it provides extra decay when $\partial_v \phi$ decays slowly. We remark that the regime of slow decay for $\partial_v \phi$ (i.e.\ the small-$s$ case), is associated with Cauchy horizon instabilities.
    \item \textbf{Weaker estimates due to point (2)}: the bootstrap procedure of \cite{Rossetti} (corresponding to the smooth setting, where the scalar field is possibly massive and charged) had to be revisited so to avoid any bootstrap assumption on the quantity $\partial_u \phi$. Nonetheless, we are still able to obtain sufficiently many $L^p$ bounds that, in turn, allow to close the bootstrap.
    \item \textbf{The blueshift effect and weaker bounds near $\mathcal{CH}^+$ due to point (3)}:
    the previous issues come on top of the challenges that already arise  near the Cauchy horizon in the smooth setting. During the construction of the black hole interior, as we approach the to-be Cauchy horizon, less decay is available to propagate the main estimates and, due to our rough setting, we have weaker bounds at our disposal. In the no-shift region, both induction and bootstrap estimates need to be controlled at the same time so to propagate decay. In the early-blueshift region, the monotonic function $\mathcal{C}_s$ (already introduced in \cite{Rossetti}) allows to bootstrap our exponential bounds without losing more than a small $\varepsilon$ power of decay. This allows to deal with the fact that the early-blueshift region is very large when measured in terms of the $v$ coordinate, since its future boundary is not given by a curve of constant area-radius. In the late-blueshift region, we further have that the $L^p$ bound on $\partial_u \phi$ is degenerate, due to the potential blow-up of energy at the Cauchy horizon. Still, the smallness of this region (in terms of $r$) is such that the $L^p$ bound on $\partial_u \phi$ allows to close the bootstrap estimate on $\partial_v \phi$. We stress that the exponential nature of the decay rates is not a necessary condition to propagate the estimates: polynomial decay is sufficient, as can be seen from the $\Lambda = 0$ case. In fact, the key point of our proof is that, even in this low-regularity setting, we propagate the main exponential bounds from the event horizon to the Cauchy horizon while losing only an $\varepsilon$ power in the exponent, despite the presence of blueshift. This is possible by a careful analysis of the redshift and blueshift effects in the interior and how such effects interact with the blow-up encoded in the initial data (see already section \ref{section:related_works}). The fact that we have an  $\varepsilon$-loss of decay in our proof, rather than more sub-optimal results, plays a crucial role when studying the stability of the Cauchy horizon.\footnote{For instance, in \cite{Rossetti}, the estimates for $\partial_v \phi$ in the late-blueshift region were improved compared to those in \cite{CGNS4}. The results of the former paper, that can be applied to the real-scalar-field case, led to the construction of $H^1$ extensions for a larger set of parameters of the reference black hole.  The improvement was mainly possible due to the bootstrap methods employed in \cite{Rossetti}, since they allowed to bound $\lambda$ and $\partial_v \phi$ simultaneously via the Raychaudhuri equation in $v$.}
    \item \textbf{Decay constraints required to prove mass inflation in the smooth case}: to show mass inflation, our proof follows those in \cite{Dafermos_2005_uniqueness} (case $\Lambda = 0$) and \cite{CGNS3} (case $\Lambda > 0$). On the other hand, the way we propagate estimates in the non-trapped region departs from those approaches. Indeed, this part relies on our construction of the  initial data, that allows to prove mass inflation in the rough  case for a larger set of black hole parameters  (compare fig.\ \ref{fig:smooth_results} with fig.\ \ref{fig:main_results}). We also refer to remark \ref{remark:comparison} for further comparisons between the rough and smooth cases in the context of mass inflation.
\end{enumerate}

\subsection{Related works and outlook} \label{section:related_works}

The motivation behind the present article stems from the following result, proved in \cite{DafYak18}: there exist \textbf{finite-energy initial data} $(\phi_0, \phi_0') \in H^1_{\text{loc}} \times L^2_{\text{loc}}$ to the linear wave equation
\[
    \square_g \phi = 0,
\]
on a sub-extremal Reissner-Nordstr{\"o}m-de Sitter black hole $g$, where the data are prescribed on a complete spacelike hypersurface in the black hole exterior, such that 
\begin{equation} \label{infty_H1}
    \|\phi\|_{H^1_{\text{loc}}(\mathcal{CH}^+)}=+\infty
\end{equation}
at the Cauchy horizon $\mathcal{CH}^+$, for any choice of black hole parameters that allow the existence of a  black hole interior region.
In particular, $\phi$ has infinite energy at the Cauchy horizon.\footnote{Continuous extensions are still allowed.}

This is in contrast with the violations of strong cosmic censorship expected in the case $\Lambda > 0$ when starting from \textbf{smooth} data. These expectations are supported by the existence of $H^1$ extensions for solutions to the wave equation (at the linear level) and for metric solutions to the Einstein equations (at the non-linear level).
In fact, in the near-extremal case, linear waves are expected to decay at a fast\footnote{How fast the scalar field decays depending on the black hole parameters is still an open question.} exponential rate along the event horizon of Reissner-Nordstr{\"o}m-de Sitter black holes (see the numerical results for neutral fields \cite{QNMandSCCC} and for charged fields \cite{SCCCstill_subtle}). If such a rapid decay is prescribed along the event horizon for solutions to the linear wave equation on these black hole backgrounds, then the scalar field can be extended with $H^1$ regularity beyond the Cauchy horizon, i.e.\ \eqref{infty_H1} is false in the smooth setting (see the rigorous works \cite{CostaFranzen, HintzVasy1}). These interior results were further confirmed in the non-linear setting of the spherically symmetric Einstein-Maxwell-scalar field model, see \cite{CGNS4} (neutral scalar field case) and \cite{Rossetti}  (charged scalar field case).\footnote{The works \cite{CGNS2} and \cite{CGNS3}, whose results are employed in \cite{CGNS4}, dealt with the spherically symmetric Einstein-Maxwell-real scalar field system with $\Lambda > 0$ under the additional assumptions of trivial initial data along the event horizon, i.e.\ $\phi_{|\mathcal{H}^+} \equiv 0$.} Outside spherical symmetry, similar results are expected for suitable perturbations of Kerr-Newman-de Sitter black holes near a charged non-rotating solution (see e.g.\ the numerical work \cite{DaveyDiasGil}).

The above results for rough  and smooth  initial data solutions are in line with the following heuristics:
\begin{itemize}
    \item \textbf{The blueshift effect causes an exponential growth in the radiation energy}: one of the main sources of Cauchy horizon instability is given by the blueshift effect, that is an exponential growth of the  energy of linear waves, proportional to $K_{-} > 0$, i.e.\ the absolute value of the surface gravity of the black hole Cauchy horizon.
    \item \textbf{The blueshift effect is not necessarily an instability mechanism}: in the near-extremal limit $K_{-} \to 0$, the exponential growth due to the blueshift effect dampens and can be counter-balanced by the strong dispersive effects due to the presence of a cosmological region.
    \item \textbf{Ingoing radiation can restore the role of blueshift as an instability mechanism:} If sufficient non-trivial radiation is entering the black hole interior, this counterbalances the dispersive effects of spacetime.\footnote{In the language of the present paper, non-trivial ingoing radiation can be described by a divergent behaviour  in the non-smooth wave component $\partial_u \phi$. This causes, in the black hole interior region, a slower decay rate for the  $\partial_v \phi$ component.} In particular, the blueshift effect is still able to trigger an instability.
\end{itemize}
Related studies on the third point above in the case $\Lambda = 0$ are, e.g., in \cite{IsraelPoisson} for the Einstein-null dust system  and in \cite{VdM_mass} for the Einstein-Maxwell-(charged) scalar field system. In both cases, a relation between the amount of ingoing radiation and instabilities at the Cauchy horizon is established. 

It is also remarkable that an analogous type of blueshift effect plays a paramount role in the context of the weak cosmic censorship conjecture and is similarly affected by regularity assumptions, see fig.\ \ref{fig:SCCCvsWCCC}.

The main goal of the current work is to generalize the results of \cite{DafYak18} to the non-linear case. In particular, we  encode a mild blow-up behaviour\footnote{In such a way to still have solutions with finite energy.} in the ingoing direction, so to allow the blueshift effect to, loosely speaking, give rise to an instability mechanism. In our case, linearity cannot be exploited, hence we need to consider a sufficiently general set of initial data and prove that each of these initial configurations yields a black hole interior solution.

\begin{figure}
\centering
\begin{subfigure}{.5\textwidth}
  \centering
  \includegraphics[width=.4\linewidth]{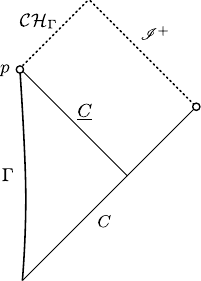}
  \caption{}
\end{subfigure}%
\begin{subfigure}{.5\textwidth}
  \centering
  \includegraphics[width=.6\linewidth]{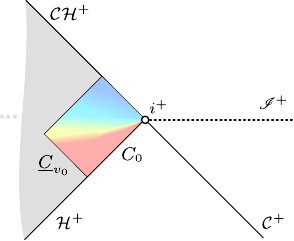}
  \caption{}
\end{subfigure}
\caption{Fig.\ (a): A self-similar naked singularity solution to the Einstein-scalar field system  (with $\Lambda = 0$) constructed by Christodoulou in the BV regularity class (see \cite{Christodoulou_instability_ann}, \cite{Christodoulou_examples} and references therein). Christodoulou proved  that an infinite blueshift effect occurs along $\underline{C}$, i.e.\ the past lightcone of the singularity point $p$. Due to the low regularity of the initial data, this effect gives rise to an \textit{instability} phenomenon: such  BV solutions are non-generic in a ``co-dimension one'' sense.  Notice that the transversal derivatives of the scalar field blow up near $p$. In particular, the solution is completely regular at $C \cap \underline{C}$. We refer to  \cite{SRReview} for a review on the problem.
Fig.\ (b): the (colored) local spacetime region inspected in the current work, embedded in a potential global spacetime structure. Here, the blueshift effect occurs at the Cauchy horizon $\mathcal{CH}^+$. The rough initial data that we prescribe along $\underline{C}_{v_0}$ blow up at $C_0 \cap \underline{C}_{v_0}$ and give rise to an instability at the Cauchy horizon in a neighbourhood of $i^+$. 
}
\label{fig:SCCCvsWCCC}
\end{figure}

We notice that, differently from the smooth case, focusing solely on the black hole interior region is not a greatly restrictive condition. In fact, we show that for finite-energy initial data (i.e.\ $p=2$ in theorem \ref{thm:main}), mass inflation occurs for every value of $s$ in the exponential Price law upper bound.\footnote{Notice that the linear solutions constructed in \cite{DafYak18} are such that $\partial_v \phi_{|\mathcal{H}^+}(0, v) \sim e^{-sv}$ with $K_{-} < s < K_+$. Indeed, the relation $K_{-}/K_+ > 1$ is exploited to construct solutions having finite energy at the event horizon and infinite energy at the Cauchy horizon. We do not have such constraints on $s$ in our case.} Indeed, as already stressed in \cite{DafYak18}, under low-regularity assumptions the dynamics of quasinormal modes in the black hole exterior region do not affect the scenario described by the strong cosmic censorship conjecture.

Starting from the present work, it would be desirable to understand the role of low-regularity solutions in more general settings, with the final goal of studying the vacuum Einstein equations (with $\Lambda$) without symmetry assumptions. The strongest results  at our disposal in this context (when $\Lambda = 0$) are the $C^0$ stability results in \cite{DafermosLuk}. These matters will be investigated in future work and are not discussed here.

\subsection{Outline of the article}
The rest of the article is organized as follows. In section \ref{section:rough_LWP} we define the notion of solution for the integral formulation of our PDE system and give an overview of the main local well-posedness results that we will employ in the article. In section \ref{section:construction_interior} we construct the black hole interior region by propagating the main estimates from the event horizon. The proof relies on a bootstrap procedure adapted to our low-regularity setting. In section \ref{section:no_mass_inflation} we show that, even when assuming only continuity and Sobolev regularity, $H^1$ metric extensions can still be constructed beyond the Cauchy horizon, provided that the blow-up of the rough initial data is not too strong.
Finally, in section \ref{section:mass_inflation} we prove that, for sufficiently rough initial data, the Hawking mass diverges at the Cauchy horizon.

\subsection*{Acknowledgements}
The author would like to thank Mihalis Dafermos for suggesting the problem of rough initial data and for the stimulating discussions during the author's visit to Princeton University in the 2023/2024 academic year. The author is also grateful to Gabriele Benomio, João Costa and José Natário for valuable comments on a preliminary version of the manuscript.

\subsection*{Funding and competing interests}
The author is supported by a postdoctoral research fellowship from the Gran Sasso Science Institute. The author declares that there are no competing interests relevant to the content of this article. 

\section{Rough case: local well-posedness} \label{section:rough_LWP}

In the following, given $U >0, v_0 \ge 0$, we express the initial hypersurfaces as
\begin{align*}
    \underline{C}_{v_0} &= [0, U] \times \{v_0\}, \\
    C_0 = \mathcal{H}^+ &= \{0\} \times [v_0, +\infty),
\end{align*}
with respect to the $(u, v)$ coordinate chart of section \ref{subsection:defs}, that we fix by imposing
\begin{align}
    \nu_{|\underline{C}_{v_0}}(u) &= -1, \quad \forall\, u \in [0, U], \label{assumption_nu}\\
    \kappa_{|C_0}(v) &= 1, \quad \forall\, v \ge v_0.\label{assumption_kappa}
\end{align}
Moreover, we assume that 
\begin{equation} \label{assumption:no_antitrapped}
\nu(0, v) < 0, \quad \text{ for every } v \ge v_0,
\end{equation}
so to avoid the presence of anti-trapped surfaces, and that
\begin{equation} \label{assumption:r_positive}
r_{|\underline{C}_{v_0}} > 0 \quad \text{ and } \quad  r_{|C_0} > 0.
\end{equation}
Without loss of generality, we assume $v_0 = 0$ in this section.

In the next lines, we closely follow \cite[section 10]{GajicLuk} to define the notion of solution for our PDE system. 
\begin{definition}[Solutions to the spherically symmetric Einstein-Maxwell-real scalar field system in an integrated sense] \label{def_integratedsolution}
    Given $\varepsilon > 0$, let $\mathcal{D} = [0, \varepsilon) \times [0, \varepsilon) \subset \mathcal{M}/SO(3) $ and $p \in [2, +\infty)$. \textbf{We say that a continuous function} $\bm{\Psi}=(r, \phi, \Omega^2)\colon \mathcal{D}  \to \mathbb{R}^3$, with $\partial_v \Psi_A \in C^0_v(C^0_u)$ and $\partial_u \Psi_A \in L^p_u(C^0_v)$ for $A=1, 2, 3$ \textbf{is a solution} to \eqref{wave_r}--\eqref{raych_v} \textbf{in an integrated sense} if the components of $\Psi$ satisfy the integral formulation of \eqref{wave_r}--\eqref{raych_v}, i.e.\ if
    \begin{align*}
    \br{\frac{\partial_u r}{\Omega^2}}(u_2, v)- \br{\frac{\partial_u r}{\Omega^2}}(u_1, v) &= - \int_{u_1}^{u_2} \br{ r \frac{(\partial_u \phi)^2}{\Omega^2}}(u', v) du', \,\,\forall\, 0 \le u_1 < u_2 < \varepsilon, \,\, \text{ for a.e. }v \in [0, \varepsilon) \\   
    \br{\frac{\partial_v r}{\Omega^2}}(u, v_2) - \br{\frac{\partial_v r}{\Omega^2}}(u, v_1) &= -\int_{v_1}^{v_2} \br{ r \frac{(\partial_v \phi)^2}{\Omega^2}}(u, v')dv',  \,\,\forall\, 0 \le v_1 < v_2 < \varepsilon, \,\, \text{ for a.e. }u \in [0, \varepsilon)
    \end{align*}
    and if,  
    for every $0 \le u_1 < u_2 < \varepsilon$:
     \[
    \br{\partial_v r}(u_2, v) = \br{\partial_v r}(u_1, v) + \int_{u_1}^{u_2} \br{-\frac{\Omega^2}{4r} -\frac{\partial_u r \partial_v r}{r} + \frac{\Omega^2 e^2}{4 r^3} + \frac{\Omega^2 \Lambda r}{4}}(u', v)du', \quad \text{ for a.e. } v \in [0, \varepsilon), 
    \]
    and for every $0 \le v_1 < v_2 < \varepsilon$:
    \[
    \br{\partial_u r}(u, v_2) = \br{\partial_u r}(u, v_1) + \int_{v_1}^{v_2} \br{-\frac{\Omega^2}{4r} -\frac{\partial_u r \partial_v r}{r} + \frac{\Omega^2 e^2}{4 r^3} + \frac{\Omega^2 \Lambda r}{4}}(u, v')dv', \quad \text{ for a.e. } u \in [0, \varepsilon),
    \]
    and similarly for \eqref{wave_phi} and \eqref{wave_Omega}.    
\end{definition}

\begin{remark}[On regularity] \label{remark:regularity}
    By definition, any solution $\Psi$ in the integrated sense of definition \ref{def_integratedsolution} is such that $\partial_u \Psi(u, \cdot) \in C^1_v$ for a.e.\ $u$ and $\partial_v \Psi(\cdot, v) \in C^1_u$ for a.e.\ $v$. 

    Moreover, as stressed in \cite[Remark 10.1]{GajicLuk}, such a $\Psi$ is necessarily a weak solution of \eqref{wave_r}--\eqref{raych_v} (in the sense of integration against a $C^{\infty}_c$ function in $\mathcal{D}$).
\end{remark}

In appendix \ref{appendix:useful_expressions} we  express the integral formulation of \eqref{wave_r}--\eqref{raych_v} in more convenient forms, similarly to how it is done in the smooth case, even though lower regularity is assumed here.

The following is an adaptation of \cite[Proposition 10.1]{GajicLuk} to a slightly more general setting.
\begin{proposition}[Local existence and uniqueness for solutions in the integrated sense] \label{prop:local_existence}
For $A=1, 2, 3$, consider the quasi--linear hyperbolic system
\begin{align}
\partial_u \partial_v \Psi_A &= f_A(\Psi) + N_A^{BC}(\Psi) \partial_u \Psi_B \partial_v \Psi_C + K_A^{BC}(\Psi) \partial_v \Psi_B \partial_v \Psi_C  \nonumber \\
&+ L_A^B(\Psi) \partial_u \Psi_B + R_A^B(\Psi)\partial_v \Psi_B, \label{nonlinear_system}
\end{align}
where $N_A^{BC}, K_A^{BC}, L_A^B, R_A^B$ are smooth, real-valued functions for every choice of $A, B, C \in \{1, 2, 3\}$.
Given $2 \le p < +\infty$ and $U, V > 0$, prescribe initial data on $[0, U] \times \{0 \} \cup \{0\} \times [0, V]$  such that, for every $A = 1, 2, 3$, we have 
\begin{align*}
(\Psi_A)_{|v = 0} \in C^0_u([0, U]), \\
(\Psi_A)_{|u = 0} \in C^1_v([0, V]), \\
(\partial_u \Psi_A)_{|v = 0} \in L^p_u([0, U])
\end{align*}
and
\begin{align}
\int_0^U |\partial_u \Psi_A|^p (u', 0) du' &\le C_0, \\
\sup_{v' \in [0, V]} |\partial_v \Psi_A|^p (0, v') &\le C_0.
\end{align}
If assumptions \eqref{assumption:no_antitrapped} and \eqref{assumption:r_positive} hold, then there exists $\epsilon = \epsilon(U, V, C_0) > 0$ such that a  solution to \eqref{nonlinear_system},  in the sense of definition \ref{def_integratedsolution}, exists  in $[0, \epsilon) \times [0, \epsilon)$, satisfies
\[
r > 0 \quad \text{ and } \quad \Omega^2 > 0
\]
in $[0, \epsilon) \times [0, \epsilon)$ and is unique in  $W^{1, p}_{u}(C^1_v)$.
\end{proposition}
\begin{proof}
    The proof consists of a fixed-point argument in a space that can be expressed as the product of metric spaces endowed with a $L^{\infty}$ norm and of metric spaces endowed with an $L^p$ norm. This is a common technique: we refer the reader to, e.g.\ \cite{CGNS1} and \cite{mythesis} for more details.

    The a priori bounds can be obtained as in \cite[proposition 10.1]{GajicLuk}. Here we repeat their proof, including a straightforward adaptation to our case of $W^{1, p}$ regularity, for the convenience of the reader. 
    
    If we make the bootstrap assumption
    \[
    \sup_{(u, v) \in [0, \epsilon) \times [0, \epsilon)} |\partial_v \Psi|(u, v) \le 4 M,
    \]
    with $M > C_0$, 
    then $\Psi$ is close, in $L^{\infty}$ norm, to the initial data. Indeed, this is a consequence of the fundamental theorem of calculus when  $\epsilon$ is suitably small:
    \[
    \sup_{(u, v) \in [0, \epsilon) \times [0, \epsilon)} |\Psi(u, v) - \Psi(u, 0)| \le \sup_{u \in [0, \epsilon)} \int_0^{\epsilon} |\partial_v \Psi|(u, v') dv'  \le 4 \epsilon M.
    \]
    This then implies that $f_A(\Psi)$, $N_A^{BC}(\Psi)$, $K_A^{BC}(\Psi)$, $L_A^B(\Psi)$, $R_A^B(\Psi)$ are bounded for every $A, B, C=1, 2, 3$.

    Since $v \mapsto \partial_u \Psi(u, v)$ is $C^1_v$ for almost every $u$, then  \eqref{nonlinear_system} gives:
    \begin{align}
    (\partial_u \Psi_A)^p(u, v) &= (\partial_u \Psi_A)^p(u, 0)  \nonumber \\
    &+p \int_0^v (\partial_u \Psi_A)^{p-1}(u, v') \left [  f_A(\Psi) + N_A^{BC}(\Psi) \partial_u \Psi_B \partial_v \Psi_C + \right. \label{p_integration} \\
    & \left. + K_A^{BC}(\Psi) \partial_v\Psi_B \partial_v \Psi_C  + L_A^B(\Psi) \partial_u \Psi_B + R_A^B(\Psi)\partial_v \Psi_B \right](u, v') dv' \nonumber
    \end{align}
   for every $A= 1, 2, 3$, for a.e.\ $u \in [0, \epsilon)$ and for every $v \in [0, \epsilon)$. After integrating, exploiting the bootstrap bound and the fact that $p \ge 2$, we obtain
    \begin{align*}
    \int_0^{\epsilon} \sup_{v \in [0, \epsilon)} |\partial_u \Psi|^p (u', v) du' & \le M + C \int_0^{\epsilon} \int_0^{\epsilon} |\partial_u \Psi|^{p-1}(u', v') \br{1 + 4 M + 16M^2} du' dv' \\ 
    &+ C(1 + 4M) \int_0^{\epsilon } \int_0^{\epsilon} |\partial_u \Psi|^p(u', v') du' dv' \\
    &\le  M + \epsilon C(M) \int_0^{\epsilon} \sup_{v \in [0, \epsilon)} |\partial_u \Psi|^p(u', v) du',
    \end{align*}
    for some $C>0$.
    Hence:
    \[
      \int_0^{\epsilon} \sup_{v \in [0, \epsilon)} |\partial_u \Psi|^p (u', v) du' \le 2M,
    \]
    for a suitably small choice of $\epsilon$.

    So, after exploiting the integral formulation of the system once again, together with the bootstrap bound and Young's inequality, we obtain:
    \begin{align*}
    \sup_{(u, v) \in [0, \epsilon) \times [0, \epsilon)} |\partial_v \Psi|(u, v) &\le M + C \int_0^{\epsilon} \sup_{v \in [0, \epsilon)}\br{1 + |\partial_v \Psi| + |\partial_u \Psi| + |\partial_v \Psi| |\partial_u \Psi| + |\partial_v \Psi|^2} (u', v)du' \\
    &\le M + (C + C(M)) \epsilon + C \epsilon  \int_0^{\epsilon} \sup_{v \in [0, \epsilon)} |\partial_u \Psi|^p (u', v) du' \le 2 M,
    \end{align*}
    provided that $\epsilon$ is small.
    In particular, the bootstrap argument is closed.
\end{proof}
\begin{corollary}
    Consider the Einstein-Maxwell-scalar field system \eqref{wave_r}--\eqref{raych_v}. This can be cast as \eqref{nonlinear_system} with
\[
\Psi = 
\begin{pmatrix}
r \\%
\phi \\%
{\Omega^2} 
\end{pmatrix} 
\quad f = \begin{pmatrix}
-\frac{\Omega^2}{4r} + \frac{\Omega^2 e^2}{4 r^3} + \frac{\Omega^2 \Lambda r}{4} \\%
0 \\%
-\frac{\Omega^4 e^2}{r^4} + \frac{\Omega^4}{2r^2}
\end{pmatrix} 
\]
and
\[
N_1^{11}= -\frac{1}{r}, \quad N_2^{12} = N_2^{21} = -\frac{1}{r}, \quad N_3^{11} = \frac{2 \Omega^2}{r^2}, \quad N_3^{22} = -2 \Omega^2, \quad N_3^{33} = \frac{1}{\Omega^2},
\]
while all the remaining components are zero. The Raychaudhuri equations are taken as constraints.

Under the assumptions of proposition \ref{prop:local_existence}, the solution $\Psi$ in the specific case of the Einstein-Maxwell-scalar field system satisfies
\[
\partial_u \phi, \partial_u \log \Omega^2 \in L^p_u(C^0_v), \qquad  \partial_v \phi, \partial_v \log \Omega^2, \partial_u r, \partial_v r \in C^0_u (C^0_v)
\]
in $[0, \epsilon) \times [0, \epsilon)$ and, for a sufficiently small choice of $\epsilon$, both $r$ and $\Omega^2$ are bounded and bounded away from zero. Note that the fact that $\partial_u r \in C^0_u(C^0_v)$ is a consequence of the Raychaudhuri equations.
\end{corollary}

\begin{definition}[$W^{1, p}$ solutions to the spherically symmetric Einstein-Maxwell-real scalar field system]  \label{def_W1psolution}
    Given $\varepsilon > 0$, let $\mathcal{D} = [0, \varepsilon) \times [0, \varepsilon) \subset \mathcal{M}/SO(3) $ and $2 \le p < +\infty$. \textbf{We say that }$\bm{\Psi}=(r, \phi, \Omega^2)\colon \mathcal{D} \to \mathbb{R}^3$ \textbf{is a $\bm{W^{1, p}}$ solution} to \eqref{wave_r}--\eqref{raych_v} if the following two conditions are satisfied:
    \begin{itemize}
        \item $\Psi$ is a solution in the integrated sense of definition \ref{def_integratedsolution}, and
        \item for every $v \in [0, \epsilon)$:
    \[
        \partial_u \phi(\cdot, v), \partial_u \log \Omega^2(\cdot, v) \in L^p_u([0, \varepsilon)) \setminus L^q_u([0, \varepsilon)), \quad \forall\,  q > p.        
    \]
    \end{itemize}    
\end{definition}

\begin{remark}[A physical motivation of $W^{1, p}$ solutions]
    If we set up the characteristic IVP for \eqref{wave_r}--\eqref{raych_v} in such a way to describe the interior of a spherically symmetric charged black hole (see already section \ref{section:construction_interior}), then the $W^{1, p}$ solutions defined in \ref{def_W1psolution} describe those spacetimes evolved from finite-energy initial data, where a non-trivial amount of ingoing radiation is assumed. This is related to the fact that the second condition in definition \eqref{def_W1psolution} implies, in particular, that non-trivial data for the scalar field are prescribed along the initial ingoing segment.

    In light of the discussion in section \ref{section:related_works}, we expect that this condition concerning the ingoing radiation is able to trigger mass inflation at the Cauchy horizon of the black hole under inspection.

    This is indeed the main result of section \ref{section:mass_inflation}, for  $W^{1, 2}$ solutions. Moreover, we will show that a slight increase in the regularity of initial data (i.e.\ allowing for $W^{1, p}$ regularity with $p > 2$) corresponds to tamer amounts of ingoing radiation. Although mass inflation still generally occurs in this case, nonetheless there exists a subset of initial data (corresponding to near-extremal reference black holes) for which the Hawking mass remains bounded and $H^1$ metric extensions can be constructed beyond the Cauchy horizon. The higher the regularity, the larger the subset of initial data allowing for $H^1$ extensions.
\end{remark}

\begin{remark}[A technical remark on the regularity threshold] \label{remark:technical}
    First, we notice that solutions in definition \ref{def_W1psolution}, that are the low-regularity solutions that we will consider throughout the present paper, are strictly less regular than solutions of the second-order systems in \cite{Rossetti} and in the papers \cite{CGNS1, CGNS2, CGNS3, CGNS4}. 
    
    If we further assumed that $\partial_u \phi, \partial_u \log \Omega^2 \in C^0_u$, then our system would be equivalent to the first-order system in \cite{CGNS1} (or the first-order system in \cite{Rossetti}, provided that, in the latter, we set the scalar field charge and mass to zero). To get the equivalence with the second-order formulation, more regularity would need to be assumed, see for instance \cite[Section 6]{CGNS1}.
\end{remark}

\begin{remark}[Properties of the main quantities] \label{remark:signs}
    It follows from \eqref{appendix_nu_modified}, \eqref{assumption_nu} and \eqref{def_kappa}, that
    \begin{itemize}
        \item $\nu < 0$ and
        \item $\kappa > 0$
    \end{itemize}
    in the domain of existence of solutions to our characteristic initial value problem.

    Moreover, in the low-regularity setting in which we prove mass inflation, \textbf{the geometric quantity $\bm{\varpi}$ is not of $\bm{C^1_{u, v}}$ regularity}! It is a  geometric quantity belonging to $C^0_u C^1_v  \cap W^{1, p}_u W^{1, p}_v$. This is manifest from definition \ref{def_varpi}, from the Raychaudhuri equations \eqref{raych_u}, \eqref{raych_v} and from the fact that the mild blow-up assumed on the initial data is propagated along the $v$-direction (see already section \ref{section:mass_inflation}).
\end{remark}

We now retrieve an extension criterion for the system \eqref{wave_r}--\eqref{raych_v}. Related results were obtained in the case $\Lambda = 0$ (see \cite{CGNS1}) and in the case $\Lambda = 0$ (where the scalar field was possibly charged) in \cite{Kommemi} and in \cite{GajicLuk}. None of these results directly apply here, due to the lower regularity we are assuming.
In particular, the bootstrap argument of \cite{Kommemi} does not immediately generalize to our setting, since now $\partial_u \phi$ and $\partial_u \Omega^2$ are not continuous.
In \cite{GajicLuk}, even though weak extensions were  constructed, such extensions were obtained  beyond the Cauchy horizon  by exploiting a sequence of initial value problems where smooth initial data were prescribed.  

\begin{proposition}[Extension criterion] \label{prop:extension_criterion}
Let $\Psi_0$ be the initial data for \eqref{wave_r}--\eqref{raych_v} prescribed on 
\[
\mathcal{D}_0 \coloneqq [0, U] \times \{0\} \cup \{0 \} \times [0, V],
\]
for some $U, V \in \mathbb{R}_0^+$. Assume that such initial data satisfy the assumptions of proposition \ref{prop:local_existence}. For $p \in [2, +\infty)$, let us consider the unique $W^{1, p}$ solution $\Psi = (r, \phi, \Omega^2)$ to the characteristic IVP in 
\[
\mathcal{D} \coloneqq [0, \varepsilon) \times [0, \varepsilon),
\]
for some $0 < \varepsilon < \min \{U, V \}$. Assume that there exists $L, R$ such that
\[
0 < L \le r(u, v) \le R  < +\infty, \quad \forall\, (u, v) \in \mathcal{D}.
\]
Let 
\begin{align*}
N_{\text{i.d.}} &\coloneqq \|\Psi \|_{L^{\infty}(\mathcal{D}_0)} + \|\partial_v \Psi(0, \cdot) \|_{L_v^{\infty}([0, V])} +\|\partial_u \Psi(\cdot, 0) \|_{L_u^{p}([0, U])}, \\
N(\mathcal{D}) &\coloneqq  \|\Psi\|_{L^{\infty}(\mathcal{D})} + \sup_{u \in [0, \varepsilon)} \|\partial_v \Psi(u, \cdot)\|_{L_v^{\infty}([0, \varepsilon))} + \sup_{v \in [0, \varepsilon)}\|\partial_u \Psi(\cdot, v)  \|_{L_u^{p}([0, \varepsilon))} + \left \| \frac{1}{\Omega^2} \right \|_{L^{\infty}(\mathcal{D})}.
\end{align*}
Then: 
\[
N(\mathcal{D}) \le C = C(\varepsilon, L, R, N_{\text{i.d.}}) < +\infty
\]
and there exists $\delta > 0$ and a $W^{1, p}$ solution $\tilde \Psi$ to the same characteristic initial value problem such that $\tilde \Psi$ is defined on $\mathcal{D}^{\delta} \coloneqq [0, \varepsilon+\delta] \times [0, \varepsilon+ \delta]$.
\end{proposition}
We present the proof of proposition \ref{prop:extension_criterion} in appendix \ref{appendix:extension}.

\section{Rough case: construction of the black hole interior} \label{section:construction_interior}
\subsection{Assumptions} \label{section:assumptions}
In addition to the assumptions in \eqref{assumption_nu}--\eqref{assumption:r_positive} and in definition \ref{def_W1psolution}, we further require the following.
\begin{enumerate}
    \item \textbf{Expression for the rough initial data}:
    along the initial ingoing segment, we require that
\begin{equation} \label{rough_data_general}
    \partial_u \phi_{|\underline{C}_{v_0}} = (\partial_u \phi)_0 + \omega f_1,
\end{equation}
for some $0 \ne \omega \in \mathbb{R}$, where $(\partial_u \phi)_0 \in C^0([0, U])$ and $f_1 \in L^p([0, U]) \setminus L^q([0, U])$ for every $q > p$. 

In section  \ref{section:mass_inflation}, in order to prove the mass inflation results, we will further assume that, for some $\gamma > \frac12$:
\begin{equation} \label{chosen_initial_data2}
        f_1(u) = \begin{cases}
        u^{-\frac{1}{p}} |\log(u)|^{-\frac{2 \gamma}{p}}  , &\text{if }u \in (0, U],\\
        0, &\text{if } u = 0,
        \end{cases}
\end{equation}
which belongs to $L^p_u([0, U]) \setminus L^q_u([0, U])$ for every $q > p$.
    \item \textbf{Exponential Price's law}: for some $s > 0$, we demand that the scalar field obeys
    \begin{equation} \label{price_law_upper}
         |\partial_v \phi |_{|\mathcal{H}^+} (0, v) < e^{-sv},
    \end{equation}
    along the event horizon $\mathcal{H}^+ = \{ u = 0\}$. Only to prove the instability results in the case $p > 2$ (see section \ref{section:mass_inflation}), we will also require that, for some $\omega \in \mathbb{R} \setminus \{ 0 \}$:
    \begin{equation} \label{rough_data_general_v}
        \partial_v \phi_{|\mathcal{H}^+} = (\partial_v \phi)_0 + \omega f_2, 
    \end{equation}
    with $(\partial_v \phi)_0$, $f_2 \in C^0([0, +\infty))$ and where
    \begin{equation} \label{price_law_lower}
    \begin{cases}
    f_2(v) >C  e^{-l(s)v}, \\
    \limsup\limits_{v \to +\infty} \frac{(\partial_v \phi)_0(v)}{f_2(v)} = 0,
    \end{cases}
    \end{equation}
    hold with  $C> 0$ and $s < l(s) < 3s$.
    \item \textbf{Asymptotically approaching a sub-extremal black hole}: along the event horizon, we require:
    \begin{align*}
        \lim_{v \to +\infty} r(0, v) &= r_+, \\
        \lim_{v \to +\infty} \varpi(0, v) &= M, \\
        \lim_{v \to +\infty} \phi(0, v) &= 0,
    \end{align*}
    where  $r_+$ and $M > 0$ are, respectively, the radius of the event horizon associated to the reference black hole and its mass.
    \item \textbf{Hawking's area theorem}: we demand that
    \[
        \lambda(0, v) > 0, \quad \forall\, v \ge v_0.
    \]
    \end{enumerate}
    We will also take $v_1 \ge v_0$ to be a constant, sufficiently large compared to the initial data.
    Following a procedure similar to those in \cite{CGNS1, CGNS2, CGNS3, CGNS4, Rossetti} (see also \cite{Dafermos_2003, Dafermos_2005_uniqueness, LukOh1, VdM1} for related methods in the $\Lambda = 0$ case), in the next sections we construct the maximal past set\footnote{$\mathcal{P} \subset [0, U] \times [v_0, +\infty)$ is called a \textbf{past set} if $J^{-}(u, v) \coloneqq [0, u] \times [v_0, v]$ is contained in $\mathcal{P}$ for every $(u, v) \in \mathcal{P}$.} $\mathcal{P}$ containing
    \[
        \mathcal{D}_0 \coloneqq [0, U] \times \{v_0\} \cup \{0\} \times [v_0, +\infty), \quad \text{ for some } 0 < U < 1, v_0 \ge 0. 
    \]
    In this context, the causal and chronological past of a point are defined as $J^{-}(u, v) \coloneqq [0, u] \times [v_0, v]$ and $I^{-}(u, v) \coloneqq [0, u) \times [v_0, v)$, respectively, and analogous definitions hold for $J^+(u, v)$ and $I^+(u, v)$.
    Differently from the aforementioned works, the construction of the present paper is given in terms of a low-regularity, $W^{1, p}$ solution to the characteristic IVP of section \ref{section:rough_LWP}.

    A posteriori, the black hole interior we construct will be subdivided into several regions determined by level sets of the area-radius. Once the construction of the non-smooth black hole interior is completed, many properties of such curves can be inferred from the smooth case, see e.g.\ fig.\ 4 and section 4 in \cite{Rossetti}, due to the fact that, even in our low-regularity setting, we have $r \in C^1_{u, v}$.
    For the curves of constant area-radius $\rho \in (r_{-}, r_+)$, we use the notation
    \[
        \Gamma_{\rho} \coloneqq \{ (u, v) \in \mathcal{P}\colon \, r(u, v) = \rho \} = \{ (u_{\rho}(v), v) \colon \, v\ge v_0 \} = \{ (u, v_{\rho}(u)) \colon \, 0 < u \le U \},
    \]
    for two suitable $C^1$ functions $u_{\rho}$ and $v_{\rho}$. Moreover, we use the notation $N_{\text{i.d.}}$ to express the norm of the initial data described in the statement of proposition \ref{prop:extension_criterion}.

    \subsection{Notations and conventions}
    Given two non-negative functions $f$ and $g$, we write $f \lesssim g$ (alternatively, $f = O(g)$) if there exists $C > 0$ such that $f \le C g$. Similarly, we define $f \gtrsim g$. If both $f \lesssim g$ and $f \gtrsim g$ hold, then we write $f \sim g$. We employ the notation $f = o(g)$ as $z \to z_0$ (where $f$ does not need to be non-negative) if for every $\epsilon > 0$, there exists $\delta > 0$ such that $|f(z)| \le \varepsilon g(z)$ whenever $|z - z_0| < \delta$.
    
    Different constants may be denoted by the same symbols if the
    value of such constants is not relevant in the computations.

    Throughout the article, we employ the null chart $(u, v) = (u_{\text{Kru}}, v_{\text{EF}})$, where $u \in [0, U]$ is a Kruskal type of coordinate and $v \in [v_0, +\infty)$ is an Eddington-Finkelstein type of coordinate.

\subsection{Event horizon}
\begin{proposition}[Bounds along the event horizon] \label{prop:event_horizon}
    For every $v \ge v_0$, we have
    \begin{align}
        0 < r_+ - r(0, v) &\le C_{\mathcal{H}} e^{-2sv}, \label{EH:r_bound} \\
        0 < \lambda(0, v) &\le C_{\mathcal{H}} e^{-2sv}, \label{EH:lambda}\\
        \frac{e^{2K_+v}}{C_{\mathcal{H}}} \le |\nu|(0, v) = \frac{\Omega^2}{4}(0, v) &\le C_{\mathcal{H}} e^{2K_+v} \label{EH:omega} \\
        |\varpi(0, v) - M| &\le C_{\mathcal{H}} e^{-2sv}, \label{EH:varpi} \\
        |K(0, v) - K_+| &\le  C_{\mathcal{H}} e^{-2sv}, \\
        |\partial_v \log \Omega^2(0, v) -2K(0, v)| &\le C_{\mathcal{H}} e^{-2sv}, \label{EH:logOmega} \\
        |\phi|(0, v) + |\partial_v \phi|(0, v) &\le C_{\mathcal{H}} e^{-sv}, \label{EH:phi_bound}
    \end{align}
    where $C_{\mathcal{H}} > 0$ depends uniquely on the initial data.
\end{proposition}
\begin{proof}
    The proof is similar to that of \cite[proposition 4.2]{Rossetti}, where now we need to employ the integral formulation of \eqref{wave_r}--\eqref{raych_v} (for instance, the equation for $\partial_v \log |\nu|$ is proved by using the integrated Leibniz rule of lemma \ref{lemma:integral_leibniz_rule}). Since, along the event horizon, there are only minor differences between the smooth and the rough case, we omit the proof here.

    However, we stress a crucial distinction between the two cases: no bounds on $|\partial_u \phi|$ are generally available here. Indeed, under our assumptions it is possible to choose  initial data for $\partial_u \phi$ given by \eqref{rough_data_general}, with $f_1$ as in \eqref{chosen_initial_data2}. This is an admissible choice of rough initial data yielding $\partial_u \phi = +\infty$ identically along the event horizon (as follows e.g.\ from \eqref{appendix_rduphi}).
\end{proof}

\subsection{Redshift region}
\begin{proposition}[Propagation of estimates near the event horizon] \label{prop:redshift}
    Let $2 \le p < +\infty$, let $\eta > 0$ be small and define 
    \begin{equation} \label{def_cps}
    c_p(s) \coloneqq \min \left \{ s, \frac{2K_+}{p'} - C(\eta) \right \},
    \end{equation}
    where $\frac{1}{p'} = 1- \frac{1}{p}$ and for some $C(\eta) > 0$ that vanishes as $\eta$ goes to zero.

    Then, for every $\delta \in (0, \eta)$ small compared to the initial data, given $R$ satisfying
    \[
    0 < r_+ - R < \delta
    \]
    and given the curve $\Gamma_R=\{(u, v) \colon\, r(u, v) = R\}$, 
    we have that $J^{-}(\Gamma_R) \cap \{ v \ge v_1\} \ne \emptyset$. Moreover, for every $(u, v) \in J^{-}(\Gamma_R) \cap \{v \ge v_1\}$, the following inequalities hold:
    \begin{align}
        2\br{ r(0, v) - r(u, v)} \le u\Omega^2(0, v) &\le C \delta, \\
        0 < r_+ - r(u, v) &\le C e^{-2sv} + u\Omega^2(0, v), \\
        |\lambda|(u, v) &\le C e^{-2sv} + C u\Omega^2(0, v), \label{R:lambda} \\
        |\nu|(u, v) &\sim \Omega^2(0, v), \\
        |\varpi(u, v) - M| &\le C \delta^2 e^{-c_p(s)v}, \\
        |K(u, v) - K_+| &\le C \br{\delta^2 e^{-c_p(s)v} + u\Omega^2(0, v)} \lesssim \delta, \label{redshift:K_bound}\\
        |\partial_v \log \Omega^2(u, v) - 2K(u, v)| &= o_v(1), \quad \text{ as } v \to +\infty  \label{redshift:dvlogOmega}\\
        \|\partial_u \phi(\cdot, v)\|_{L^1_u([0, u])} &\le \delta e^{-c_p(s)v}, \label{redshift:duphi_L1} \\
        |\phi|(u, v) + |\partial_v \phi|(u, v) &\le C e^{-c_p(s)v}, 
    \end{align}
    where $C=C(N_{\text{i.d.}}, \eta)$. Furthermore:
    \begin{equation} \label{R:monotonicity}
    \lambda(u_2, v) < \lambda(u_1, v), \quad \forall\, (u_1, v), (u_2, v) \in J^{-}(\Gamma_R) \cap \{v \ge v_1\} \text{ such that } u_2 > u_1.
    \end{equation}
\end{proposition}

\begin{proof}
Let
\begin{equation} \label{def_Pdelta}
\mathcal{P}_{\delta} \coloneqq \big \{ (u, v) \in [0, U] \times [v_1, +\infty) \colon 0 < r_+ - r(u, v) \le \delta \big \},
\end{equation}
which is non-empty due to the fact that the $r$ coordinate gives a parametrization of the event horizon.
Define $E$ as the set of points $q$ in $\mathcal{P}_{\delta}$ such that the following inequalities hold for every $(u, v) \in J^{-}(q) \cap \mathcal{P}_{\delta}$:
\begin{align}
|\phi|(u, v) + |\partial_v \phi|(u, v) &\le 4 C_{\mathcal{H}}  e^{-c_p(s)v}, \label{redshift:bootstrap_phi} \\
\kappa(u, v) & \ge 1-\eta, \label{redshift:bootstrap_kappa} \\
\frac{\Omega^2(0, v)}{8} \le |\nu|(u, v) &\le \Omega^2(0, v), \label{redshift:bootstrap_nu}\\
|\varpi(u, v) - M| &\le \eta. \label{redshift:bootstrap_varpi}
\end{align}
The above inequalities hold along the event horizon, for large values of $v$ (see proposition \ref{prop:event_horizon}, \eqref{assumption_kappa} and recall that $\Omega^2(0, v)=4 |\nu|(0, v)$). Moreover, all the quantities above are continuous in $\mathcal{P}_{\delta}$.

Let $(u, v) \in E$. Using the bootstrap inequality for $|\nu|$ and the fact that $r(0, v) < r_+$ for every $v$:
\begin{equation} \label{first_redshift}
\delta \ge r(0, v) - r(u, v)  = \int_0^u |\nu|(u', v)du' \ge C u \Omega^2(0, v).
\end{equation}
A similar reasoning, together with \eqref{EH:r_bound}, gives:
\[
r_+ - r(u, v) = r_+ - r(0, v) +r(0, v) - r(u, v) \le C e^{-2sv} + u\Omega^2(0, v).
\]
Inequality \eqref{first_redshift} and analogous computations (recall that $\Omega^2(0, v) \sim e^{2K_+v}$) show that
\begin{equation} \label{redshift:bound_u}
    (r(0, v)- r(u, v)) e^{-2K_+v}\lesssim u  \lesssim \delta e^{-2K_+v}.
\end{equation}
Moreover, \eqref{appendix_rlambda_eqn}, \eqref{EH:lambda}, the fact that $\Omega^2 = 4 |\nu| \kappa$ and the bootstrap estimates imply:\footnote{Here and in the following, we will repeatedly use that $r$ is bounded and bounded away from zero and that $\Omega^2 = -4 \nu \kappa$ (see \eqref{def_kappa}).}
\begin{equation} \label{redshift:lambda_bdd}
|r \lambda|(u, v) \le r_+ |\lambda|(0, v)  + C(N_{\text{i.d.}}, \eta) u \Omega^2(0, v)  \le C e^{-2sv} + C u\Omega^2(0, v).
\end{equation}
The previous bounds on $r$ (notice that $u\Omega^2(0, v) \lesssim \delta$) and the bootstrap bound on $\varpi$ also give:
\begin{equation} \label{redshift:K}
| K(u, v) - K_+ | = \left | \br{\frac{\varpi}{r^2} - \frac{e^2}{r^3} - \frac{\Lambda}{3}r }(u, v) - \br{\frac{M}{r_+^2} - \frac{e^2}{r_+^3} - \frac{\Lambda}{3}r_+ } \right | = O(\eta) + O(\delta).
\end{equation}

Let us assume, for the moment, that\footnote{Under this assumption, our PDE system is (at least) equivalent  to the first-order system in \cite{CGNS1}, see also remark \ref{remark:technical}. Notice that an adaptation of the Picard iteration argument in \cite[theorem 4.2]{CGNS1} allows to find a sequence of smooth solutions satisfying \eqref{wave_r}--\eqref{raych_v} that converges,  in $W^{1, p}_u C^0_v$ norm, to a solution in the sense of definition \ref{def_W1psolution}.}
\begin{equation} \label{redshift:temp_assumptions}
\partial_u \phi(\cdot, v), \partial_u \log\Omega^2(\cdot, v) \in C^{\infty}_u([0, u]).
\end{equation}
We will ultimately drop this assumption and approximate the Sobolev norms of $\partial_u \phi(\cdot, v)$ by a density argument.
Under the validity of \eqref{redshift:temp_assumptions}, we can exploit some expressions obtained in \cite[Proposition 4.7]{Rossetti} (where we can set the scalar field charge and mass to zero). In particular, given
\[
f \coloneqq \frac{r \partial_u \phi}{\nu},
\]
and a constant $a$ defined such that
\[
\begin{cases}
a - \kappa (2K) < 0, \\
a -c_p(s) > 0,
\end{cases}
\]
it is showed in \cite{Rossetti} that the following holds: 
\begin{equation} \label{preliminary_Duphi}
\partial_v \br{e^{av} f(u, v)} 
= -e^{av}\partial_v \phi + e^{a v} f(u, v) (a - \kappa(2K)).
\end{equation}
Notice that, due to \eqref{redshift:K}, to the bootstrap estimates and the smallness assumptions on $\eta$ and $\delta$, we can take $a < 2K_+$ close to $2K_+$.
 Then, by integrating the above and using \eqref{redshift:bootstrap_phi}:
\begin{align*}
e^{av} |f|(u, v) &= \left| e^{a v_0} f(u, v_0) e^{\int_{v_0}^v \br{a- \kappa(2K)}(u, y)dy} - \int_{v_0}^v e^{\int_{v'}^v (a - \kappa(2K))(u, y)dy} e^{a v'}\partial_v \phi (u, v')dv'  \right| \\ 
& \lesssim |f|(u, v_0) + \int_{v_0}^v e^{a v'} |\partial_v \phi|(u, v')dv'  \lesssim |f|(u, v_0) + e^{(a-c_p(s))v}.
\end{align*}
Hence, \eqref{redshift:bootstrap_nu} and \eqref{assumption_nu} give:
\begin{equation} \label{redshift:pointwise_duphi1}
|\partial_u \phi|(u,v ) \lesssim e^{-av}\Omega^2(0, v) |\partial_u \phi|(u, v_0) +e^{-c_p(s)v}\Omega^2(0, v),
\end{equation}
when $\partial_u \phi, \partial_u \log\Omega^2 \in C^{\infty}_u([0, U])$. This implies:
\begin{equation} \label{redshift:L2_duphi}
\| \partial_u \phi(\cdot, v) \|_{L^2_u([0, u])}^2 \lesssim e^{-2av} \Omega^4(0, v)\| \partial_u \phi(\cdot, v_0) \|_{L^2_u([0, u])}^2 + u e^{-2c_p(s)v}\Omega^4(0, v).
\end{equation}
Moreover, we have the following $L^1$ bound. Let $p'$ be the conjugate exponent of $p$, i.e.\ $pp' = p+ p'$. By \eqref{redshift:pointwise_duphi1}, \eqref{first_redshift} and by H{\"o}lder's inequality:
\begin{align}
\int_0^u |\partial_u \phi|(u', v)du'& \lesssim  e^{-av}\Omega^2(0, v) \int_0^u |\partial_u \phi|(u', v_0) du' + e^{-c_p(s)v} \Omega^2(0, v) u \nonumber \\
&\le  e^{-av}\Omega^2(0, v) \|1\|_{L^{p'}_u([0, u])} \| \partial_u \phi(\cdot, v_0) \|_{L^p_u([0, u])} + \delta e^{-c_p(s)v } \label{first_holder} \\
&\le o_u(1)\, e^{-av} \Omega^2(0, v) u^{1-\frac{1}{p}} + \delta e^{-c_p(s)v}, \nonumber
\end{align}
due to our assumptions on the initial data, where $o_u(1)$ is referred to the limit $u \to 0$. Now, using that $u \lesssim e^{-2K_+ v}$, $\Omega^2(0, v) \sim e^{2K_+v}$, the definition of $c_p(s)$ and since $a$ can be chosen arbitrarily close to $2K_+$, we have 
\[
e^{-av} \Omega^2(0, v) u^{1-\frac{1}{p}} \lesssim e^{-c_p(s)v}e^{-(a-c_p(s)-2K_+ + 2K_+\br{1-\frac{1}{p}} )v} \lesssim e^{-(c_p(s)+\epsilon)v},
\]
for some small $\epsilon > 0$ and for a suitable choice of $C(\eta)$ in \eqref{def_cps}.
In particular, we proved that
\begin{equation} \label{redshift:L1_duphi}
\|\partial_u \phi(\cdot, v)\|_{L^1_u([0, u])} \lesssim \delta e^{-c_p(s)v}
\end{equation}
for $v$ large.
We now drop the regularity assumptions \eqref{redshift:temp_assumptions} on $\partial_u \log \Omega^2$ and $\partial_u \phi$: by density, estimates \eqref{redshift:L2_duphi} and \eqref{redshift:L1_duphi} still hold for $\partial_u \phi \in L^p_u([0, u])$.

Now, \eqref{appendix_u_varpi_eqn}, the fact that $|\lambda| \lesssim \delta$ (see \eqref{redshift:lambda_bdd}), that $\Omega^2(u, v)\sim \Omega^2(0, v)\sim e^{2K_+v}$, and \eqref{redshift:L2_duphi} give:
\begin{align}
|\varpi(u, v) - \varpi(0, v)| & \lesssim \int_0^u \left |   \frac{ r^2 (\partial_u \phi)^2}{\Omega^2} \lambda \right |(u', v) du' \nonumber \\
&\lesssim  \frac{\delta}{\Omega^2(0, v)} \|\partial_u \phi(\cdot, v)\|^2_{L^2_u([0, u])} \nonumber \\
&\lesssim  \delta o_1(U) e^{-(2a - 2K_+)v} + \delta^2 e^{-2c_p(s)v} \nonumber \\
&\lesssim \delta^2 e^{-c_p(s)v} \label{redshift:close_varpi}
\end{align}
where we used that $u \Omega^2(0, v) \lesssim \delta$ and that $a$ is chosen close to $2K_+$, hence, for some $\varepsilon = \varepsilon(\eta) > 0$ and for a suitable choice of $C(\eta)$ in  \eqref{def_cps}: $e^{-2av} \Omega^2(0, v) \lesssim e^{-(2K_+ -\varepsilon)v} \lesssim e^{-c_p(s)v}$. Thus, \eqref{EH:varpi} and \eqref{redshift:close_varpi} allow to close \eqref{redshift:bootstrap_varpi} for $v_1$ large. This result also entails that the estimate for $K$ can be improved so to obtain \eqref{redshift:K_bound}.

Similarly, we exploit \eqref{appendix_kappa2_eqn}, \eqref{redshift:L2_duphi}, the fact that $\kappa > 0$ and $\partial_u \kappa < 0$ for a.e.\ $u$ (see remark \ref{remark:signs} and \eqref{appendix_kappa_eqn}) and that $|\nu|(u, v) \sim \Omega^2(0, v)$ to show that
\[
 - \int_0^u  \frac{\partial_u \kappa}{\kappa}(u', v)du'  = \int_0^u \frac{r |\partial_u \phi|^2}{|\nu|}(u', v) du' \le \frac{\| \partial_u \phi(\cdot, v) \|^2_{L^2_u([0, u])}}{\Omega^2(0, v)} \lesssim  \delta e^{-c_p(s)v}.
\]
The bootstrap for $\kappa$ closes since the above implies\footnote{The fact that $\kappa \le 1$ follows from \eqref{assumption_kappa} and  \eqref{appendix_kappa_eqn}, which is a consequence of the integrated Raychaudhuri equation along the $u$ direction.}
\begin{equation} \label{redshift:close_kappa}
\exp \br{-o_1(v) } \le \kappa(u, v) \le 1,
\end{equation}
as $v \to +\infty$.
To close the bootstrap on $\partial_v \phi$, consider \eqref{appendix_rdvphi} and use \eqref{EH:phi_bound}, \eqref{redshift:lambda_bdd} and \eqref{redshift:L1_duphi} to obtain:
\begin{align*}
|r \partial_v \phi|(u, v) &= \left | (r \partial_v \phi)(0, v) - \int_0^u \br{\lambda \partial_u \phi } (u', v)du' \right | \\
&\le  C_{\mathcal{H}} r_+ e^{-sv}+\delta^2 e^{-c_p(s)v}.
\end{align*}
Now, for a small choice of $\delta$ we have $\frac{r_+}{r(u, v)} \le 2$ (see \eqref{def_Pdelta}) and
\[
| \partial_v \phi|(u, v) < \frac{5}{2}C_{\mathcal{H}} e^{-c_p(s)v}.
\]
An analogous estimate for $\phi$ follows if we exploit the fundamental theorem of calculus\footnote{That holds under the mere assumption of $L^1$-integrability of $\partial_u \phi$.} and \eqref{redshift:L1_duphi}. Indeed, we obtain:
\[
|\phi|(u, v) = \left | \phi(0, v) + \int_0^u \partial_u \phi(u', v) du' \right | < \frac{3}{2}C_{\mathcal{H}} e^{-c_p(s)v}.
\]
The bootstrap inequality for $\phi$ and $\partial_v \phi$ then closes.

To close the bootstrap on $|\nu|$, take the integral formulation of \eqref{wave_Omega} and use \eqref{redshift:bootstrap_phi}, \eqref{redshift:L1_duphi}, \eqref{redshift:lambda_bdd} and the fact that $|\nu| \sim \Omega^2(u, v) \sim \Omega^2(0, v)$ to obtain
\[
\left | \partial_v \log \frac{\Omega^2(u, v)}{\Omega^2(0, v)} \right | \lesssim e^{-c_p(s)v} \| \partial_u \phi(\cdot, v) \|_{L^1_u([0, u])} + u\Omega^2(0, v)\le  \delta e^{-2c_p(s)v}+ u \Omega^2(0, v).
\]
After integrating in $v$, we exploit \eqref{assumption_nu}, \eqref{assumption_kappa}, the fact that $\kappa(u, v_0) = 1 + o_u(1)$ and the integrability of $u \Omega^2(0, v)$ to get:
\[
\left| \log \frac{\Omega^2(u, v)}{\Omega^2(0, v)} + o_u(1) \right| \lesssim \delta.
\]
Since $\Omega^2=4|\nu|\kappa$ and $\kappa$ is close to 1, as previously seen, the bootstrap inequality for $\nu$ closes for $\delta$ small and $v$ large.\footnote{In the redshift region (and in the no-shift and early-blueshift regions as well), requiring $v$ large is equivalent to requiring $u$ small. See also \eqref{redshift:bound_u}.} In particular, we have $E = \mathcal{P}_{\delta}$.

Furthermore, \eqref{redshift:dvlogOmega} follows after using \eqref{appendix_log_eqn} (where the term containing $\partial_u \varpi$ is integrated by parts) together with \eqref{EH:logOmega}, \eqref{redshift:bootstrap_phi}, \eqref{redshift:L1_duphi}, the bounds for $\varpi$ and for $\kappa$.

Regarding the presence of trapped surfaces, it follows from \eqref{wave_r_alternative} and from the  the facts that $\nu < 0$, $\kappa \sim 1$ and $K \sim K_+$ in this region, that:
\[
\lambda(u_2, v)-\lambda(u_1, v) = \int_{u_2}^{u_1} \left [ \nu \kappa (2K ) \right ] (u', v) du' < 0,
\]
for every $(u_1, v), (u_2, v) \in \mathcal{P}_{\delta}$ such that $u_1 < u_2$. This will be exploited in the next section to show that, similarly to the smooth case, the apparent horizon is contained in the redshift region.

Finally, by reasoning as in the end of the proof of \cite[Proposition 4.7]{Rossetti}, it follows that $\emptyset \ne J^{-}(\Gamma_R) \cap \{v \ge v_1\} \subset \mathcal{P}_{\delta}$.
\end{proof}

\subsection{Apparent horizon} \label{section:AH}
By using the integral formulation of \eqref{wave_r}--\eqref{raych_v}, and the results in \cite[section 4.4]{Rossetti}, we have that the apparent horizon
\[
\mathcal{A} = \{ (u, v) \colon \lambda(u, v) = 0 \} = \{ (u_{\mathcal{A}}(v), v) \colon v \ge v_1 \}
\]
is a non-empty $C^0$ curve contained in $I^{-}(\Gamma_R)$. Moreover, for every $\tilde u$, the hypersurfaces $\{u = \tilde u\}$ can intersect $\mathcal{A}$ either at a single point or at a single outgoing null segment having finite length (in terms of the coordinate $v$). 

We stress that, differently from the smooth case, $\mathcal{A}$ is a $C^0$, rather than $C^1$, curve and the function $v \mapsto u_{\mathcal{A}}(v)$ is, in general, merely continuous.
This follows, using the monotonicity of $\lambda$ (see \eqref{R:monotonicity}), from a low-regularity version of the implicit function theorem.

\subsection{No-shift region}
\begin{proposition}[Propagation of estimates in the no-shift region] \label{prop:noshift}
    Let $\Delta > 0$ be small, compared to the initial data, and let $Y > 0$ be such that
    \[
    0 < Y - r_{-} < \Delta.
    \]
    Then, given $v_1 \ge v_0$ large, we have
    $J^+(\Gamma_R) \cap J^{-}(\Gamma_Y) \cap \{ v \ge v_1 \} \ne \emptyset$ and, for every $(u, v)$ in $J^+(\Gamma_R) \cap J^{-}(\Gamma_Y) \cap \{ v \ge v_1 \}$, the following inequalities hold:
    \begin{align}
        -C_{\mathcal{N}} \le \lambda(u, v) & \le -c_{\mathcal{N}} < 0, \\
        |\nu|(u, v) &\sim |\nu|(u, v_R(u)), \\
        |\varpi(u, v) - M| &\le C_{\mathcal{N}} e^{-2c_p(s)v}, \\
        |K(u, v) - K_+| + |K(u, v) + K_{-}| &\le C_{\mathcal{N}}, \\
        |\partial_v \log \Omega^2(u, v) - 2K(u, v)| &\le C_{\mathcal{N}} e^{-2c_p(s)v}, \\
        \| \partial_u \phi(\cdot, v) \|_{L^1_u([0, u])} &\le C_{\mathcal{N}} e^{-c_p(s)v}, \\
        |\phi|(u, v) + |\partial_v \phi|(u, v) &\le C_{\mathcal{N}} e^{-c_p(s)v}, \\
        u &\sim e^{-2K_+ v}, \label{noshift:u_relation} \\
        0 \le v - v_R(u) &\le C_{\mathcal{N}},   \label{noshift:finite_v}     
    \end{align}
    for $C_{\mathcal{N}}$, $c_{\mathcal{N}} > 0$ depending on the initial data and possibly depending on $\eta, R, Y$. The constant $K_{-} > 0$ is defined in \eqref{def_Kminus}.    
\end{proposition}
\begin{proof}
For every $l \in \mathbb{N}$, we define
\begin{equation} \label{def_Cl}
\begin{cases}
    C_l \coloneqq 4 C_{l-1}, \\
    C_0 \coloneqq 4 C_{\mathcal{H}}.
\end{cases}
\end{equation}
Let $N = N(N_{\text{i.d.}}, R, Y) \in \mathbb{N}$ large, and 
\[
    \epsilon = \epsilon(N) \coloneqq \frac{R -Y}{N} > 0.
\]
For $l=0, \ldots, N$, define
\[
\mathcal{N}_l \coloneqq J^+(\Gamma_R) \cap \{ v \ge v_1\} \cap \left \{ (u, v) \in \mathcal{P} \colon \, R-l\epsilon \le r(u, v) \le R-(l-1)\epsilon \right \}
\]
and
\[
    \mathcal{N} \coloneqq \bigcup_{l=0}^N \mathcal{N}_l.
\]
Notice that $\mathcal{N}_0 = \Gamma_R$. Moreover, we parametrize the constant-$r$ curves that constitute the boundary of each $\mathcal{N}_l$ as in the following example:
\[
\Gamma_{R-(l-1)\epsilon} = \{ (u_{l-1}(v), v)\colon v \ge v_1 \},
\]
for a suitable $C^1$ function $u_{l-1}$.
See also \cite{Rossetti} for more details and \cite{VdM1, Dafermos_2005_uniqueness}, where this technique was previously used.

We now prove by induction, for every $l \in \{0, \ldots, N\}$, that the following hold:
\begin{align}
|\phi|(u, v) + |\partial_v \phi|(u, v) & \le C_l e^{-c_p(s) v}, \label{induction1} \\
\kappa(u, v) &\ge 1-\Delta,  \label{induction2} \\
|\partial_v \log \Omega^2(u, v) - 2K(u, v)| &= o_v(1), \quad \text{ as } v \to +\infty, \label{induction3}
\end{align}
for every $(u, v) \in \mathcal{N}_l$.

The first step of the induction is proved due to the results of proposition \ref{prop:redshift}. 
At the $l$-th step of the induction, where $l \in \{1, \ldots, N\}$ fixed, we assume that inequalities \eqref{induction1}, \eqref{induction2}, \eqref{induction3} are satisfied in $\mathcal{N}_i$ for every $i = 0, \ldots, l-1$.

To prove the inductive step, we define $E_l$ as the set of points $q$ in $\mathcal{N}_l$ such that the following inequalities hold for every $(u, v) \in J^{-}(q) \cap \mathcal{N}_l$:
\begin{align}
|\phi|(u, v) + |\partial_v \phi|(u, v) & \le C_l e^{-c_p(s) v}, \label{bootstrap1} \\
\kappa(u, v) &\ge 1- \Delta \label{bootstrap2}, \\
|\varpi(u, v) - M| &\le \Delta. \label{bootstrap3}
\end{align}
From now on, let $(u, v) \in E_l$.
Notice that, by construction, we have
\[
\begin{cases}
    r_{-} < Y \le r \le R < r_+, \quad \text{ in } \mathcal{N}, \\
    0 \le R - r(u, v) \le R- Y, \\
    0 \le R- (l-1)\epsilon - r(u, v) \le \epsilon, \\
    0 \le r(u, v) - (R - l \epsilon) \le \epsilon.
\end{cases}
\]

Now, take $\Delta$ sufficiently small compared to the initial data, so that we can use \cite[lemma 4.5]{Rossetti} and obtain
\begin{equation} \label{bound_lambda_noshift}
-C(N_{\text{i.d.}}) \le \lambda(u, v) \le -C(N_{\text{i.d.}}, R, Y) < 0
\end{equation}
and similarly for $1 - \mu = \frac{\lambda}{\kappa}$, due to \eqref{bootstrap2} and to the fact that $\kappa \le 1$ (as already noticed in the case of the redshift region).
By integrating \eqref{bound_lambda_noshift}, we obtain
\[
0 \le v - v_R(u) \le \frac{R - r(u, v)}{C(N_{\text{i.d.}}, R, Y)},
\]
which, together with the bounds \eqref{redshift:bound_u} obtained in the redshift region, implies $u \sim e^{-2K_+v}$.

To get an estimate on $\nu$, we notice that \eqref{appendix_lognumu_eqn} gives
\[
|\nu|(u, v) = |\nu|(u, v_R(u)) \frac{|1-\mu|(u, v)}{|1-\mu|(u, v_R(u))} e^{\int_{v_R(u)}^v \frac{r}{\lambda}(\partial_v \phi)^2(u, v') dv' }. 
\]
Since we saw that $|1-\mu|$ and $|\lambda|$ are bounded and bounded away from zero in this region, and due to \eqref{bootstrap1}, we have
\begin{equation} \label{nu_bound_noshift}
|\nu|(u, v) \sim |\nu|(u, v_R(u)).
\end{equation}

Now, we assume, for the moment, that $\partial_u \phi$ is smooth. The pointwise bound  \eqref{redshift:pointwise_duphi1} for the redshift region gives
\begin{equation} \label{noshift:pointwise_duphi_old}
    |\partial_u \phi|(u, v_R(u)) \lesssim e^{-a v_R(u)} \Omega^2(0, v_R(u)) |\partial_u \phi|(u, v_0) + e^{-c_p(s)v_R(u)} \Omega^2(0, v_R(u)),
\end{equation}
where we recall that $0 < a < 2K_+$ can be chosen close to $2K_+$ (the error term depending on the small parameter $\eta$ of proposition \ref{prop:redshift}). Moreover, in the smooth case, the wave equation for $\phi$ can be cast as
\[
    \partial_v\br{r \partial_u \phi} = -\nu \partial_v\phi.
\]
Using \eqref{bootstrap1}, \eqref{nu_bound_noshift}, $v_R(u) \sim v$ and that $|\nu|(u, v) \sim$ $\Omega^2(u, v) \sim \Omega^2(0, v) $ $\sim e^{2K_+ v}$, the two relations above give
\begin{align} 
    |r \partial_u \phi|(u, v) &\le |r \partial_u \phi|(u, v_R(u))  + \int_{v_R(u)}^v |\nu \partial_v \phi|(u, v')dv' \nonumber \\
    &\le C e^{-av} \Omega^2(0, v) |\partial_u \phi|(u, v_0) + C C_l e^{-c_p(s)v} \Omega^2(u, v). \label{noshift:pointwise_duphi}
\end{align}
If we integrate in $[u_{l-1}(v), u]$ and use \eqref{bootstrap2} and the bounds on $r$:
\begin{align*}
    \int_{u_{l-1}(v)}^u |\partial_u \phi|(u', v) du' &\lesssim e^{-av} \Omega^2(0, v) \int_0^u |\partial_u \phi|(u', v_0) du' + e^{-c_p(s)v}\int_{u_{l-1}(v)}^u |\nu|(u', v) du' \\
    &\le  o_u(1) u^{1 - \frac{1}{p}} e^{-(a-2K_+)v} + \epsilon e^{-c_p(s)v}, 
\end{align*}
where we used H{\"o}lder's inequality, similarly to \eqref{first_holder}. The computations immediately after \eqref{first_holder}, together with the bounds on $u$ obtained for the current region, also allow to bound the first term by $\epsilon e^{-(c_p(s))v}$. Thus, if $v_1$ is sufficiently large:
\begin{equation} \label{noshift:duphi_L1small}
\|\partial_u \phi(\cdot, v)\|_{L^1_u([u_{l-1}(v), u])} \lesssim \epsilon e^{-c_p(s)v}.
\end{equation}
A similar computation shows that, whenever we integrate in the entire interval $[0, u]$:
\[
\| \partial_u \phi \|_{L^1_u([0, u])} \lesssim (\delta + N\epsilon) e^{-c_p(s)v} \lesssim e^{-c_p(s)v}.
\]
Moreover, analogous computations show that
\begin{equation} \label{noshift:duphi_L2small}
\|\partial_u \phi(\cdot, v) \|^2_{L^2_u([u_{l-1}(v), u])} \lesssim e^{-2av}\Omega^4(0, v)\|\partial_u \phi(\cdot, v_0)\|^2_{L^2_u([0, u])} + \epsilon \Omega^4(0, v)e^{-2c_p(s)v}
\end{equation}
and
\[
\|\partial_u \phi(\cdot, v) \|^2_{L^2_u([0, u])} \lesssim e^{-2av}\Omega^4(0, v)\|\partial_u \phi(\cdot, v_0)\|^2_{L^2_u([0, u])} + u \Omega^4(0, v)e^{-2c_p(s)v}.
\]
 We now drop the smoothness assumption on $\partial_u \phi$. Similarly to the case of the redshift region, the above $L^p$ bounds hold in our low-regularity setting as well, by means of an approximation argument.

Moreover, equation \eqref{appendix_u_varpi_eqn} and the above estimates allow to close the bootstrap inequality for $\varpi$ (see the analogous computations in \eqref{redshift:close_varpi}). In particular:
\[
|\varpi(u, v) - M| \lesssim e^{-2c_p(s)v}.
\]
The bootstrap on $\kappa$ is also closed in the same way seen for the redshift region, now using estimate
 \eqref{redshift:close_kappa} as well. 

We now complete the bootstrap procedure by analysing $\partial_v \phi$ and $\phi$. Regarding the former, \eqref{appendix_rdvphi}, \eqref{induction1}, \eqref{def_Cl} and the bounds on $r$ and $\lambda$ yield:
\begin{align}
    |r \partial_v \phi|(u, v) &= \left | \br{r \partial_v \phi}(u_{l-1}(v), v) - \int_{u_{l-1}(v)}^u \br{\lambda \partial_u \phi}(u', v) du' \right |  \label{noshift_close_dvphi} \\
    & \le \frac{C_l}{4}\br{R - (l-1)\epsilon} e^{-c_p(s)v} + C \|\partial_u \phi (\cdot, v)\|_{L^1_u([u_{l-1}(v), u])}. \nonumber
\end{align}
If we divide both sides by $r(u, v)$ and use \eqref{noshift:duphi_L1small} and
\[
\frac{R-(l-1)\epsilon}{R - l \epsilon} < 1 + \frac{\epsilon}{Y} < 2,
\]
we get
\[
|\partial_v \phi|(u, v) < \frac58 C_l e^{-c_p(s)v},
\]
provided that $\epsilon$ is sufficiently small.
We proceed similarly for $\phi$, using the fundamental theorem of calculus:
\begin{align*}
    |\phi|(u, v) &\le |\phi|(u_{l-1}(v), v) + \| \partial_u \phi(\cdot, v) \|_{L^1_u([u_{l-1}(v), u]))} \\
    &\le \frac{C_l}{4} e^{-c_p(s)v} + C \epsilon e^{-c_p(s)v} \\
    &< \frac38 C_l e^{-c_p(s)v},
\end{align*}
provided that $\epsilon$ is sufficiently small. This closes the bootstrap for \eqref{bootstrap1}.

All bootstrap inequalities were closed, so we have $E_l = \mathcal{N}_l$. To complete the inductive step, we notice that \eqref{induction3} is proved by following the same procedure seen at the end of proposition \ref{prop:redshift} (see also the end of \cite[Proposition 4.12]{Rossetti}).
\end{proof}

\subsection{Early-blueshift region} \label{prop:early_blueshift}
\begin{proposition}[Propagation of estimates in the early-blueshift region] \label{proposition:early_blueshift}
Let $\beta > 0$ be small compared to the initial data. Define the curve
\begin{equation} \label{def_gamma}
\gamma \coloneqq \left \{  (u, v_{\gamma}(u))\colon\, u \in [0, U] \right \},
\end{equation}
where $v_{\gamma}(u) = (1+\beta)v_Y(u)$.
Then, for every $\varepsilon = \varepsilon(N_{\text{i.d.}}) > 0$, there exist $\Delta \in (0, \varepsilon)$ and $Y > 0$ such that 
\[
0 < Y - r_{-} < \Delta
\]
and such that the following inequalities hold for every $(u, v) \in $ $J^{+}(\Gamma_Y) \cap J^{-}(\gamma)$ $\cap \{v \ge v_1\}$ $\ne \emptyset$:
\begin{align}
    r(u, v) &\ge r_{-} - 2 \varepsilon, \\
    |\varpi(u, v) - M| &\le C e^{-2(c_p(s) - \tau)v}, \label{earlyblueshift:varpi} \\
    |K(u, v) + K_{-}| &\le C \varepsilon, \label{earlyblueshift:K}\\
    |\partial_v \log \Omega^2(u ,v) - 2K(u, v)| & \le C e^{-(c_p(s) - \tau)v}, \label{earlyblueshift:dvlog} \\
    \|\partial_u \phi(\cdot, v)\|_{L^1_u([0, u])} & \le C  e^{-(c_p(s) - \tau)v}, \\
    |\phi|(u, v) + |\partial_v \phi|(u, v) &\le C  e^{-(c_p(s) - \tau)v}, \label{earlyblueshift:phi}
\end{align}
for some $v_1 \ge v_0$ large, where $\tau > 0$, $\tau = O(\beta)$ and where $C > 0$ depends on the initial data and possibly on the constants $\eta, R, Y$ defined in the previous regions. 
\end{proposition}
\begin{proof}
For every $(u, v) \in J^+(\Gamma_Y) \cap J^{-}(\gamma)$, define:
\begin{equation} \label{def_Cs}
\mathcal{C}_s(u, v) \coloneqq c_p(s)v_Y(u)-2(2K_{-}+\varepsilon)(v - v_Y(u)).
\end{equation}
Notice that 
\[
\mathcal{C}_s > 0, \quad \text{ in }  J^+(\Gamma_Y) \cap J^{-}(\gamma).
\]
Indeed, we have\footnote{The curve $\gamma$ is  spacelike even though it is not a curve of constant area-radius. See, for instance, \cite[Section 4.6]{Rossetti}.} 
\begin{equation} \label{EB:v_relation}
v_Y(u) \le v \le (1+\beta)v_Y(u)
\end{equation}
in this region, whence
\[
\mathcal{C}_s(u, v)  \ge \br{c_p(s)-2(2K_{-}  + \varepsilon)\beta}v_Y(u) > 0,
\]
provided that $\beta$ is small. We stress that $\mathcal{C}_s$ is not necessarily positive in other regions.

Moreover, the fact that $v_Y'(u) < 0$ implies:\footnote{In fact, this follows after differentiating the relation $r(u, v_Y(u)) = Y$. See also \cite[Remark 4.8]{Rossetti} and recall that $r \in C^1$ in our case as well.}
\begin{equation} \label{earlyblueshift:Cs_monotonic}
\partial_u \mathcal{C}_s < 0
\end{equation}
in the early-blueshift region.

Now, let $E$ be the set of points $q$ in $\mathcal{D} \coloneqq J^+(\Gamma_Y) \cap J^{-}(\gamma) \cap \{ v \ge v_1\}$ such that the following inequalities hold for every $(u, v) \in J^{-}(q) \cap \mathcal{D}$:
\begin{align*}
r(u, v) &\ge r_{-}-2\varepsilon, \\
|\phi|(u, v) + |\partial_v \phi|(u, v) & \le 4 C_{\mathcal{N}} e^{-\mathcal{C}_s(u, v)}, \\
\kappa(u, v) &\ge 1-\varepsilon \\
|\varpi(u, v) - M| &\le \varepsilon.
\end{align*}
The above inequalities hold along $\Gamma_Y$, due to the results of proposition \ref{prop:noshift} and since
\[
\mathcal{C}_s(u, v_Y(u)) = c_p(s) v_Y(u), \quad \forall\, u \in [0, u_Y(v_1)].
\]
Now, let $(u, v) \in E$. By construction:
\[
0 < \frac{r_{-}}{2} < r_{-} - 2 \varepsilon \le r(u, v) \le Y < r_+,
\]
assuming that $\varepsilon$ is sufficiently small. Moreover, equation \eqref{appendix_rlambda_eqn}, the fact that $\Omega^2=4|\nu| \kappa$, the bootstrap assumption on $\kappa$ and the results obtained for the previous region imply that
\[
|\lambda|(u, v) \le C(N_{\text{i.d.}}).
\]
Similarly to the procedure followed in the previous regions, the bootstrap on $\varpi$ and the assumptions on $r$ imply that $K$ is $\varepsilon$-close to $-K_{-}$.

Before continuing, we derive some preliminary bounds that will be exploited during the rest of the proof. Let us momentarily assume that (at least) $\nu$ and $\partial_u \phi$ are smooth, analogously to the previous regions. Then, \cite[(4.126)]{Rossetti} and following imply that
\begin{align}
|\nu|(u, v) &= |\nu|(u, v_Y(u)) \exp \br{\int_{v_Y(u)}^v \left [ \kappa (2K) \right ] (u, v') dv'} \nonumber \\
&= |\nu|(u, v_Y(u)) e^{-2K_{-}\br{1 + o(1)}(v - v_Y(u))}, \label{earlyblueshift:prelbound2}
\end{align}
as $v_1 \to +\infty$, 
and
\begin{equation} \label{nu_smoothness_monot}
    \partial_v\br{|\nu|(u, v) e^{-\mathcal{C}_s(u, v)}} \ge (2K_{-} + C\varepsilon)|\nu|(u, v) e^{-\mathcal{C}_s(u, v)} > 0,
\end{equation}
hence implying
\begin{equation} \label{earlyblueshift:prelbound1}
    \int_{v_Y(u)}^v |\nu|(u, v') e^{-\mathcal{C}_s(u, v')} dv' \le \frac{1}{C} |\nu|(u, v) e^{-\mathcal{C}_s(u, v)},
\end{equation}
for some $C=C(N_{\text{i.d.}}) > 0$.

Now, the pointwise estimate in \eqref{noshift:pointwise_duphi} (obtained under smoothness assumptions), \eqref{def_Cs},  \eqref{nu_smoothness_monot}, the facts that $\kappa \sim 1$ in the no-shift region and that
\[
e^{-c_p(s) v_Y(u)} \Omega^2(u, v_Y(u)) \lesssim \br{|\nu|(u, v) e^{-\mathcal{C}_s(u, v)}}_{|v = v_Y(u)}
\]
entail:
\begin{align*}
|r \partial_u \phi|(u, v_Y(u)) &\le C e^{-(a- 2K_+) v_Y(u)} |\partial_u \phi|(u, v_0) + C e^{-c_p(s) v_Y(u)} \Omega^2(u, v_Y(u)) \\
&\lesssim e^{-(a- 2K_+) v_Y(u)} |\partial_u \phi|(u, v_0) + |\nu|(u, v) e^{-\mathcal{C}_s(u, v)}.
\end{align*}
As a result of the above, we have that \eqref{appendix_rduphi}, the bootstrap inequality for $\partial_v \phi$ and the relation $\Omega^2(u, v_Y(u)) \sim \Omega^2(0, v_Y(u)) \sim  e^{2K_+ v_Y(u)}$ give:
\begin{align}
|r \partial_u \phi|(u, v) &\lesssim |r \partial_u \phi|(u, v_Y(u)) + \int_{v_Y(u)}^v | \nu \partial_v \phi |(u, v')dv' \nonumber \\
& \lesssim  e^{-a v_Y(u)} \Omega^2(0, v_Y(u)) |\partial_u \phi|(u, v_0) + |\nu|(u, v) e^{-\mathcal{C}_s(u, v)} + \int_{v_Y(u)}^v |\nu|(u, v') e^{-\mathcal{C}_s(u, v')}dv' \nonumber \\ 
& \lesssim e^{-a v_Y(u)} \Omega^2(u, v_Y(u)) |\partial_u \phi|(u, v_0) + |\nu|(u, v) e^{-\mathcal{C}_s(u, v)}, \label{pointwise_duphi_last}
\end{align}
where in the last step we used \eqref{earlyblueshift:prelbound1}.

After integrating \eqref{pointwise_duphi_last} and using \eqref{earlyblueshift:Cs_monotonic} and the bounds on $r$ (notice that the first term at the right hand side can be bounded as seen, e.g.\ in \eqref{first_holder}, since $\Omega^2(u, v_Y(u)) \sim \Omega^2(0, v_Y(u)) \sim  e^{2K_+ v_Y(u)}$):
\begin{align} \label{early_blueshift:L1_bound_proof}
\| \partial_u \phi(\cdot, v) \|_{L^1_u([u_Y(v), u])} &\le e^{(2K_+ - a)v} \int_{u_Y(v)}^u  |\partial_u \phi|(u', v_0) du' + e^{- \mathcal{C}_s(u, v)} \int_{u_Y(v)}^u |\nu|(u', v)du' \nonumber \\
&\lesssim o_u(1) e^{(2K_+ - a)v} u^{1 - \frac{1}{p}} + \varepsilon e^{-\mathcal{C}_s(u, v)}   \nonumber \\
&\lesssim \varepsilon e^{-\mathcal{C}_s(u, v)},
\end{align}
where in the last step we used that $u \lesssim e^{-2K_+ v_Y(u)} \le e^{-2K_+\frac{v}{1+\beta}}$ (see \eqref{noshift:u_relation} and \eqref{EB:v_relation}) and required $\beta$ to be small.

The  pointwise estimate in \eqref{pointwise_duphi_last} can be cast in a more convenient form: since $\Omega^2(u, v_Y(u)) \sim |\nu|(u, v_Y(u))$ and due to \eqref{earlyblueshift:prelbound2} and \eqref{def_Cs} we have
\begin{align*}
e^{-av_Y(u)} \Omega^2(u, v_Y(u)) &\lesssim e^{-a v_Y(u)} |\nu|(u, v) e^{2K_{-}(1 + o(1))(v - v_Y(u))} e^{-c_p(s) v_Y(u)} e^{c_p(s) v_Y(u)} \\
&\lesssim |\nu|(u, v) e^{-\mathcal{C}_s(u, v)} e^{-(a - c_p(s))v_Y(u)}.
\end{align*}
The above can be used to express \eqref{pointwise_duphi_last} as
\begin{equation}
    |r \partial_u \phi|(u, v) \lesssim  |\nu|(u, v) e^{-\mathcal{C}_s(u, v)} \br{1 + |\partial_u \phi|(u, v_0) e^{-(a - c_p(s))v_Y(u)} }.
\end{equation}
The latter, \eqref{earlyblueshift:Cs_monotonic}, \eqref{earlyblueshift:prelbound2} and the bootstrap inequality for $\kappa$ give:
\begin{align*}
    & \left \| \frac{(\partial_u \phi)^2}{\Omega^2}(\cdot, v)  \right \|_{L^1_u([u_Y(v), u])} \lesssim \int_{u_Y(v)}^u |\nu|(u', v) e^{-2\mathcal{C}_s(u', v)}\br{1 + |\partial_u \phi|^2(u', v_0) e^{-2(a - c_p(s))v_Y(u')}} du' \\
    &\lesssim  e^{-2\mathcal{C}_s(u, v)} \int_{u_Y(v)}^u |\nu|(u', v) du' + \\
    &+\int_{u_Y(v)}^u |\partial_u \phi|^2(u', v_0) e^{2K_+ v_Y(u')} e^{-2K_{-}(1 + o(1))(v - v_Y(u'))} e^{-2\mathcal{C}_s(u', v)-2(a-c_p(s)) v_Y(u')} du'.
\end{align*}
The first term is smaller than $\varepsilon e^{-2\mathcal{C}_s(u, v)}$ due to the bounds on $r$. Due to \eqref{def_Cs}, the fact that $a < 2K_+$ can be chosen $\eta$-close to $2K_+$ and by H{\"o}lder's inequality, the second term is smaller than
\begin{align}
    &\int_{u_Y(v)}^u |\partial_u \phi|^2(u', v_0) e^{-(2a - 2K_+)v_Y(u')}e^{3(2K_{-}+o_{v_1, \varepsilon}(1))(v - v_Y(u'))} du' \lesssim \nonumber \\
    &\lesssim e^{3(2K_{-} + o_{v_1, \varepsilon}(1))(v - v_Y(u))}e^{-(2K_+ +o_{\eta}(1))v_Y(u)} \int_{u_Y(v)}^u |\partial_u \phi |^2(u', v_0)du' \nonumber \\
   &  \lesssim e^{3(2K_{-} + o(1))(v - v_Y(u))}e^{-(2K_+ +o(1))v_Y(u)} u^{1-\frac{2}{p}} \|\partial_u \phi(\cdot, v_0)\|_{L^p([0, u])}^2. \label{eb_long_bound}
\end{align}
We recall that $u \lesssim e^{-2K_+ v_Y(u)}$, so by using the above, \eqref{def_cps} and \eqref{def_Cs} again, we have that \eqref{eb_long_bound} is smaller than
\begin{align*}
   &e^{-2 \mathcal{C}_s(u, v)+2\mathcal{C}_s(u, v)} e^{3(2K_{-} + o(1))(v - v_Y(u))}e^{-(2K_+ +o(1))v_Y(u)} u^{1-\frac{2}{p}} \|\partial_u \phi(\cdot, v_0)\|_{L^p([0, u])}^2 \\
   \lesssim & \, e^{-2 \mathcal{C}_s(u, v)} e^{-(2K_{-} + o_{v_1, \varepsilon}(1))(v - v_Y(u))} o_u(1),
\end{align*}
where we exploited the assumptions on the initial data.
So, we conclude that
\begin{equation} \label{early_blueshift:additional_L1_bound}
    \left \| \frac{(\partial_u \phi)^2}{\Omega^2}(\cdot, v)  \right \|_{L^1_u([u_Y(v), u])} \lesssim \varepsilon e^{-2 \mathcal{C}_s(u, v)}.
\end{equation}
From now on, we drop the smoothness assumption and notice that  estimates \eqref{early_blueshift:L1_bound_proof} and \eqref{early_blueshift:additional_L1_bound} still hold in the rough setting, by density.

To close the bootstrap inequality for the Hawking mass we exploit  \eqref{appendix_u_varpi_eqn},  \eqref{early_blueshift:additional_L1_bound} and the previous bounds on $\lambda$ and $r$:
\[
|\varpi (u, v) - \varpi(u_Y(v), v)| \lesssim \int_{u_Y(v)}^u \left |  \lambda r^2 \frac{(\partial_u \phi)^2}{\Omega^2} \right | (u', v) du' \lesssim \varepsilon e^{-2 \mathcal{C}_s(u, v)}.
\]
The bootstrap on $r$ closes after noticing that \eqref{def_mu} and the fact that\footnote{See also section \ref{section:AH} where the apparent horizon was localized in the redshift region.} $(1 - \mu)(r(u, v), \varpi(u, v)) < 0$ give
\[
\br{1-\mu}(r(u, v), M) < \frac{2}{r}\br{\varpi(u, v) - M}  \lesssim e^{-2 \mathcal{C}_s(u, v)}.
\]
In fact, this provides a constraint on $r$, depending on $v$, that can be extracted from the plot of $r \mapsto (1 - \mu)(r, M)$ (see e.g.\ \cite[section 2]{CGNS2}).
The bootstrap on $\kappa$ can be closed analogously to the procedure of the previous regions.

As regards the estimate on $\phi$: the fundamental theorem of calculus, the results of proposition \ref{prop:noshift}, \eqref{early_blueshift:L1_bound_proof} and the fact that $e^{-c_p(s)v} \le e^{-\mathcal{C}_s(u, v)}$ yield
\begin{align*}
|\phi|(u, v) &\le C_{\mathcal{N}}e^{-c_p(s)v_Y(u)} + \int_{u_Y(v)}^u |\partial_u \phi| (u', v) du' \\
&\le C_{\mathcal{N}} e^{-\mathcal{C}_s(u, v)} + \varepsilon e^{-\mathcal{C}_s(u, v)}.
\end{align*}
A similar bound follows for $\partial_v \phi$ (see also the procedure in \eqref{noshift_close_dvphi}). Then, the bootstrap estimate for $\phi$ and $\partial_v \phi$ is closed for small choices of $\varepsilon$.

Inequality \eqref{earlyblueshift:K} is a consequence of the bounds on $r$ and $\varpi$. Moreover, the estimate for $\partial_v \log \Omega^2$ follows analogously to the case of the redshift and no-shift regions.

Finally, we note that, since $v_Y(u) \le v \le (1+\beta)v_Y(u)$ in the early-blueshift region, then
\begin{equation} \label{bound_on_Cs}
\mathcal{C}_s(u, v) \ge \frac{c_p(s)}{1+\beta}v - 2(2K_{-} +  \varepsilon)\frac{\beta}{1+\beta} v \ge (c_p(s) -  \tau)v,
\end{equation}
for some $ \tau > 0$, $ \tau  =O(\beta)$. This is the final expression we use for the estimates we obtained.
\end{proof}

\begin{remark}
    As occurs in the smooth case \cite{Rossetti}, the results of proposition \ref{prop:early_blueshift} still hold even if we take $\Delta > \varepsilon$, provided that $\Delta + \varepsilon$ is small compared to the initial data. 
\end{remark}

We now improve the bounds on $\lambda$ and $\nu$ along the future boundary of the early-blueshift region, as follows.
\begin{proposition}[Estimates along $\gamma$] \label{prop:along_gamma}
    Let $\beta > 0$ be the parameter that defines the curve $\gamma$. Then, for every $\varepsilon > 0$ small compared to the initial data and for every $(u, v) \in \gamma \cap \{v \ge v_1\}$, we have:
    \begin{align}
        e^{-(2K_{-}+\varepsilon)\frac{\beta}{1+\beta}v} \lesssim - \lambda(u, v) &\lesssim e^{-(2K_{-}-\varepsilon) \frac{\beta}{1+\beta}v}, \label{bound_lambda_gamma} \\
        u^{-1 + \frac{K_{-}}{K_+}\beta + \tilde{\alpha}} \lesssim -\nu(u, v) &\lesssim u^{-1 + \frac{K_{-}}{K_+}\beta - \alpha}, \label{bound_nu_gamma}
    \end{align}
    for $\alpha = \alpha(\beta, \varepsilon) > 0$, $\tilde \alpha = \tilde \alpha(\beta, \varepsilon) > 0$, and where 
    \[
    \frac{K_{-}}{K_+} \beta - \alpha(\beta, \varepsilon) > 0.
    \]
\end{proposition}
\begin{proof}
    The proof follows from that of \cite[Lemma 4.16]{Rossetti}, where now the integral formulation of \eqref{wave_r}--\eqref{raych_v} is adopted. In particular, equations \eqref{appendix_lognumu_eqn}, \eqref{appendix_nu_modified} and \eqref{appendix_lambda_modified} are used in place of the respective smooth formulations. 
\end{proof}

\subsection{Late-blueshift region} 
\begin{proposition}[Estimates in the late-blueshift region] \label{prop:late_blueshift}
    Let $U$ and $\beta$ be small compared to the initial data and choose $\varepsilon$ small so that
    \begin{equation} \label{assumpt_betaalpha}
        \frac{K_{-}}{K_+} \beta  -\alpha (\beta, \varepsilon) > 0.
    \end{equation}
    Let $p \in [2, +\infty)$.
    Then, for every $(u, v) \in J^{+}(\gamma) \cap \{ r(u, v) \ge r_{-} - \varepsilon\} \ne \emptyset$, we have:
    \begin{align}
        0 < -\lambda(u, v) &\lesssim \frac{\Omega^2(u, v)}{\Omega^2(u, v_{\gamma}(u))} |\lambda|(u, v_{\gamma}(u)) + e^{-\min \{2c_p(s)-2\tau, 2K_{-}-\tau\}v}, \label{blueshift:lambda} \\
        0 < -\nu(u, v) &\lesssim u^{-1 + \frac{K_{-}}{K_+}\beta - \alpha}, \\
        |\partial_v \phi|(u, v) &\lesssim  e^{-(c_p(s) - \tau)v}, \\
        |\phi|(u, v) &\lesssim e^{-(c_p(s) - \tau)v_{\gamma}(u)}, \\
        \| \partial_u \phi(\cdot, v) \|_{L^p_u([u_{\gamma}(v), u])} &\lesssim u_{\gamma}(v)^{-C \eta} \| \partial_u \phi(\cdot, v_0) \|_{L^p_u([0, u])} + u_{\gamma}(v)^{-\frac{1}{2K_+}\br{ \frac{2K_+}{p'} - c_p(s) } + O(\beta)}, \label{blueshift:Lp_duphi_final}  \\
        \partial_v \log \Omega^2(u, v) &=-2K_{-} + O(\varepsilon), \label{blueshift:dvlog} 
    \end{align}
    where $\tau > 0$, $\tau = O(\beta)$, $C > 0$. Here $\eta > 0$ (which is the constant introduced in proposition \ref{prop:redshift}) is taken small compared to $\beta$ and  $p'$ is the exponent conjugate to $p$.
\end{proposition}
\begin{remark}
    The $L^p$ estimate in \eqref{blueshift:Lp_duphi_final} gives a bound on the energy of $\partial_u \phi$ for large values of $v$ and fixed  values of $u$. This is the region near the Cauchy horizon we will mainly focus on. A sharper estimate in the region $u_{\gamma}(v) \sim u$ easily follows from our proof, as well.
\end{remark}
\begin{proof}[Proof of proposition \ref{prop:late_blueshift}]
During the proof, we will take $\varepsilon > 0$ small so that
\[
\frac{r_{-}}{2} < r_{-} - \varepsilon \le r \le Y, \quad \text{ in } J^+(\Gamma_Y) \cap \{ r \ge r_{-} - \varepsilon, v \ge v_1 \},
\]
and we will repeatedly employ the relation
\[
u \sim e^{-2K_+ v_Y(u)} = e^{-\frac{2K_+}{1+\beta}v_{\gamma}(u)}, \quad \forall\, u \in [0, U],
\]
which follows from \eqref{noshift:u_relation}. 
Furthermore, let $E$ be the set of points $q \in \mathcal{D} \coloneqq J^+(\gamma) \cap \{r \ge r_{-} - \varepsilon, v \ge v_1\}$ such that the following inequalities are satisfied for every $(u, v) \in J^{-}(q) \cap \mathcal{D}$:
\begin{align}
    |r \partial_v \phi|(u, v) &\le M  e^{-(c_p(s) - \tau)v}, \label{blueshift:B1} \\
    |\lambda|(u, v) &\le C(N_{\text{i.d.}}), \label{blueshift:B2}\\
    \partial_v \log \Omega^2(u, v) &\le -K_{-}, \label{blueshift:B3}
\end{align}
where $M=M(N_{\text{i.d.}}) > 0$ is a large constant and $\tau > 0$, $\tau = O(\beta)$ (see the definition of the curve $\gamma$ in proposition \ref{prop:early_blueshift}). These estimates are satisfied along the curve $\gamma$, due to the results of proposition \ref{prop:early_blueshift}.

Let $(u, v) \in E$. By integrating \eqref{blueshift:B1} and using \eqref{earlyblueshift:phi}:
\[
    |\phi|(u, v) \le |\phi|(u, v_{\gamma}(u)) + C(N_{\text{i.d.}}, M) \int_{v_{\gamma}(u)}^v e^{-(c_p(s)-\tau)v'}dv' \le C e^{-(c_p(s) - \tau)v_{\gamma}(u)},
\]
where $C = C(N_{\text{i.d.}}, M, \eta, R, Y)$. Since $\eta, R$ and $Y$ will not be affected in the rest of the proof, we have $C = C(N_{\text{i.d.}}, M)$.

Now, \eqref{appendix_rnu_eqn}, the bounds on $r$ and on $\phi$, \eqref{bound_nu_gamma}, \eqref{blueshift:B3} and the fact that $\kappa_{|\gamma} \sim 1$ give:
\begin{align}
    |r \nu|(u, v) &\lesssim |r \nu|(u, v_{\gamma}(u)) + \int_{v_{\gamma}(u)}^v \Omega^2(u, v')dv' \nonumber\\
    &\lesssim |r \nu|(u, v_{\gamma}(u)) -\frac{1}{K_{-}} \int_{v_{\gamma}(u)}^v \br{\Omega^2 \partial_v \log \Omega^2}(u, v') dv' \nonumber \\
    &\lesssim u^{-1 + \frac{K_{-}}{K_+}\beta - \alpha} + C(N_{\text{i.d.}}) \Omega^2(u, v_{\gamma}(u)) \nonumber \\
    &\le C(N_{\text{i.d.}}) u^{-1 + \frac{K_{-}}{K_+}\beta - \alpha}. \label{blueshift:nu_bound}
\end{align}

Once again (see the proofs of propositions \ref{prop:redshift}, \ref{prop:noshift} and \ref{prop:early_blueshift}) we assume, for the moment, that $\partial_u \phi$ is smooth.
The pointwise bound \eqref{pointwise_duphi_last} (see also \eqref{bound_on_Cs}) can be cast as
\[
|r \partial_u \phi|(u, v_{\gamma}(u)) \lesssim e^{-a v_Y(u)} \Omega^2(u, v_Y(u)) |\partial_u \phi|(u, v_0) + |\nu|(u, v_{\gamma}(u)) e^{-(c_p(s) - \tau)v_{\gamma}(u)},
\]
where we recall that $a$, with $0 < a < 2K_+$, can be chosen $\eta$-close to $2K_+$ (see \eqref{preliminary_Duphi} and following).
We now exploit the above, \eqref{appendix_rduphi}, \eqref{blueshift:B1} and \eqref{blueshift:nu_bound} to write:
\begin{align}
    |r \partial_u \phi|(u, v) &\lesssim |r \partial_u \phi|(u, v_{\gamma}(u)) + \int_{v_{\gamma}(u)}^v |\nu \partial_v \phi|(u, v')dv'  \label{lb:duphi_pointwise} \\
    &\le C(N_{\text{i.d.}})e^{-av_Y(u)}\Omega^2(u, v_Y(u))|\partial_u \phi|(u, v_0) + C(N_{\text{i.d.}}, M)u^{-1 + \frac{K_{-}}{K_+}\beta - \alpha} e^{-(c_p(s) - \tau)v_{\gamma}(u)}. \nonumber
\end{align}
We now obtain an $L^p$ bound, for $p \in [2, +\infty)$. For the sake of convenience, we define the following constants:
\begin{align}
    \tilde{A} &\coloneqq  1  - \frac{K_{-}}{K_+}\beta + \alpha - \frac{c_p(s) - \tau}{2K_+}(1+\beta), \label{def_tildeA} \\
    A &\coloneqq \tilde A - \frac{1}{p} > 0, \label{def_A}\\
    B &\coloneqq \frac{2K_+ - a}{2K_+} > 0, \label{def_B}
\end{align}
where $a < 2K_+$ can be chosen $\eta$-close to $2K_+$. We stress that $-p \tilde{A} < -1$ for $\beta$ small, due to \eqref{def_cps} and to the fact that $\tilde{A}=1 - \frac{c_p(s)}{2K_+} + O(\beta)$. Then, using the above definitions, Minkowski's inequality applied to \eqref{lb:duphi_pointwise}, and the facts that $e^{(2K_+ - a)v_Y(u)} \sim u^{-B}$ and $e^{-(c_p(s) - \tau)v_{\gamma}(u)} \sim u^{\frac{c_p(s) - \tau}{2K_+}(1 + \beta)}$:
\begin{align} 
    \|\partial_u \phi(\cdot, v)\|_{L^p_u([u_{\gamma}(v), u])} &\lesssim \| e^{-a v_Y(\cdot)} \Omega^2(\cdot, v_Y(\cdot)) \partial_u \phi(\cdot, v_0) \|_{L^p_u} + C(M) \| (\cdot)^{-1 + \frac{K_{-}}{K_+}\beta - \alpha} e^{-(c_p(s) - \tau)v_{\gamma}(\cdot)} \|_{L^p_u} \nonumber  \\
    &\le  u_{\gamma}(v)^{-B} \|\partial_u \phi(\cdot, v_0)\|_{L^p_u([0, u])} + C(M) \br{ \int_{u_{\gamma}(v)}^u x^{-p \tilde{A}} dx }^{\frac{1}{p}} \nonumber \\
    &\lesssim o_u(1) u_{\gamma}(v)^{-B} + C(M) u_{\gamma}(v)^{-A}, \label{blueshift:duphi_Lp}
\end{align}
where we emphasize that $A, B > 0$. The above upper bound is not optimal when $u_{\gamma}(v) \sim u$ since we neglected some terms depending on $u$ to obtain a simpler expression. On the other hand,  this $L^p$ bound encodes the dominant behaviour for $u$ fixed and large values of $v$, that is the case we will focus on.

Moreover, as a consequence of the decomposition $[u_{\gamma}(v), u] = ( [u_{\gamma}(v), u]$ $ \cap $ $\{|\partial_u \phi| \le 1\})$ $\cup$ $ ([u_{\gamma}(v), u]$ $ \cap$ $ \{|\partial_u \phi| > 1\}) $, the following $L^1$ bound follows:
\begin{equation} 
    \|\partial_u \phi(\cdot, v)\|_{L^1_u([u_{\gamma}(v), u])} \le u - u_{\gamma}(v) +  \|\partial_u \phi(\cdot, v)\|^2_{L^2_u([u_{\gamma}(v), u])}.
\end{equation}
This can be cast in a more convenient form by applying \eqref{blueshift:duphi_Lp} for $p=2$ and the relation $u_{\gamma}(v) \sim e^{\frac{2K_+ - a}{1+\beta}v}$. In particular:
\begin{equation} \label{blueshift:duphi_L1}
    \|\partial_u \phi(\cdot, v)\|_{L^1_u([u_{\gamma}(v), u])} \le u - u_{\gamma}(v) +  o_u(1) e^{Bv} + C(M)e^{Av},
\end{equation}
where we used the same symbols $A$ and $B$ to denote rescalings (in terms of the initial data) of the constants in \eqref{def_A} and $\eqref{def_B}$. We stress that, due to \eqref{def_cps}, we have
\begin{equation} \label{ABsmall}
    A =O(\beta) \quad \text{ and } \quad   B = O(\eta),
\end{equation}
where $\eta$ is the small constant associated to the redshift region (see proposition \ref{prop:redshift}), whereas $\beta$ is the small constant  introduced in the early-blueshift region (see proposition \ref{prop:early_blueshift}).
 
We now drop the smoothness assumptions on $\partial_u \phi$: by a density argument,  \eqref{blueshift:duphi_Lp} and \eqref{blueshift:duphi_L1}  still hold in the rough setting.

To close the bootstrap inequality for $\partial_v \log \Omega^2$, we first notice that \eqref{earlyblueshift:K} and  \eqref{earlyblueshift:dvlog} give:
\begin{align*}
    |\partial_v \log \Omega^2(u_{\gamma}(v), v) + 2K_{-}| &\le |\partial_v \log \Omega^2(u_{\gamma}(v), v) - 2K(u_{\gamma}(v), v)| + |2K(u_{\gamma}(v), v) + 2K_{-}| \\
    &= o_{v_1}(1) + O(\varepsilon),
\end{align*}
as $v_1 \to +\infty$.
Moreover, note that \eqref{blueshift:B1} and \eqref{blueshift:duphi_L1} yield
\begin{align*}
\int_{u_{\gamma}(v)}^u |\partial_u \phi \partial_v \phi| (u', v) du' &\lesssim C(N_{\text{i.d.}}, M) e^{-(c_p(s) - \tau)v} \| \partial_u \phi(\cdot, v) \|_{L^1_u([0, u])} = o_{v_1}(1), 
\end{align*}
as $v_1 \to +\infty$.
Furthermore, if we apply \cite[lemma 4.15]{Rossetti} (that still holds in our low-regularity setting) with $D = K_{-}$:
\[
\int_{u_{\gamma}(v)}^u \Omega^2(u', v) du' \le C(N_{\text{i.d.}})(v - v_{\gamma}(u)) e^{-\frac{\beta}{1+\beta}K_{-}v}.
\]
We now  exploit \eqref{blueshift:B2}, \eqref{blueshift:B3} and the above bounds in the integral formulation of \eqref{wave_Omega} to obtain
\begin{align*}
o_{v_1}(1)  &= \partial_v \log \Omega^2 (u, v) - \partial_v \log \Omega^2(u_{\gamma}(v), v) \\
&=\partial_v \log \Omega^2(u, v) + 2K_{-} - 2K_{-} - \partial_v \log\Omega^2(u_{\gamma}(v), v) \\
&= \partial_v \log \Omega^2(u, v) + 2K_{-} + o_{v_1}(1) + O(\varepsilon),
\end{align*}
thus implying \eqref{blueshift:dvlog} for $v_1$ large.
By integrating the latter, we also get:
\begin{equation} \label{lb:omega_bound}
e^{-(2K_{-} + \varepsilon)(v - v_{\gamma}(u))} \le \frac{\Omega^2(u, v)}{\Omega^2(u, v_{\gamma}(u))} \le e^{-(2K_{-} - \varepsilon)(v - v_{\gamma}(u))},
\end{equation}
where $\varepsilon > 0$ is the (possibly rescaled) parameter introduced in the early-blueshift region, and we recall that  $\varepsilon$ can be chosen small compared to $\beta$ (see e.g.\ the proof of proposition \ref{prop:along_gamma}).

To close the bootstrap on $|\lambda|$: the integral formulation of \eqref{raych_v} and the bounds on $r$ give
\begin{equation} \label{blueshift:bounds_Omega}
\frac{|\lambda|}{\Omega^2}(u, v) \le \frac{|\lambda|}{\Omega^2}(u, v_{\gamma}(u)) + C(N_{\text{i.d.}}) \int_{v_{\gamma}(u)}^v \frac{(\partial_v \phi)^2}{\Omega^2}(u, v') dv'.
\end{equation}
Using \eqref{blueshift:B1} and the previous bound on $\Omega^2$:
\begin{align*}
|\lambda|(u, v) &\le |\lambda|(u, v_{\gamma}(u)) \frac{\Omega^2(u, v)}{\Omega^2(u, v_{\gamma}(u))} + C(N_{\text{i.d.}}, M) \frac{\Omega^2(u, v)}{\Omega^2(u, v_{\gamma}(u))} \times \\
&\times e^{-(2K_{-} + \varepsilon)v_{\gamma}(u)} \int_{v_{\gamma}(u)}^v e^{-(2 c_p(s) + 2\tau - 2K_{-} - \varepsilon) v'} dv'.
\end{align*}
Similarly to the smooth case, \eqref{blueshift:lambda} follows after distinguishing between the cases $c_p(s) < K_{-}$ and  $c_p(s) \ge K_{-}$, and using \eqref{blueshift:bounds_Omega} again.

In the last part of the proof, we close the bootstrap on $\partial_v \phi$. As a consequence of \eqref{appendix_rdvphi}, the bound on $\lambda$ (see \eqref{blueshift:lambda}) and \eqref{earlyblueshift:phi}, we obtain:
\begin{align}
    |r \partial_v \phi|(u, v) &\le |r \partial_v \phi|(u_{\gamma}(v), v) + \int_{u_{\gamma}(v)}^u \br{|\lambda \partial_u \phi|}(u', v) du' \nonumber \\
    &\le C e^{-(c_p(s)-\tau)v}+ \int_{u_{\gamma}(v)}^u\frac{\Omega^2(x, v)}{\Omega^2(x, v_{\gamma}(x)) }|\lambda|(x, v_{\gamma}(x))|\partial_u \phi|(x, v)dx + \label{row1}\\
    &+e^{-\min\{2c_p(s) - 2\tau, 2K_{-}-\tau\}v}\|\partial_u \phi(\cdot, v)\|_{L^1_u([u_{\gamma}(v), u])}.    \label{row2}
\end{align}
Let us analyse the second term in \eqref{row1} and let us denote it by $II$. To bound such a term, we first notice that  \eqref{bound_lambda_gamma} gives $|\lambda|_{|\gamma}\lesssim 1$. Moreover, for the sake of convenience, we define the constants
\begin{align}
    \tilde{Z} &\coloneqq \frac{K_{-}}{K_+}\br{1+\beta} - \frac{\varepsilon}{2K_+}(1 + \beta), \\
    Z &\coloneqq \tilde{Z} -1 + \frac{1}{p}, \\
    W &\coloneqq \frac{2K_{-} - c_p(s) - \varepsilon + \tau}{2K_+}(1 + \beta).
\end{align}
Then, \eqref{lb:omega_bound}, the fact that $e^{(2K_{-} - \varepsilon)v_{\gamma}(u)} \sim u^{-\tilde{Z}}$, $e^{-(2K_{-}-c_p(s)-\varepsilon + \tau)v} \sim u_{\gamma}(v)^W$,  H{\"o}lder's inequality and \eqref{blueshift:duphi_Lp} give:
\begin{align}
    II &\lesssim e^{-(c_p(s) - \tau)v} \int_{u_{\gamma}(v)}^v e^{(2K_{-}-\varepsilon)v_{\gamma}(x)} e^{-(2K_{-}  - c_p(s)-\varepsilon + \tau)v} |\partial_u \phi|(x, v) dx  \nonumber \\
    &\le e^{-(c_p(s) - \varepsilon)v} u_{\gamma}(v)^W \int_{u_{\gamma}(v)}^u x^{-\tilde{Z}}  |\partial_u \phi|(x, v) dx  \nonumber \\
    &\le e^{-(c_p(s) - \varepsilon)v} u_{\gamma}(v)^W \br{ \int_{u_{\gamma}(v)}^u x^{-\tilde{Z}\frac{p}{p-1}} dx}^{\frac{p-1}{p}} \|\partial_u \phi(\cdot, v)\|_{L^p_u([u_{\gamma}(v), u])} \nonumber \\
    &\lesssim e^{-(c_p(s) - \varepsilon)v} u_{\gamma}(v)^{W-Z} \br{u_{\gamma}(v)^{-B} o_u(1) + C(M) u_{\gamma}(v)^{-A}}, \label{II_last}
\end{align}
where we used that $1 - \tilde{Z} \frac{p}{p-1} < 0$ since $K_{-}/K_+ > 1$.
It is an algebraic fact that 
\[
W-Z-B = 1 - \frac{1}{p} - \frac{c_p(s)}{2K_+} + O(\beta) + O(\eta) > 0,
\]
since $2K_+(1 - 1/p) > c_p(s)$ (see \eqref{def_cps}), and due to the facts that $a$ is $\eta$-close\footnote{Notice, in particular, that the constant $C(\eta)$ in \eqref{def_cps} can be chosen so to have a positive sign for $W - Z - B$.} to $2K_+$ and that $\varepsilon= O(\beta)$, $\tau = O(\beta)$. 
Moreover, for small values of $\beta$ we have 
\[
W-Z-A = \frac{K_{-}}{K_+} \beta - \alpha = \frac{K_{-}}{K_+} \beta + O(\varepsilon \beta) > 0,
\]
since $\alpha = O(\beta \varepsilon)$ (see e.g.\ \cite[appendix C]{Rossetti}). So, the expression in \eqref{II_last} decays faster than $e^{-(c_p(s) - \tau)v}$ (since $\varepsilon$ can be chosen small compared to $\tau = O(\beta)$). As a by-product, term II is of sub-leading order for large values of $v$.

Finally,  \eqref{blueshift:duphi_L1} and \eqref{ABsmall} applied to the last term in 
\eqref{row2} allow to close the bootstrap inequality for $\partial_v \phi$, provided that $M$ and $v_1$ are large compared to the initial data.
\end{proof}

\subsection{\texorpdfstring{$C^0$ extensions}{C0 extensions}} \label{section:C0_extensions}
The estimates obtained in the late-blueshift region are sufficient to show that
\begin{equation} \label{Cauchy_rminus}
    \lim_{u \to 0} \lim_{v \to +\infty} r(u, v) = r_{-}
\end{equation}
and to construct continuous extensions beyond the Cauchy horizon in some coordinate system $(u, V)$ for $u \in (0, U]$, where we can parametrize the Cauchy horizon as $\mathcal{CH}^+ = \{ V = 1 \}$. In particular, the methods of, e.g.\ \cite{Dafermos_2003, Dafermos_2005_uniqueness, CGNS4, LukOh1, VdM1, mythesis} to construct continuous extensions can be applied to our low-regularity setting as well.

\section{\texorpdfstring{$H^1$ extensions when $p > 2$}{H1 extensions when p > 2}} \label{section:no_mass_inflation}
\begin{theorem}[Boundedness of the Hawking mass] \label{thm:no_mass_inflation}
Under the assumptions of section \ref{section:assumptions} (and even in the absence of a lower bound for $\partial_v \phi$), if both
\begin{equation}  \label{no_mass_inflation_assumptions}
    s > K_{-} \quad \text{ and } \quad \rho \coloneqq \frac{K_{-}}{K_+} < 2\br{1 - \frac{1}{p}},
\end{equation}
then there exists a positive constant $C$, depending uniquely on the initial data, such that 
\[
    |\varpi|(u, v) \le C, \quad \forall\, (u, v) \in J^+(\gamma),
\]
provided that the value of $U$ is sufficiently small.
\end{theorem}
\begin{remark}
    The second condition in \eqref{no_mass_inflation_assumptions} is never satisfied if $p=2$, due to \eqref{basic_relation}. Hence, $H^1$ extensions can be constructed for our $W^{1, p}$ solutions (see definition \ref{def_W1psolution}) only if $p>2$.

    Even though we do not attempt to prove it in this paper, we are confident that theorem \ref{thm:no_mass_inflation} can be similarly proved for the case of a (possibly massive) \textbf{charged} scalar field. Indeed, the bootstrap proof of the previous sections does not rely on monotonicity results and can be easily adapted to the charged case. 
    
    The instability results that we will present in section \ref{section:mass_inflation}, however, do not generalize to the case of a charged scalar field in a straightforward way.
\end{remark}
\begin{proof}[Proof of theorem \ref{thm:no_mass_inflation}]
    The proof is analogous to that of \cite[theorem 5.1]{Rossetti} (where we can set the scalar field mass to be zero and the black hole charge to be constant), where now \eqref{appendix_varpi_boundedness}, \eqref{appendix_lognumu_eqn}, \eqref{appendix_lambda_modified} are used in the place of the equations of the smooth formulation of the  PDE system. 
    
    In particular, by following the same proof and using the estimates of the present paper, we get the Hawking mass is bounded if\footnote{This is related to the convergence of the integral $\int_{v_Y(u)}^v \frac{(\partial_v \phi)^2}{|\lambda |}(u, v')dv'$, which can be estimated from above by the integral of the term $\exp((-2c_p(s) + 2K_{-} + C(\beta, \varepsilon))v)$.}
    \[
    -2c_p(s)+2K_{-}+C(\beta, \epsilon) < 0.
    \]
    This is true if
    \[
    s > K_{-} \quad \text{ and } \quad \rho < 2\br{1 - \frac{1}{p}}.
    \]
The above is a non-empty region only for $p>2$.
\end{proof}

\subsection{\texorpdfstring{Construction of $H^1$ extensions}{Construction of H1 extensions}}
As already seen  in the smooth case \cite{CGNS4, Rossetti}, the absence of mass inflation is a sufficient condition to construct metric extensions of $H^1_{\text{loc}}$ regularity beyond the Cauchy horizon.

We omit the proof of this result in the rough case, since it is a straightforward adaptation of the proofs of \cite[theorem 12.3]{CGNS4} and \cite[theorem 7.5]{Rossetti}: it is sufficient to exploit the coordinate systems $(u, \tilde v)$ and $(u, \mathring v)$, where
\[
    \frac{d \mathring v}{dv} = \Omega^2(U, v), \quad \text{ and } \quad \tilde{v}(v) \coloneqq r(U, v_{\mathcal{A}}(U))-r(U, v).
\]
These two charts are $C^1-$compatible even in the rough setting. The main bounds for the metric components and Christoffel symbols then follow from the integral formulation of the PDE system.

\section{\texorpdfstring{Mass inflation when $p \ge 2$}{Mass inflation when p >= 2}} \label{section:mass_inflation}

In the following, we exploit the proofs in \cite{Dafermos_2005_uniqueness} (case $\Lambda = 0$) and \cite{CGNS3} (case $\Lambda > 0$), where mass inflation results were proved in terms of a polynomial and an exponential, respectively, pointwise lower bound along the initial outgoing hypersurface. Differently from these papers, however, we also employ the mild blow-up encoded in the transversal derivatives along $\mathcal{H}^+$ to propagate the good signs of $\partial_u \phi$ and $\partial_v \phi$ through the non-trapped region $I^{-}(\mathcal{A})$. In particular, our construction of the initial data plays a key role in proving mass inflation for a larger set of black hole parameters, compared to the smooth case with $\Lambda > 0$ (see already remark \ref{remark:comparison}).

Our instability results hold under the validity of assumptions \eqref{rough_data_general}--\eqref{price_law_lower}, i.e.\ they hold for the (infinitely many)  non-smooth initial data that we consider. Although this does not imply that all $W^{1,p}$ solutions have a divergent mass at $\mathcal{CH}^+$, still it does entail that the smooth initial data giving $H^1$ extensions are non-generic in the ``positive co-dimension'' sense described in section \ref{section:intro}.

From now on we assume that \eqref{chosen_initial_data2} and \eqref{rough_data_general_v} hold. In particular, for every $\varepsilon > 0$ and for every $u$ small, we have:
\begin{equation} \label{pointwise_duphi}
    u^{-\frac{1}{p}+\varepsilon} \lesssim |\partial_u \phi|_{|\underline{C}_{v_0}}(u) \lesssim u^{-\frac{1}{p}-\varepsilon}.
\end{equation}
We also recall that 
\begin{equation}
    \partial_v \phi_{|\mathcal{H}^+}(v) = (\partial_v \phi)_0(v) +\omega f_2(v), \quad \forall \, v \in [v_0, +\infty),
\end{equation}
that $(\partial_v \phi)_0 \in C^0([v_0, +\infty))$ satisfies \eqref{price_law_upper} and that $f_2 > 0$,  $f_2  \in C^0([v_0, +\infty))$ satisfies \eqref{price_law_lower}. Moreover, in the large $v$ limit, the leading order contribution to $\partial_v \phi_{|\mathcal{H}^+}$ is given by $f_2$:
\begin{equation} \label{asympt_H}
    \limsup_{v \to +\infty} \frac{(\partial_v \phi)_0(v)}{ f_2(v)} = 0.
\end{equation}

\begin{remark}[On the signs of the initial data] \label{remark:signs_id}
    Under our assumptions, for every $\omega > 0$ there exists $U, V > 0 $ such that
    \[
    (\partial_v \phi)_{|\mathcal{H}^+}(v) > 0 \quad \text{ and } \quad
    (\partial_u \phi)_{|\underline{C}_{v_0}}(u) > 0,
    \]
     for every $v > V$ and for every $u < U$, respectively.
    Similarly, for every $\omega < 0$ there exists $U, V > 0$ such that:
    \[
    (\partial_v \phi)_{|\mathcal{H}^+}(v) < 0 \quad \text{ and } \quad
    (\partial_u \phi)_{|\underline{C}_{v_0}}(u) < 0,
    \]
     for every $v > V$ and for every $u < U$, respectively.
     Notice that we can take $V = v_0$ by a suitable choice of $f_2$.
    Furthermore, without loss of generality,\footnote{As stressed in \cite{Dafermos_2005_uniqueness} and \cite{CGNS4}, the system \eqref{wave_r}--\eqref{raych_v} is symmetric under the substitution $(\partial_u \phi, \partial_v \phi) \to (-\partial_u \phi, -\partial_v \phi)$.} \textbf{we only deal with the case $\bm{\omega > 0}$}.
\end{remark}

\begin{lemma} \label{lemma:positive_trapped}
    Let us assume that $\omega > 0$ in \eqref{rough_data_general} and \eqref{rough_data_general_v}. Then, under our assumptions, for every $U > 0$ small we have 
    \[
        \partial_v \phi (u, v) > 0 \quad \text{ and } \quad \partial_u \phi(u, v) > 0,
    \]
    when $(u, v) \in \mathcal{P} = [0, U] \times [v_0, +\infty)$. 
\end{lemma}
\begin{proof}

    First, we analyse the region $\mathcal{R} \cup \mathcal{A} = \{\lambda \ge 0\}$ by following the analysis of \cite{CGNS4}  in the smooth case. Recall the parametrization of the apparent horizon introduced in section \ref{section:AH}. Then, by \eqref{wave_r_alternative} and since $-(\nu \kappa)(u, v) \sim \Omega^2(u, v) \sim \Omega^2(0, v)$ and $K(u, v) = K_+ + O(\delta)$ for every $(u, v)$ in the redshift region:
    \[
        0 -\lambda(0, v) = \int_0^{u_{\mathcal{A}}(v)} \br{\nu \kappa \cdot 2K}(u', v)du' \sim -\Omega^2(0, v)u_{\mathcal{A}}(v),
    \]
    even in our low-regularity setting.
    Hence, by using \eqref{EH:lambda} and \eqref{EH:omega}:
    \begin{equation} \label{u_app_upper}
        u_{\mathcal{A}}(v) \sim \frac{\lambda(0, v)}{\Omega^2(0, v)} \lesssim e^{-(2K_+ + 2s)v}, \quad \forall\, v \ge v_1.
    \end{equation}
    A similar computation also gives:\footnote{Alternatively, the improved bound for $\lambda$ in the regular region can be derived from \eqref{R:monotonicity}.}
    \[
        \lambda(u, v) \lesssim \lambda(0, v) \lesssim e^{-2sv},
    \]
    (compare with estimate \eqref{redshift:lambda_bdd}, that holds in the  redshift region).
    Moreover, \eqref{appendix_rdvphi}, \eqref{R:monotonicity} and \eqref{redshift:duphi_L1} yield
    \[
    |\br{r \partial_v \phi}(u, v) - \br{r \partial_v \phi}(0, v)| \lesssim \int_0^u \lambda(0, v) |\partial_u \phi|(u', v)du' \lesssim \delta e^{-2sv} e^{-c_p(s)v}, \quad \forall \, (u, v) \in J^{-}(\mathcal{A}).
    \]
    Hence, by \eqref{price_law_lower},  \eqref{price_law_upper},  \eqref{asympt_H} and remark \ref{remark:signs_id}:
    \begin{equation} \label{dvphi_bounds_regular_region}
        e^{-l(s) v} \lesssim \partial_v \phi(0, v) - \delta e^{-(2s + c_p(s)) v} \lesssim \partial_v \phi(u, v) \lesssim \partial_v \phi (0, v) + \delta e^{-(2s + c_p(s)) v} \lesssim e^{-sv}, 
    \end{equation}
    for every $(u, v) \in J^{-}(\mathcal{A})$. In particular, $\partial_v \phi$ is positive in $J^{-}(\mathcal{A})$.
    By \eqref{appendix_rduphi} and remark \ref{remark:signs_id}, we have that $\partial_u \phi > 0$ in $J^{-}(\mathcal{A})$ as well.

    By continuity of $\partial_u \phi$ and $\partial_v \phi$ away from the event horizon, it follows that they are positive in a neighbourhood of the apparent horizon. Then, similarly to how it is seen in \cite{Dafermos_2005_uniqueness, CGNS4}, the good signs in \eqref{appendix_rdvphi} and \eqref{appendix_rduphi} imply that $\partial_v \phi, \partial_u \phi > 0$ in $J^+(\mathcal{A})$ as well.\footnote{If, by contradiction, $(u, v) \in I^+(\mathcal{A})$ is the first point at which either $\partial_u \phi(u, v) = 0$ or $\partial_v \phi(u, v) = 0$, then integrating \eqref{appendix_rdvphi} and \eqref{appendix_rduphi} gives a contradiction.}

\end{proof}

\begin{lemma}[Monotonicity of the Hawking mass] \label{lemma:monotonicity}
We have:
\[
\varpi(u_2, v_2) - \varpi(u_2, v_1) - \varpi(u_1, v_2) + \varpi(u_1, v_1) \ge 0, 
\]
for every $(u_1, v_1), (u_2, v_2) \in J^+(\mathcal{A})$ such that $(u_2, v_2) \in J^+(u_1, v_1)$.    
\end{lemma}
\begin{proof}
    The proof is analogous to those in \cite[section 12]{Dafermos_2005_uniqueness} or \cite[Corollary B.2]{CGNS3}. The main difference here is that now this result follows from the integral formulation of the system, i.e.\  \eqref{appendix_u_varpi_eqn} and \eqref{appendix_v_varpi_equation} are exploited.
\end{proof}

For the following steps, it is useful to recall the fundamental inequality
\begin{equation} \label{basic_relation}
    \rho \coloneqq \frac{K_{-}}{K_+} > 1,
\end{equation}
which holds for sub-extremal Reissner-Nordstr{\"o}m-de Sitter solutions \cite[Appendix A]{CGNS3}.

\begin{theorem}[Mass inflation] \label{thm:mass_inflation}
    Under the assumptions of section \ref{section:assumptions}, 
    if we additionally require that
     either
    \[
    \rho > 2\br{1 - \frac{1}{p}}
    \]
    or 
    \[
     l(s) < K_{-}
    \]
    holds, 
    where $e^{-l(s)} \lesssim \partial_v \phi_{|\mathcal{H}^+}(0, v) < e^{-sv}$, 
    then
    \[
    \lim_{v \to +\infty} \varpi(u, v) = +\infty,
    \]
    for every $u \in (0, U]$, 
    provided that $U > 0$ is sufficiently small compared to the initial data. In particular, for $p=2$, we have mass inflation for every parameter of the reference black hole and for every $s > 0$ in the exponential Price law bounds.
\end{theorem}

The proof is an adaptation of the mass inflation proofs of  \cite[Section 12]{Dafermos_2005_uniqueness}, \cite[theorem 3.1]{CGNS3} to the case of an integral formulation of \eqref{main_system}. In the rough case we consider, however, we exploit the integral formulation of system \eqref{main_system} and the specific construction of our non-smooth initial data. Moreover, a careful analysis of the exponential terms is required so to show that mass inflation can be proved for a larger set of black hole parameters, compared to the smooth case. We report the proof in appendix \ref{appendix:proof_mass_inflation}.

\begin{figure}
    \centering
    \includegraphics[width=0.4\linewidth]{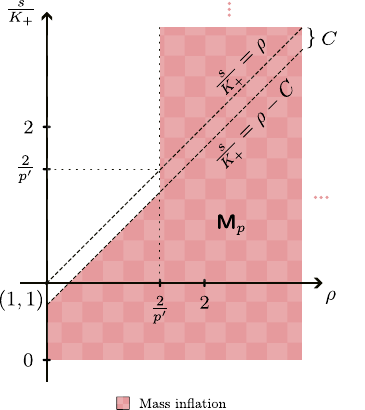}
    \caption{Region of parameters space corresponding to mass inflation, if $l(s) = s + C K_+$ in \eqref{price_law_lower} for some $C>0$ small. As $C \to 0$, a larger region of parameter space is covered. The case $p=p'=2$ is also covered here.}
    \label{fig:inflation_lowerbound}
\end{figure}

\begin{remark}[Comparison with the mass inflation results for the Einstein-Maxwell-scalar field system with smooth initial data] \label{remark:comparison}
    In the smooth case \cite{CGNS4} (where $\Lambda > 0$), mass inflation is proved under the additional constraint $\{ s < 2K_+\}$. The latter was used to show that $\partial_u \phi $ and $\partial_v \phi$ have the same sign along the apparent horizon $\mathcal{A}$. This statement is at the core of the mass inflation proof, since it is later used to exploit the good signs at the right hand side of \eqref{appendix_rdvphi} and \eqref{appendix_rduphi} in $J^+(\mathcal{A})$.

    In our non-smooth setting, the analogue of the above constraint is
    \[
        s < \frac{2K_+(1 - \frac{1}{p})}{1+ \frac{2}{p}},
    \]
    which gives the strict inequality $s < \frac{K_+}{2}$ in the case $p=2$ (see also fig.\ \ref{fig:main_results}). 

    The reason why this constraint can be avoided in the present paper (and mass inflation can be proved in a considerably larger region) is that we can exploit the mild blow-up encoded in our non-smooth initial data. For instance, due to our construction, we can easily determine the sign of the initial data along the initial hypersurfaces, as explained in remark \ref{remark:signs_id}.

    Finally, we stress that the situation is drastically different in the case $\Lambda = 0$.  Indeed, let us assume that the initial data for $\partial_u \phi$ and $\partial_v \phi$ are positive when considering \eqref{main_system} with $\Lambda = 0$. Further assume that, along the event horizon, the initial data satisfy polynomial lower bounds motivated by the asymptotic flatness of spacetime. Then, even for smooth data, these polynomial (rather than exponential) lower bounds   dominate over the exponentially small negative contributions arising from the redshift region. This allows to propagate the signs of $\partial_u \phi$ and $\partial_v \phi$ without imposing any constraint on the upper bound for the scalar field (see the case $\Lambda > 0$ \cite{CGNS4}) and without imposing lower regularity of the initial data (see the case of the present paper).
    
    In the case $\Lambda = 0$, mass inflation for the Einstein-Maxwell-real scalar field system has been proved conditionally to the validity of a pointwise lower bound for the scalar field along $\mathcal{H}^+$ \cite{Dafermos_2005_uniqueness} and, in a different work, conditionally to the validity of an upper bound\footnote{This instability result still relies on an integral (rather than pointwise) lower bound for $\partial_v \phi$. Such an integral lower bound was shown to hold for a generic class of initial data \cite{LukOh2}.}  on higher derivatives of the scalar field along $\mathcal{H}^+$ \cite{LukOhShlap}. In \cite{OnyxGautam}, it was recently proved that, for such a system, both the upper bounds on higher derivatives of $\partial_v \phi$ and the lower bound on $\partial_v \phi$ hold generically along $\mathcal{H}^+$. The latter was shown by exploiting the results in  \cite{LukOh24late}.
\end{remark}

\begin{remark}[$C^2$ inextendibility] \label{remark:C2}
    Mass inflation implies, in particular, the blow-up of the Kretschmann scalar. Being the latter a $C^2$ geometric quantity, in this case no spacetimes extensions of $C^2$ regularity can be constructed beyond the Cauchy horizon of the maximal past set that we constructed in section \ref{section:construction_interior}.

    Notably, there is  a subset of the parameter space $(s, \rho)$, for which $H^1$ extensions are allowed but $C^2$ extensions are not. This requires $p>2$ (so to have $H^1$ extensions, see section \ref{section:no_mass_inflation}) and the lower bounds in   \eqref{price_law_lower} and \eqref{pointwise_duphi} (to trigger the instability). In particular, if
    \begin{equation} \label{no_C2_extensions}
        l(s) < 2K_{-},
    \end{equation}
    then the Ricci scalar $R(u, v)$ blows up as $v \to +\infty$, for every $u \in (0, U]$. This follows, for instance, from a computation of the trace of the energy-momentum tensor in the late-blueshift region:
    \[
    |g^{\mu \nu}T_{\mu \nu}| \gtrsim \left | \frac{\partial_v \phi \partial_u \phi}{\Omega^2} \right | \gtrsim u^{-\frac{1}{p}+\varepsilon} e^{-l(s)v}e^{(2K_{-} + O(\varepsilon))(v - v_{\gamma}(u))}.
    \]
    It is worth noting that \eqref{no_C2_extensions} does not depend on $p$.
\end{remark}

\appendix
\section{Useful expressions} \label{appendix:useful_expressions}
The integrated equations stemming from definition \ref{def_integratedsolution} can often be put in a more convenient form.
To accomplish this, we exploit a basic formulation of the Leibniz rule under the low-regularity assumptions of definition \ref{def_W1psolution}. A similar version of the rule was already described in \cite{GajicLuk} (see equation (10.15) there).  In particular, the following is a consequence of the fundamental theorem of calculus for Lebesgue integrals \cite{Folland}.
\begin{lemma}[Integral Leibniz rule] \label{lemma:integral_leibniz_rule}
    Given  $A > 0$, let $h, g, F \colon [0, A] \times [0, A] \to \mathbb{R}$. Assume that  $g(u, \cdot) \in C^1_v$ for a.e.\ u,  $F \in    C^0_u L^1_v$, and that, for a.e.\ $u \in [0, A]$:
    \[
    h(u, v_2) - h(u, v_1)= \int_{v_1}^{v_2} F(u, v')dv', \quad \forall\, 0 \le v_1 < v_2 \le A.
    \]
    Then, $v \mapsto h(u, v)$ is differentiable in $v$ for a.e.\ $v$, for a.e.\ $u$. Moreover, for a.e.\ $u \in [0, A]$ and for every $0 \le v_1 < v_2 \le A$:
    \[
    \int_{v_1}^{v_2} F(u, v')g(u, v')dv' = h(u, v_2)g(u, v_2)-h(u, v_1)g(u, v_1)-\int_{v_1}^{v_2}h(u, v') \partial_v g(u, v')dv'.
    \]
   A similar statement holds if we invert the roles of the $u$ and $v$ coordinates.
    
\end{lemma}
Here we list the main integral equations of interest in this work.\footnote{See also \cite{Rossetti, CGNS4}, where the differential versions of these relations were derived using smoothness assumptions.} For every $0 \le u_1 < u_2$ and for every $v$:
\begin{align}
    (r \partial_v \phi)(u_2, v) - (r \partial_v \phi)(u_1, v) &= -\int_{u_1}^{u_2} (\lambda \partial_u \phi)(u', v)du', \label{appendix_rdvphi} \allowdisplaybreaks \\
    |\lambda|(u_2, v) &= |\lambda|(u_1, v) e^{\int_{u_1}^{u_2} \br{\frac{\nu}{1-\mu}\cdot 2K }(u', v)du'} \label{appendix_lambda_modified}, \text{ if } \lambda(\cdot, v) \ne 0 \text{ in } [u_1, u_2],  \\
    (r\lambda)(u_2, v) - (r\lambda)(u_1, v)&= \int_{u_1}^{u_2}\br{\frac{\Omega^2 e^2}{4 r^2} + \frac{\Omega^2 \Lambda r^2}{4} - \frac{\Omega^2}{4} }(u', v)du', \label{appendix_rlambda_eqn} \\
    \kappa(u_2, v) - \kappa(u_1, v) &= \int_{u_1}^{u_2} \frac{r \kappa (\partial_u \phi)^2}{\nu}(u', v) du', \label{appendix_kappa_eqn} \\
    \int_{u_1}^{u_2} \left| \frac{\partial_u \kappa}{\kappa} \right| (u', v)du' &= \int_{u_1}^{u_2}\frac{r (\partial_u \phi)^2}{|\nu|}(u', v)du', \label{appendix_kappa2_eqn} \\
     \varpi(u_2, v) -  \varpi(u_1, v) &=-2\int_{u_1}^{u_2}  \br{ \frac{r^2 (\partial_u \phi)^2}{\Omega^2} \lambda } (u', v)du',\label{appendix_u_varpi_eqn} \\ 
     \partial_v \log \Omega^2(u_2, v) - 2K(u_2, v) &= \partial_v \log \Omega^2(u_1, v) - 2K(u_1, v) + 2K(u_2, v)(\kappa(u_2, v) - 1)  \nonumber \\
     &+2K(u_1, v)\br{\kappa(u_1, v) - 1}  \nonumber \\
    &-\int_{u_1}^{u_2} \br{ 2K \partial_u \kappa + 2 \partial_u \phi \partial_v \phi + \frac{2 \kappa \partial_u \varpi}{r^2} }(u', v)du', \label{appendix_log_eqn}
\end{align}
Moreover, for every $0 \le v_1 < v_2$:
\begin{align}
    (r \partial_u \phi)(u, v_2) - (r \partial_u \phi)(u, v_1) &= -\int_{v_1}^{v_2} (\nu \partial_v \phi)(u, v')dv', \label{appendix_rduphi} \\
     (r\nu)(u, v_2) - (r\nu)(u, v_1)&= \int_{v_1}^{v_2}\br{\frac{\Omega^2 e^2}{4 r^2} + \frac{\Omega^2 \Lambda r^2}{4} - \frac{\Omega^2}{4} }(u, v')dv', \label{appendix_rnu_eqn} \\
    \nu(u, v_2) &= \nu(u, v_1)e^{\int_{v_1}^{v_2} \kappa \cdot (2K)(u, v')dv'}, \label{appendix_nu_modified}\\
    \varpi(u, v_2) - \varpi(u, v_1) &= \int_{v_1}^{v_2} \frac{r^2 (\partial_v \phi)^2}{2 \kappa}(u, v')dv', \label{appendix_v_varpi_equation} \\
    \log \frac{\nu}{1-\mu}(u, v_2) - \log \frac{\nu}{1-\mu}(u, v_1) &= \int_{v_1}^{v_2} \frac{r (\partial_v \phi)^2}{\lambda}(u, v')dv', \quad \text{ if } \lambda(u, v') \ne 0 \text{ for }v' \in [v_1,  v_2],\label{appendix_lognumu_eqn}
\end{align}
where  relations \eqref{appendix_rnu_eqn}--\eqref{appendix_lognumu_eqn} hold for every $u$, whereas  \eqref{appendix_rduphi}  holds for a.e.\ u.

All the above equations are obtained from the system of integrated equations coming from \eqref{wave_r}--\eqref{raych_v} after applying lemma \ref{lemma:integral_leibniz_rule} (see also remark \ref{remark:regularity}). Equation \eqref{appendix_rdvphi} is obtained from the wave equation for $\phi$ and by lemma \ref{lemma:integral_leibniz_rule} with $g = r$, and analogously for \eqref{appendix_rduphi}.  Equation \eqref{appendix_log_eqn} follows from the integral formulation of \eqref{wave_Omega}, the definitions of $\varpi$ and $K$, the fact that
\[
-\frac{2\kappa \partial_u \varpi}{r^2} + \kappa \partial_u (2K) = -\frac{\Omega^2 e^2}{r^4} + \frac{\Omega^2}{2r^2} + \frac{2 \nu \lambda}{r^2},
\]
and integration by parts.

Equations \eqref{appendix_lambda_modified} and \eqref{appendix_nu_modified} are obtained from the wave equation for $r$ and by applying the lemma with $g = (\lambda)^{-1}$ and $g = (-\nu)^{-1}= |\nu|^{-1}$, respectively, and using the definitions of $K$, $\varpi$ and the facts that $\lambda = \kappa(1-\mu)$ and $\Omega^2 = -4 \nu \kappa$. Equation \eqref{appendix_rlambda_eqn} is obtained similarly, where however $g = r$ now (an analogous reasoning gives \eqref{appendix_rnu_eqn}). Equation \eqref{appendix_kappa_eqn} follows from the Raychaudhuri equation in $u$ and by applying lemma \ref{lemma:integral_leibniz_rule} with $g=\kappa^2$. Equation \eqref{appendix_kappa2_eqn} follows from the Raychaudhuri equation in $u$, as well, but $g=\kappa$ is applied in lemma  \ref{lemma:integral_leibniz_rule}. Equation \eqref{appendix_u_varpi_eqn} is also obtained from the Raychaudhuri equation, but here the lemma is applied to $g = -\frac{r \lambda}{2}$. This actually yields additional terms that can be removed after differentiating
\[
\lambda(u_2, v) -\lambda(u_1, v) = - \int_{u_1}^{u_2} \frac{2 \nu\kappa}{r^2} \br{\frac{e^2}{r} + \frac{\Lambda}{3}r^3 - \varpi}(u', v)du',
\]
and using the definition of $\varpi$ again.
Similarly, equation \eqref{appendix_v_varpi_equation} follows from the Raychaudhuri equation in $v$ and by applying lemma \ref{lemma:integral_leibniz_rule} with $g = 2 r \nu$. This gives additional terms that can be removed using that $\frac{2r \nu \lambda}{\Omega^2} = \varpi - \mathcal{E}$, with $\mathcal{E} = \frac{e^2}{2r} + \frac{r}{2} - \frac{\Lambda}{6}r^3$. Moreover, $\mathcal{E} \in C^1_v$, $\nu \in C^1_v$ (the latter follows from the regularity of the terms at the right hand side of \eqref{wave_r} in its integral formulation) and $\partial_v \mathcal{E} = \frac{2\lambda}{\Omega^2}(\lambda \nu + r(\nu \kappa \cdot 2K))$ (this follows from definitions \eqref{def_varpi} and \eqref{def_K}). 
Equation \eqref{appendix_lognumu_eqn} follows from the Raychaudhuri equation in $v$, lemma \eqref{lemma:integral_leibniz_rule} with $g= -\frac{\Omega^2}{\lambda}$ and from the fact that
\[
\frac{\lambda}{\Omega^2} = -\frac{1-\mu}{4 \nu \kappa}.
\]
We further notice that, using definitions \eqref{def_varpi} and \eqref{def_K}, equation \eqref{wave_r} can be expressed in its integral formulation as
\begin{equation} \label{wave_r_alternative}
    \lambda(u_2, v) = \lambda(u_1, v) + \int_{u_1}^{u_2} \left [  \nu \kappa \br{2K}  \right ] (u', v)du', \quad \forall\, 0 \le u_1 < u_2 \text{ and for every } v.
\end{equation}
Furthermore, \eqref{appendix_rdvphi} and lemma \ref{lemma:integral_leibniz_rule} with $g = r \partial_v \phi$ imply that, for every $0 \le u_1 < u_2$ and for every $v$:
\begin{equation} \label{appendix_square_eqn}
    \br{r \partial_v \phi}^2(u_2, v) - \br{r \partial_v \phi}^2(u_1, v)  = -2 \int_{u_1}^{u_2} \br{r \lambda  \partial_u \phi \partial_v \phi}(u', v)du'.
\end{equation}
Finally, we also notice that, once we differentiate \eqref{appendix_v_varpi_equation} in $v$, the right hand side is continuous. By using that $\lambda = \kappa (1- \mu)$ and \eqref{def_mu}, we obtain
\begin{align} 
    \varpi(u, v_2) &= \varpi(u, v_1) e^{- \int_{v_1}^{v_2} \frac{r (\partial_v \phi)^2}{\lambda}(u, v')dv' } \label{appendix_varpi_boundedness}\\
    &+ \int_{v_1}^{v_2} e^{-\int_{v'}^{v_2} \frac{r (\partial_v \phi)^2}{\lambda}(u, y)dy } \left [ \frac{r^2 (\partial_v \phi)^2}{2\lambda} \br{1 + \frac{e^2}{r^2} - \frac{\Lambda}{3}r^2 } \right ] (u, v')dv', \nonumber
\end{align}
for every $u$ and for every $0 \le v_1 < v_2$.

\section{\texorpdfstring{Proof of proposition \ref{prop:extension_criterion}}{Proof of proposition 2.9}} \label{appendix:extension}
\begin{proof}
We split the proof into two separate parts. First, we prove the uniform bound on $N(\mathcal{D})$ and, then, construct an explicit extension. \newline
\underline{\textbf{1. Bounded area-radius $\Longrightarrow$ finite norms.}} \newline
\textbf{Preamble}: in the following, we will use the fact that $\nu < 0$ in $\mathcal{D}$ (see remark \ref{remark:signs}) multiple times, and we will denote by $C$ any positive constant depending on one (or more) value(s) in $\{\varepsilon, L, R, N_{\text{i.d.}}\}$. \newline
\textbf{Bound on $\bm{\kappa}$}: it follows from \eqref{appendix_kappa2_eqn} (notice that $\partial_u \kappa < 0$, see e.g.\ \eqref{appendix_kappa_eqn}), that $0 \le \kappa \le 1$. \newline
\textbf{Bounded spacetime volume}: for every $(u, v) \in \mathcal{D}$, we  use \eqref{def_kappa} and the previous estimate on $\kappa$  to conclude: 
\begin{equation} \label{boundTuv}
0 \le \int_0^v \int_0^u \Omega^2(u', v') du' dv' = -4 \int_0^v \int_0^u (\nu \kappa)(u', v') du' dv' \le C.
\end{equation} 
\textbf{Preliminary integral bound on $\bm{|\lambda|}$}:
equation \eqref{appendix_rlambda_eqn}, the bounds on $\kappa$ and the bounds on $r$ give
\[
L \cdot  \sup_{u \in [0, \varepsilon)} |\lambda|(u, v) \le C\br{1 + \int_0^{\varepsilon}  \Omega^2(u', v)du'}, \quad \forall\, v \in [0, \varepsilon).
\]
A further integration in the $v$ variable, together with \eqref{boundTuv}, gives
\begin{equation} \label{preliminary_bound_lambda_extension}
\int_0^v \sup_{u \in [0, \varepsilon)} |\lambda|(u, v')dv' \le C, \quad \forall\, v \in [0, \varepsilon).
\end{equation}
\textbf{Uniform bound on $\bm{|\nu|}$}:
as a consequence of \eqref{appendix_nu_modified}, \eqref{assumption_nu}, \eqref{def_K} and \eqref{def_varpi}, we get:
\begin{equation} \label{bound_nu_extension}
|\nu|(u, v) = \exp \br{-\int_0^v \left[ \kappa \br{\frac{e^2}{r^3} - \frac{1}{r} + \Lambda r  + \frac{\lambda}{\kappa r} } \right ] (u, v')dv'} \le C,
\end{equation}
in $\mathcal{D}$. Since the above integral is absolutely integrable in $[0, \varepsilon)$, we can similarly obtain a uniform lower bound on $|\nu|$.\newline
\textbf{Preliminary bound on $\bm{r|\partial_v \phi|}$}:
 we observe that   \eqref{appendix_rdvphi} yields:
\begin{equation} \label{theta_preliminary}
r|\partial_v \phi|(u, v) \le \sup_{v \in [0, V')} |r\partial_v \phi|(0, v) + \int_0^u  |\lambda \partial_u \phi| (u', v) du',
\end{equation}
in $\mathcal{D}$. \newline
\textbf{$\bm{L^p_u}$ bound on $\bm{|r \partial_u \phi|}$}:
 expression \eqref{appendix_rduphi}, together with the uniform bounds on $r$ and $\nu$, gives, for every $v  \in [0, \varepsilon)$ and for a.e.\ $u \in [0, \varepsilon)$:
\begin{equation} \label{zeta_preliminary}
|r \partial_u \phi|(u, v) \le |r \partial_u \phi|(u, 0) + C \int_0^v |\partial_v \phi|(u, v') dv'.
\end{equation}
By plugging \eqref{theta_preliminary} in \eqref{zeta_preliminary} and using Fubini-Tonelli's theorem:
\[
|r \partial_u \phi|(u, v) \le |r \partial_u \phi|(u, 0) + C + C \int_0^u \int_0^v |\lambda \partial_u \phi | (u', v') dv' du',
\]
Moreover, for a.e.\ $u \in [0, \varepsilon)$ and for some fixed $0 < \gamma < \varepsilon$ (not to be confused with the curve $\gamma$ of section \ref{section:construction_interior}), the quantity
\[
\zeta_{\gamma}(u) \coloneqq  \max_{v \in [0, \varepsilon-\gamma]} |r\partial_u \phi|(u, v),
\]
is well-defined (see also remark \ref{remark:regularity}). Therefore, the previous lines give
\begin{equation} \label{zeta_preliminary_2}
\zeta_{\gamma}(u) \le  |r \partial_u \phi|(u, 0) + C \int_0^u \zeta_{\gamma}(u') \br{\int_0^{\varepsilon - \gamma} |\lambda|(u', v') dv'}du', \quad \text{ for a.e. }u, 
\end{equation}
where we emphasize that $ |r \partial_u \phi|(u, 0)$ is integrable in $u$.
We now use \eqref{preliminary_bound_lambda_extension} in \eqref{zeta_preliminary_2} and integrate to obtain
\[
\int_0^u \zeta_{\gamma}(u')du' \le C \br{ 1+ \int_0^u \int_0^{s} \zeta_{\gamma}(u') du' ds},
\]
By applying Gr{\"o}nwall's inequality:
\[
\int_0^u \zeta_{\gamma}(u') du' \le C,
\]
for every $\gamma > 0$ and for a possibly different value of $C$. Using the integrability of $\zeta_{\gamma}$ and taking the limit $\gamma \to 0$, we finally get an $L^1$ bound on $r \partial_u \phi$. Using the convexity of $|x| \mapsto |x|^p$, a similar procedure also yields an $L^p$ bound for $r \partial_u \phi$. \newline
\textbf{Uniform bound on $\bm{|\phi|}$}:
by using the fundamental theorem of calculus, we write $\phi(u, v) = \phi(0, v) + \int_0^u \partial_u \phi(u', v) du'$  and use the previous bound on $\partial_u \phi$, to have a bound on $\phi$. \newline
\textbf{Uniform (lower) bound on $\bm{\kappa}$}:
the estimates on $\partial_u \phi$ and the lower bound on $|\nu|$  can be applied to \eqref{appendix_kappa2_eqn} to bound $\kappa$ away from zero. Together with the bounds on $|\nu|$, this implies that $\Omega^2$ is uniformly bounded and bounded away from zero. \newline
\textbf{Remaining bounds}:
a uniform bound on $\lambda$ follows from \eqref{appendix_rlambda_eqn}, after noticing that every term at the right hand side of the equation is uniformly bounded.
The previous estimate on $\lambda$ can be applied to \eqref{theta_preliminary} to get a uniform bound on $\partial_v \phi$, since we have estimated each function appearing in the integrand terms of \eqref{theta_preliminary} in the $\sup_{\mathcal{D}}$-norm. Finally, by the integral formulation of \eqref{wave_Omega} and due to the uniform and integral bounds of the quantities at its right hand side, we have that $\partial_v \log \Omega^2$ is uniformly bounded and that $\partial_u \log \Omega^2$ is bounded in $ L^1_u$ norm. Finally, the $L^p_u$ bound for $\nu$ follows from its $L^{\infty}$ bound.
\newline
\underline{\textbf{2. Construction of the extension:}} 
 let $\delta, \gamma > 0$ be small and consider the initial value problem with initial data prescribed on 
\[
[0, \varepsilon-\gamma] \times \{\varepsilon-\delta\} \cup \{0\} \times  [\varepsilon-\delta, \varepsilon+\delta].
\]
\begin{figure}[H]
\centering
\includegraphics[width=0.5\textwidth]{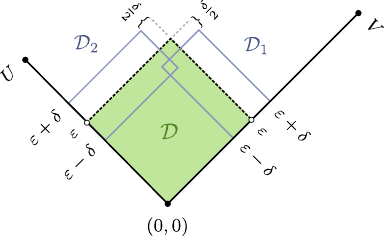}
\end{figure}
By the local existence result of proposition \ref{prop:local_existence}, and due to the uniform bound on $N(\mathcal{D})$ proved in the previous steps, there exists a minimum time of existence $\tilde \delta > 0$, whose size depends \textbf{only} on $N_{\text{i.d.}}$, such that we can construct a solution to our PDE system in $\tilde{\mathcal{D}}_1 \coloneqq [0, \tilde \delta] \times [\varepsilon-\delta, \varepsilon-\delta+\tilde \delta]$. Moreover, the norms of the solution in this domain are controlled by $N_{\text{i.d.}}$. Since we chose the minimum time of existence, this result holds whenever we translate the initial null segments along the $v$--direction, provided that the ingoing segment stays in $\mathcal{D}$ and that the outgoing segment is contained in $\{0\} \times [0, V]$. Moreover, we can extend this domain along the $u$-direction by iterating this procedure: since $\tilde \delta$ is the minimum time, it is possible to construct a solution in $[\tilde \delta, 2\tilde \delta] \times [\varepsilon-\delta, \varepsilon-\delta + \tilde \delta]$ and so on.

Without loss of generality, we therefore consider the new shifted rectangle to be $\mathcal{D}_1 \coloneqq \left[0, \varepsilon- \gamma \right] \times [\varepsilon-\delta, \varepsilon+\delta]$, where the value of $\delta$ is possibly different from the original one, and take $\gamma = \frac{\delta}{2}$.

The same construction can be followed with respect to the initial data prescribed on 
\[
[\varepsilon-\delta, \varepsilon+\delta] \times \{0\} \cup \{\varepsilon-\delta\} \times \left [0, \varepsilon-\frac{\delta}{2} \right].
\]
In particular, we define $\mathcal{D}_2 \coloneqq [\varepsilon-\delta, \varepsilon+\delta] \times [0, \varepsilon - \frac{\delta}{2}]$, for the same value of $\delta$ (if a smaller value of $\delta$ is required by the local existence theorem, we repeat the construction of $\mathcal{D}_1$ and $\mathcal{D}_2$ from scratch, starting from the ingoing null segment).

We now consider the point $p=\left(\varepsilon-\frac{\delta}{2}, \varepsilon-\frac{\delta}{2} \right)$. By performing a coordinate shift, we can set $p=(0, 0)$. We then construct a solution from the initial data prescribed on $\left[0, \frac32 \delta \right] \times \{0\} \cup \{0\} \times \left[0, \frac32 \delta \right]$. 
As a by-product of the local existence result, again, the  norms of the (continuous) solution in the compact sets $\mathcal{D}_1$ and $\mathcal{D}_2$ are controlled by $N_{\text{i.d.}}$. Therefore, for the initial value problem prescribed on these null segments, the minimum time of existence of solutions is again given by the previously obtained $\delta$.
Thus, we obtain a continuous solution in 
\[
\mathcal{D}_3 \coloneqq  \left [ 0, \frac32 \delta \right ] \times \left[0, \frac32 \delta \right],
\]
which is enough to conclude the proof.
\end{proof}

\section{\texorpdfstring{Proof of theorem \ref{thm:mass_inflation}}{Proof of theorem 5.4} } \label{appendix:proof_mass_inflation}
\begin{proof} 
    First, notice that, due to \eqref{appendix_v_varpi_equation}, $\varpi$ is monotonically increasing in the $v$ direction. Moreover, due to \eqref{appendix_u_varpi_eqn}, it is also increasing in the $u$ direction in $I^+(\mathcal{A})$. We can immediately exclude the following two scenarios:
    \begin{itemize}
        \item If $\lim_{v \to +\infty} \varpi(u, v) = +\infty$ for every $u \in (0, U]$, then the proof is completed. 
        \item  On the other hand, if there exists $u \in (0, U]$ such that  $\lim_{v \to +\infty} \varpi(u, v) < +\infty$ and if $\lim_{u \to 0} \varpi(u, +\infty)  > M$ (where $M$ is the mass of the reference black hole), then  the zigzag argument in \cite{Dafermos_2005_uniqueness} (see also \cite[theorem 3.1]{CGNS3}, where the argument was repeated in the $\Lambda > 0$ case) shows that a contradiction must occur.
    \end{itemize}

    Since, to the future of the apparent horizon, $\varpi$ is monotonically increasing in the $u$ direction and since  $\varpi(u_{\gamma}(v), v) \to M$ as $v \to +\infty$ (see  \eqref{earlyblueshift:varpi}), the only alternative left is the scenario in which $\lim_{u \to 0} \varpi(u, +\infty) = M$. In this case, define 
    \[
        I(u) \coloneqq -\int_{v_{\gamma}(u)}^{+\infty} \frac{r^2 (\partial_v \phi)^2}{\lambda}(u, v')dv',
    \]
    which is a positive quantity.
    Under our assumptions, the Hawking mass is close to $M$ for $u$ small, hence we have that $K \sim - K_{-} < 0$ in $J^+(\gamma)$ (see \eqref{def_K} and the bounds on $r$ in proposition \ref{prop:late_blueshift}). 
    Since, in $J^+(\gamma)$, we have that $\lambda < 0$ is increasing along the $u$ direction  (the proof of this result is analogous to that giving \eqref{R:monotonicity}, where now the fact that $K \sim -K_{-} < 0$ is exploited) and since $(r \partial_v \phi)^2$ is also increasing in the $u$ direction (due to the good signs at the right hand side of \eqref{appendix_square_eqn}), it follows that $u \mapsto I(u)$ is monotonically increasing. We stress that we exploited the information on the signs of $\partial_u \phi$ and $\partial_v \phi$, coming from lemma \ref{lemma:positive_trapped}. Let us now distinguish between two cases.

    If $I(u) \equiv +\infty$, then \eqref{appendix_lognumu_eqn} gives 
    \[
        \frac{\nu}{1-\mu}(u, +\infty) \equiv 0.
    \]
    Since $1-\mu$ is bounded (see \eqref{def_mu}), it follows that $\nu(u, +\infty) \equiv 0$, i.e.\ $r(u, +\infty) \equiv r_{-}$ (see \eqref{Cauchy_rminus}).
    
    On the other hand, lemma \ref{lemma:monotonicity} and \eqref{appendix_u_varpi_eqn} (whose  right hand side is bounded away from zero in a neighborhood of the curve $\gamma$ due to our choice of initial data and due to proposition \ref{prop:along_gamma}) give
    \[
    \varpi(u_2, v_2) - \varpi(u_1, v_2) \ge \varpi(u_2, v_1) - \varpi(u_1, v_1) > 0,
    \]
    for every $(u_1, v_1) \in J^+(\mathcal{A}) \cap J^{-}(\gamma)$ such that $(u_2, v_2) \in J^+(u_1, v_1)$.
    Hence:
    \[
    \varpi(u_2, +\infty) >  \varpi(u_1, +\infty), \quad \forall\, 0 < u_1 < u_2 \le U.
    \]
    In particular,  since we are assuming $\lim_{u \to 0 }\varpi(u, +\infty) =M$, we have 
    \begin{equation} \label{varpi_case_I_infty}
    \varpi(u, +\infty) > M, \quad \text{ for every } u \in (0, U].
    \end{equation}
    Now, the latter is incompatible with the fact that $r(u, +\infty) \equiv r_{-}$, indeed \eqref{varpi_case_I_infty} and \eqref{def_mu} imply
    \[
    \br{1-\mu}(u, +\infty) = \br{1- \mu}(r_{-}, \varpi(u, +\infty)) < \br{1- \mu}(r_{-}, M) = 0, \quad \forall\, u \in (0, U].
    \]
    But then, using that $\lambda = \kappa (1-\mu)$, we can use \eqref{appendix_v_varpi_equation} to write
    \[
    \varpi(u, v_2) \ge \varpi(u, v_1) - C(u) \int_{v_1}^{v_2} \frac{r^2 (\partial_v \phi)^2}{\lambda}(u, v')dv', 
    \]
    for some $C=C(u)$ and for $(u_1, v_1)$, $(u_2, v_2)$ as above.
    By taking the limit $v_2 \to +\infty$, we find that $I(u) < +\infty$ for some small value of $u$, which is a contradiction.

    Finally, let us consider the case $I(u) < +\infty$ for some $u \in (0, U]$. Similar to the previous reasoning for the Hawking mass, the zigzag argument of \cite{Dafermos_2005_uniqueness} gives
    \[
    I(u) \to 0, \text{ as } u \to 0.
    \]
    The above and \eqref{appendix_v_varpi_equation}, together with the facts that $\lambda = \kappa(1-\mu)$ and that $|1-\mu|$ is bounded, imply that:
    \begin{equation} \label{mass_close_to_M}
    \varpi(u, v) = \varpi(u, v_{\gamma}(u)) -\frac12 \int_{v_{\gamma}(u)}^v \br{ \frac{r^2 (\partial_v \phi)^2}{\lambda}(1-\mu)}(u, v')dv' = M + o_u(1),
    \end{equation}
    as $u \to 0$.
    Therefore, by \eqref{def_K} and since $r \sim r_{-}$ in the late-blueshift region, we have:
    \begin{equation} \label{K_close_mKm}
        K(u, v) = -K_{-} + o_u(1),
    \end{equation}
    for every $(u, v) \in J^+(\mathcal{\gamma})$.
    
    In this scenario, we are able to improve the bounds for $\lambda$ in the late-blueshift region. In fact, by exploiting the fact that $I(u) < + \infty$ for every $u$,  \eqref{appendix_lognumu_eqn} gives
    \[
    \frac{\nu}{1-\mu}(u, v_2) \gtrsim \frac{\nu}{1-\mu}(u, v_1), \quad \forall\,  (u, v_1) \in J^+(\gamma), \text{ with } (u, v_2) \in J^+(u, v_1).
    \]
    So, using that $1-\mu$ is bounded and bounded away from zero along $\Gamma_Y$ (see proposition \ref{prop:noshift}):  
    \[
    \int_{u_{\gamma}(v)}^u |\nu|(\tilde u, v_Y(\tilde u))d \tilde u \lesssim \int_{u_{\gamma}(v)}^u \frac{\nu}{1-\mu}(\tilde u, v_Y( \tilde u ))d \tilde u \lesssim \int_{u_{\gamma}(v)}^u \frac{\nu}{1-\mu}(\tilde u, v)d \tilde u,
    \]
    for every $(u, v) \in J^+(\gamma)$.
    This chain of inequalities can be further extended by exploiting a simple change of variables and the fact that $\kappa \sim 1$ in the no-shift region (see also the analogous steps in the proof of proposition \ref{prop:along_gamma}):
    \begin{align*}
        \int_{u_{\gamma}(v)}^u \frac{\nu}{1-\mu}(\tilde u, v)d \tilde u &\gtrsim \int_{u_{\gamma}(v)}^u |\nu |(\tilde u, v_Y(\tilde u)) d \tilde u \\
        &= \int_{v_Y(u)}^{v_Y(u_{\gamma}(v))} |\lambda|(u_Y( \tilde v), \tilde v) d \tilde v \\
        &\sim \int_{\frac{v_{\gamma}(u)}{1+\beta}}^{\frac{v}{1+\beta}} \kappa(u_Y(\tilde v), \tilde v)d \tilde v \ge \frac{v - v_{\gamma}(u)}{1+\beta}.
    \end{align*}
    Thus, \eqref{appendix_lambda_modified},  \eqref{K_close_mKm} and \eqref{bound_lambda_gamma} yield:
    \begin{align}
    -\lambda(u, v) &=-\lambda(u_{\gamma}(v), v) e^{-(2K_{-}+O( \beta))\int_{u_{\gamma}(v)}^u \frac{\nu}{1-\mu}(u', v)du' } \nonumber \\
    &\le -\lambda(u_{\gamma}(v), v)e^{-(2K_{-}+O(\beta))v} \nonumber \\
    &\le e^{-(2K_{-}+O(\beta))v}, \quad \forall\, (u, v) \in J^+(\gamma). \label{UB_lambda}
    \end{align}
    A similar procedure (for more details see also the proof of \cite[theorem 5.1]{Rossetti}) yields a lower bound on  $|\lambda|$ in $J^+(\gamma)$:
    \begin{equation} \label{LB_lambda}
    |\lambda|(u, v) \gtrsim e^{-(2K_{-}+O(\beta))(v - v_{\gamma}(u))}.
    \end{equation}

    Then, using the above bounds and \eqref{appendix_rdvphi}, together with the fact that $\partial_v \phi > 0$ in $\mathcal{P}$ (see lemma \ref{lemma:positive_trapped}), that $r\partial_u \phi$ is increasing in the $v$ direction (see \eqref{appendix_rduphi}), \eqref{pointwise_duphi} and that $u_{\gamma}(v) \sim e^{-2K_+\frac{v}{1+\beta}}$: 
    \begin{align}
        (r \partial_v \phi)(u, v) &= \br{r \partial_v \phi}(u_{\gamma}(v), v) - \int_{u_{\gamma}(v)}^{u} \br{ \lambda \partial_u \phi}(u', v) du'  \nonumber \\
        &\gtrsim  e^{-(2K_{-}+O(\beta))v} \int_{u_{\gamma}(v)}^{u} e^{(2K_{-}+O(\beta)) v_{\gamma}(u')} \partial_u \phi(u', v) du' \nonumber \\
        &\gtrsim e^{-(2K_{-}+O(\beta)) v} \int_{u_{\gamma}(v)}^{u}  e^{(2K_{-}+O(\beta)) v_{\gamma}(u')}
 \partial_u \phi(u', v_0) du' \nonumber \\
        & \gtrsim e^{-(2K_{-}+O(\beta))v} \int_{u_{\gamma}(v)}^{u}  x^{-\rho-\frac{1}{p}+O(\beta)}dx  \nonumber  \\
        &\ge C(p, \rho)
        e^{-(2K_{-}+O(\beta))v} \br{u_{\gamma}(v)^{1-\frac{1}{p}-\rho} - u^{1-\frac{1}{p}-\rho}}\nonumber \\
        &\gtrsim C(p, \rho)\br{ e^{-\br{\frac{2K_+}{p'} + O(\beta)}v} - u^{1-\frac{1}{p}-\rho} e^{-(2K_{-}+O(\beta))v} }, \label{LB_dvphi} 
    \end{align}
   for every $(u, v) \in J^+(\gamma)$, where $C(p, \rho) = (\rho + \frac{1}{p}-1)^{-1} > 0$, $p'$ is the conjugate exponent to $p$ and we used the sub-extremality condition $\rho = \frac{K_{-}}{K_+}> 1$.
    Hence, \eqref{UB_lambda} and \eqref{LB_dvphi} give
    \begin{equation} \label{blow_up_I}
        I(u) = -\int_{v_{\gamma}(u)}^{+\infty} \frac{(r \partial_v \phi)^2}{\lambda}(u, v')dv' \gtrsim \int_{v_{\gamma}(u)}^{+\infty} e^{\br{-\frac{4K_+}{p'} +2K_{-} + O(\beta)}v'}dv'.
    \end{equation}
    We emphasize that $\beta$ can be chosen sufficiently small compared to the initial data. If 
    \[
    \rho > \frac{2}{p'} = 2 \br{1 - \frac{1}{p}},
    \]
    then the last term in \eqref{blow_up_I} blows up for every $u \in (0, U]$ and we obtain a contradiction. In particular, in such a range of black hole parameters, mass inflation occurs. If $p=2$, then $\frac{1}{p'} = 1 - \frac{1}{p} = \frac12$ and thus mass inflation arises for all parameters of the reference black hole, in light of \eqref{basic_relation}.

    We now show that there is an alternative way to prove mass inflation, by exploiting the lower bound on $\partial_v \phi_{|\mathcal{H}^+}$ (rather than the lower bound on $\partial_u \phi$). When $p>2$, the (near-extremal) region $\mathcal{NE} \coloneqq \{1 < \rho < \frac{2}{p'} \}$ in the  space of parameters of the reference black hole is non--empty. However, the previous bounds do not cover this region.    
    In the next steps, we prove that mass inflation also occurs in a subregion of $\mathcal{NE}$ and that the size of such a subregion depends on how strict the lower bound on $\partial_v \phi_{|\mathcal{H}^+}$ is (see  fig.\ \ref{fig:inflation_lowerbound}).
    In fact, we first notice that  \eqref{dvphi_bounds_regular_region} gives:
    \[
        \br{r \partial_v \phi}(u, v) \gtrsim e^{-l(s)v}, \quad \forall\, (u, v) \in J^{-}(\mathcal{A}).
    \]
    Moreover, using the good signs of \eqref{appendix_rdvphi}:
    \[
    \partial_v \phi(u, v) \ge \partial_v \phi(u_{\mathcal{A}}(v), v) \gtrsim e^{-l(s)v}, \quad \forall\, (u, v) \in J^+(\mathcal{A}).
    \]
    Hence, if we use the latter and \eqref{UB_lambda} in the definition of $I(u)$:
    \[
        I(u) =- \int_{v_{\gamma}(u)}^{+\infty} \frac{(r \partial_v \phi)^2}{\lambda}(u, v')dv' \gtrsim \int_{v_{\gamma}(u)}^{+\infty} e^{\br{-2l(s) +2K_{-} + O(\beta)}v'}dv', \quad \forall\, u \in (0, U],
    \]
    where the latter diverges if
    \[
    l(s) < K_{-}. 
    \]     
\end{proof}
\bibliography{bibliography} 
\end{document}